\documentclass[11pt, letterpaper]{amsart}

\usepackage{color}
\usepackage[english]{babel}
\usepackage[latin1]{inputenc}
\usepackage{times}
\usepackage[T1]{fontenc}
\usepackage{amsfonts}
\usepackage{epsfig,psfrag,latexsym}
\usepackage{amsmath, amssymb}
\usepackage{graphicx}
\usepackage{amsthm}
\usepackage{enumerate}

\newtheorem{theorem}{Theorem}[section]

\newtheorem{lemma}[theorem]{Lemma}
\newtheorem{proposition}[theorem]{Proposition}
\theoremstyle{definition}

\newtheorem{assumption}[theorem]{Assumption}

\theoremstyle{remark}
\newtheorem{remark}[theorem]{Remark}

\numberwithin{equation}{section}

\newcommand{\reals}{\mathbb R}
\newcommand{\nats}{\mathbb N}

\newcommand{\eps}{\varepsilon}

\newcommand{\such}{\ | \ }

\newcommand{\We}{\mathcal{W}}

\newcommand{\hwe}{\hat{\We}}
\newcommand{\hz}{\hat{Z}}
\newcommand{\hqprob}{\hat{\qprob}}
\newcommand{\hpi}{\hat{\pi}}
\newcommand{\prob}{\mathbb{P}}
\newcommand{\qprob}{\mathbb{Q}}

\newcommand{\esp}{\mathbb{E}}
\newcommand{\espalt}[2]{\esp^{#1}\bra{#2}}

\newcommand{\F}{\mathcal{F}}

\newcommand{\G}{\mathcal{G}}
\newcommand{\M}{\mathcal{M}}
\newcommand{\tM}{\widetilde{\mathcal{M}}}

\newcommand{\filt}{\mathbb{F}}
\newcommand{\filtg}{\mathbb{G}}

\newcommand{\relent}[2]{H\left(#1\such #2\right)}

\newcommand{\condespalt}[3]{\esp^{#1}\bra{#2\big| #3}}
\newcommand{\condprobalt}[3]{\prob^{#1}\left[#2\ \big|\ #3\right]}
\newcommand{\EN}{\mathcal{E}}

\newcommand{\nada}[1]{}

\newcommand{\dfn}{\, := \,}

\newcommand{\bra}[1]{\left[#1\right]}
\newcommand{\cbra}[1]{\left\{#1\right\}}
\newcommand{\dbra}[1]{[\kern-0.15em[ #1 ]\kern-0.15em]}
\newcommand{\dbraco}[1]{[\kern-0.15em[ #1 [\kern-0.15em[}

\newcommand{\ul}[1]{\underline{#1}}
\newcommand{\ol}[1]{\overline{#1}}

\newcommand{\ito}{It\^{o}}

\newcommand{\tr}{\mathsf{Tr}}

\title[Pricing with Default]{Optimal Investment and Pricing in the Presence of Defaults}

\author{Tetsuya Ishikawa}
\address{Morgan Stanley}
\email{testsuya.ishikawa@gmail.com}

\author{Scott Robertson}
\address{Questrom School of Business\\
Boston University\\
Boston, MA 02215}
\email{scottrob@bu.edu}
\thanks{S. Robertson is supported in part by the National Science Foundation
  under grant number  DMS-1613159.}

\date{\today}

\begin{document}

\begin{abstract}
We consider the optimal investment problem when the traded asset may default, causing a jump in its price.  For an investor with constant absolute risk aversion, we compute indifference prices for defaultable bonds, as well as a price for dynamic protection against default.  For the latter problem, our work complements \cite{MR2359373}, where it is implicitly assumed the investor is protected against default. We consider a factor model where the asset's instantaneous return, variance, correlation and default intensity are driven by a time-homogenous diffusion $X$ taking values in an arbitrary region $E\subseteq\reals^d$.  We identify the certainty equivalent with a semi-linear degenerate elliptic partial differential equation with quadratic growth in both function and gradient.  Under a minimal integrability assumption on the market price of risk, we show the certainty equivalent is a classical solution.  In particular, our results cover when $X$ is a one-dimensional affine diffusion and when returns, variances and default intensities are also affine.  Numerical examples highlight the relationship between the factor process and both the indifference price and default insurance.  Lastly, we show the insurance protection price is not the default intensity under the dual optimal measure.
\end{abstract}

\maketitle

\section*{Introduction}\label{S:intro}  The goal of this paper is to solve the optimal investment problem when the underlying asset may default, causing a jump to zero in its price.  Our primary applications are the explicit pricing of defaultable bonds and dynamic default insurance protection; taking into account investor preferences, market incompleteness, and crucially, the ability of the investor to trade in the underlying prior to default.

The issue of default, or more generally, contagion risk, is of primary importance for risk managers during times of financial crisis.  However, the precise relationship between optimal policies, defaultable bond prices, default insurance, and the factors which drive the larger economy is not well known if the investor may trade in the underlying asset.  The current practice of "risk-neutral" pricing for defaultable bonds, while suffering from the well-known  risk neutral measure selection issue, also suffers from a more significant problem in that it is implicity assumed there is separation between the defaulting entity and the traded assets.  In practice, this is not always the case, and the hedging arguments which form the basis for risk neutral pricing are called into question.

To address the above issues, we consider the utility maximization problem for an investor with constant absolute risk aversion, who may invest in a defaultable asset, and who additionally owns a non-traded claim with payoff contingent upon the survival of the asset. As the default time is not predictable for the investor, the market is incomplete.  We assume the "hybrid" model structure of \cite{MR2212266,MR2359373}, where asset dynamics and default intensities are driven by an underlying time homogenous diffusion, representing a set of economic factors.  Here, the certainty equivalent is identified with a degenerate elliptic semi-linear partial differential equation (PDE). Using powerful PDE existence results from \cite{MR1465184, MR0181836}, in conjunction with well-known duality results for exponential utility  (see, for example, \cite{MR1891731, MR1891730, MR2489605}), we show the certainty equivalent is a classical solution to the PDE under minimal assumptions on the model coefficients.  In particular, our results can handle when the factor process takes values in an arbitrary region in $\reals^d$, and when the asset return, volatility and default intensity are unbounded functions of the factor process.

The utility maximization problem taking into account default has been widely studied. Through a complete literature review is too lengthy to give here, we wish to highlight where our work fits in regards to prior studies.  First and foremost, our work fills a gap in the literature by considering the Markovian setting where one may solve the optimal investment problem using PDE techniques. In the PDE setting it is possible to obtain explicit solutions which highlight the dependency of optimal policies and bond prices upon broader economic factors. Though the Markovian case has been treated in some form dating at least back to \cite{MR2359373,lukas2001pricing}, to the best of our knowledge the problem has not been studied when one assumes the investor loses her dollar position in the stock upon default.  Indeed, in contrast to \cite{MR2359373}, we do not assume the investor is fully protected against losses due to default.

The non-Markovian case has been much more thoroughly studied. For dynamics governed by Brownian adapted processes, the problem traces to \cite{MR2489997, MR2827466,MR3378038} and the subsequent extensions in \cite{MR2863640, MR3059266,MR3412158,MR3310349}.  In fact, our setting is closest to \cite[Sections 2,3]{MR2827466} where there is a single risky asset with pre-default dynamics governed by a univariate Brownian motion. Here, under general conditions, the value function is identified as a maximal sub-solution of a backward stochastic differential equation (BSDE), which under certain bounded-ness and/or compactness conditions becomes a true solution.  In fact, as shown in the above extensions, the BSDE setting allows for model generalizations such as multiple defaults, or even credit related events, where "default" refers to any event which causes a jump in the asset's value. However, explicit computations for optimal policies and/or prices are often absent (see \cite[Section 4.3]{MR2863640} for a notable exception). Thus, while on the one hand, by working in a diffusive environment for a single stock, our assumptions are more restrictive than what can be handled using BSDEs, on the other hand, by relying upon PDE techniques we are able to significantly relax the integrability assumptions made upon the model coefficients and provide explicit solutions.  This enables our results to apply to a wide range of models used in industry: in particular our results cover the extended affine models of \cite{dai2000saa,MR1994043,filipovic2001general} (see Section \ref{S:examples}).

Briefly, we now explain our model, assumptions, and main results.  There is a diffusion $X$ satisfying $dX_t = b(X_t)dt + a(X_t)dW_t$, which represent the underlying economic factors.  We allow $X$ to take values in a general region $E\subseteq\reals^d$, and only require the diffusion matrix $A = aa'$ to be locally elliptic on $E$. The riskless asset is normalized to $1$, and a there is a single risky asset $S$ which, prior to default, has instantaneous return $\mu$, variance $\sigma^2$, and correlation $\rho$ dependent upon $X$. Thus, even absent default, the market is incomplete if $\rho'\rho \not\equiv 1$\footnote{Throughout, "'" denotes transposition.}.  Given the factor process $X$, the default time $\delta$ has an intensity $\gamma$ driven by $X$.  The investor may trade in $S$ up to $\delta$, but at $\delta$ she will lose her position in the stock, and may no longer trade. In addition to trading in $S$, the investor owns a non-traded claim $\phi$, which is received contingent upon the survival of the stock. The investor has constant absolute risk aversion and seeks to maximize her terminal wealth over the horizon $[t,T]$.  Assuming $X_t = x$ we write $u(t,x;\phi)$ for the value function and $G(t,x;\phi)$ the certainty equivalent.

In this setting, the HJB equation for $G$ is the semi-linear, degenerate elliptic PDE given in \eqref{E:G_PDE} below.  Since we are working on unbounded domains with local ellipticity and unbounded coefficients, the task of solving the PDE is challenging.   However, it is possible to obtain a solution under a \emph{mild} exponential integrability assumption on the market price of risk $\ell \dfn (\mu-\gamma)/\sigma$.  In particular, if the market absent default is "strictly incomplete" in that $\sup_{x\in E}\rho'\rho(x) < 1$, we require (see Assumption \ref{A:opt_main_ass_inc}) only that for some $\eps > 0$, the function $x\mapsto \espalt{}{e^{\eps\int_0^T \ell^2(X_u)du}\big| X_0 = x}$ is locally bounded on $E$.  This assumption is satisfied by virtually all models used in the literature.  When $\rho'\rho$ is not bounded away from $1$, we require locally uniform exponential integrability of $\ell$ under two additional probability measures: see Assumption \ref{A:opt_main_ass_inc}, but, despite the seemingly complicated formulation, it still holds in many models used in industry.

Theorem \ref{T:main_result} shows the certainty equivalent is a classical solution to the HJB equation \eqref{E:G_PDE}.  As an immediate consequence, we are able to identify the (buyer's, per unit) indifference price for owning $q$ units notional of a defaultable bond: see Proposition \ref{P:indiff_px}.  Furthermore, this indifference price is regular both in time and the factor process starting point.

The main application we consider is determining a fair price for dynamic insurance against default.  Here, our work is motivated by \cite{MR2359373}, where a similar problem was considered but it was implicitly assumed the investor is protected from losses due to default. In other words, if the investor holds $\$1$ in the stock, then upon default, she does not lose her $\$1$.  It is natural to wonder how this is possible, and we take the perspective that she has entered into a contract protecting her from losses, and seek to determine how much this contract should cost.  Of course, in reality, it may not be possible to enter into a contract which dynamically protects against default, but such a contract may be thought of a continuous time limit of opening and closing short term static contracts. To compute the contract price we assume the investor has two alternatives: she either does or does not purchase the protection. If she does purchase the protection, she pays a continuous cash flow rate of $\pi_t f_t$ for a given dollar amount $\pi_t$ in the stock.  Given $X_t = x$, her optimal indirect utility is given by $u_d(t,x)$ in the protected case, and $u(t,x)$ in the unprotected case; where $u_d$ was determined in \cite{MR2359373}.  We thus seek $f$ so that $u = u_d$, and in fact, $f$ is very easy to identify by equating the PDEs for the respectively certainty equivalents.  We make this identification in Section \ref{S:insurance}, and show in Proposition \ref{P:default_insurance} that $f$ never exceeds the default intensity under the dual optimal measure for the unprotected problem, but does exceed the default intensity under the physical measure.

The rest of the paper is organized as follows: we present the model in Section \ref{SS:model} and main results concerning the certainty equivalent in Section \ref{SS:main_results}.  Section \ref{S:examples} provides two examples highlighting the minimality of our assumptions.  Section \ref{S:indiff_px} gives the indifference price for $q$ units notional of a defaultable bond, and Section \ref{S:insurance} computes the dynamics default insurance price.  A numerical application in Section \ref{S:num_app} when $X$ is a CIR process concludes: here we explicitly compute the time 0 indifference price as a function of both the notional and factor process, as well as the default insurance cost as a function of the factor process.  Section \ref{S:proof} contains the lengthy proof of Theorem \ref{T:main_result} and Appendix \ref{S:lemmas} contains a number of supporting lemmas. Lastly, we remark that the proof of Theorem \ref{T:main_result} is lengthy since at the present level of generality, we are not able to automatically verify solutions to the PDE in \eqref{E:G_PDE} are the certainty equivalent $G$, as it is difficult to estimate the gradient of solutions near the boundary of the state space.  As such, we must first localize both the PDE and optimal investment problem. At the local level we are able to show existence and uniqueness of solutions.  We then unwind the localization by enforcing the exponential integrability conditions in Assumptions \ref{A:opt_main_ass_inc}, \ref{A:opt_main_ass_com}.

\section{The Setup}\label{S:setup}

\subsection{The probability space and factor process}\label{SS:model}
Consider a probability space $(\Omega,\G,\prob)$ rich enough to support a $d+1$ dimensional Brownian motion, written $(W,W^0)$, where $W$ is $d$-dimensional and $W^0$ one-dimensional; and a random variable $U\sim U(0,1)$ independent of $(W,W^0)$. We denote by $\filt^W$ (respectively $\filt^{W,W^0}$) the $\prob$ augmentation of the filtration generated by $W$ (resp $(W,W^0)$). Throughout,  $\beta\in (0,1]$ is a fixed constant.

There is a time homogenous diffusion factor process $X$ driven by $W$, which takes values in a region $E\subseteq\reals^d$ satisfying:
\begin{assumption}\label{A:region} $E\subset\reals^d$ is open and connected.  Furthermore, there exists a sequence of sub-regions $\cbra{E_n}_{n\in\nats}$ such that each $E_n$ is open, connected, and bounded with $C^{2,\beta}$ boundary $\partial E_n$. Lastly, $\bar{E}_n\subset E_{n+1}$ and $E = \bigcup_n E_n$.
\end{assumption}
To construct $X$ we first assume:
\begin{assumption}\label{A:factor}
$b \in C^{1,\beta}\left(E; \reals^d\right)$ and $A \in C^{2,\beta}\left(E;\mathbb{S}^d_{++}\right)$, where $\mathbb{S}^d_{++}$ is the space of $d\times d$ dimensional symmetric positive definite matrices.  Furthermore, there is a (necessarily unique: see \cite[Chapter 6]{MR2190038}) solution $\cbra{\prob^x}_{x\in E}$ to the martingale problem on $E$ for the operator
\begin{equation}\label{E:L_def}
L \dfn \frac{1}{2}\tr\left(AD^2\right) + b'\nabla.
\end{equation}
Here, $D^2f(x)\in\mathbb{S}^d$ with $(D^2f)^{ij}(x) = (\partial^2/\partial x_i\partial x_j) f(x) = f_{ij}(x)$, and $\nabla f(x) \in \reals^d$ with $(\nabla f)^{i}(x) = (\partial/\partial_{x_i})f(x) = f_i(x)$.
\end{assumption}

\begin{remark}\label{R:mart_prob} Denote by $a\dfn\sqrt{A}$ the unique symmetric positive definite square root of $A$. In light of Assumption \ref{A:factor}, for each $t\geq 0$ and $x\in E$ there is a unique strong solution to the SDE (see \cite[Chapter IX]{MR1725357})
\begin{equation*}
dX^{t,x}_s = b(X^{t,x}_s)ds + a(X^{t,x}_s)dW_s;\qquad s\geq t;\qquad X^{t,x}_s = x\qquad s\leq t.
\end{equation*}
For expectations with respect to $\prob$ we will write, for example, $\espalt{}{f(X^{t,x}_s)}, \espalt{}{\int_t^s f(X^{t,x}_u)du}, s\geq t$ or $\espalt{}{f(X_s)}, \espalt{}{\int_t^s f(X_u)du}$ if the starting point $(t,x)$ is clear.  Alternatively, for expectations with respect to $\cbra{\prob^x}_{x\in E}$ (which acts on the canonical space $C([0,\infty);E)$) we will write $\espalt{x}{f(X_s)}, \espalt{x}{\int_0^s f(X_u)du}$.
\end{remark}

Next, there is an intensity process $\gamma_s = \gamma(X^{t,x}_s), s\geq t$ for the asset default time.  We assume
\begin{assumption}\label{A:intensity}
$\gamma\in C^{1,\beta}\left(E;(0,\infty)\right)$. As such, for each $n$, $\inf_{x\in E_n}\gamma(x) > 0$.
\end{assumption}

\subsection{The underlying asset and optimal investment problem}\label{SS:opt_invest}

We assume a constant interest rate of $0$. The risky asset $S$ may default, but prior to default has instantaneous return $\mu$, variance $\sigma^2$, and correlation $\rho$ with the Brownian motion $W$.  $\mu,\sigma,\rho$ are functions the factor process $X$ and satisfy
\begin{assumption}\label{A:asset_coeff}
$\mu\in C^{1,\beta}\left(E;\reals\right)$, $\sigma \in C^{2,\beta}\left(E; (0,\infty)\right)$ and $\rho\in C^{2,\beta}\left(E;\reals^d\right)$ with $\sup_{x\in E}\rho'\rho(x)\leq 1$. In particular, for each $n$, $\inf_{x\in E_n}\sigma^2(x) > 0$.
\end{assumption}

We now identify the dynamics of $S$.  Throughout, $T>0$ is a fixed time horizon, $0\leq t\leq T$ is the fixed starting time, and $x\in E$ is the fixed factor starting point.  To alleviate notation, in this section we write $X$ rather that $X^{t,x}$ as the factor process, and in general, suppress $(t,x)$.

The default time for $S$ is given by
\begin{equation}\label{E:delta}
\delta \dfn \inf\cbra{s\geq t\such \int_t^s \gamma(X_u)du  = -\log\left(U\right)}.
\end{equation}
As such, for $s>t$: $\condprobalt{}{\delta > s}{\F^{W,W^0}_\infty} = \condprobalt{}{\delta >s}{\F^{W}_s} = e^{-\int_t^s \gamma(X_u)du}$, so that $\gamma(X)$ is the intensity for $\delta$ given $\filt^{W,W^0}$, and $\delta$ has conditional density $\gamma(X_s)e^{-\int_t^s \gamma(X_u)du}$. In particular, $\delta$ is atom-less. Next, we define the default indicator process
\begin{equation}\label{E:H_def}
H_s \dfn 1_{\delta \leq s};\qquad s \geq t,
\end{equation}
as well as the enlarged filtration $\filtg$ as the $\prob$ augmentation of the filtration generated by the sigma-algebras
\begin{equation}\label{E:G_def}
\G_s = \sigma\left(\F^{W,W^0}_s \bigcup \sigma(H_u, t\leq u\leq s)\right);\qquad s\geq t.
\end{equation}
Lastly, we set
\begin{equation}\label{E:M_def}
M_s \dfn H_s - \int_t^{s\wedge \delta} \gamma(X_u)du;\qquad s\geq t,
\end{equation}
and note that $M$ is a $\filtg$ local martingale.  With the notation in place, we assume $S_t = S_0$ and
\begin{equation}\label{E:S_def}
\begin{split}
\frac{dS_s}{S_{s-}} &= 1_{s\leq \delta}\left(\left(\mu-\gamma\right)(X_s)ds + \left(\sigma\rho\right)(X_s)'dW_s + \left(\sigma\sqrt{1-\rho'\rho}\right)(X_s)dW^0_s\right) - dM_s;\qquad s\geq t.\\
\end{split}
\end{equation}
Thus, $S$ evolves as a geometric Brownian motion driven by $\mu,\sigma,\rho,W,W^0$  until $\delta$ at which point it jumps to $0$ and stays there.  Note that the given regularity assumptions, $S$ is well defined.

To describe the optimal investment problem, we first define the class of equivalent local martingale measures with finite relative entropy.  Set
\begin{equation}\label{E:EQMM}
\M \dfn \cbra{\qprob\such \qprob\sim \prob \textrm{ on }\G_T, S \textrm{ is a } \qprob \textrm{ local martingale}}
\end{equation}
Write $Z^{\qprob}_T \dfn d\qprob/d\prob|_{\G_T}$ as the density of $\qprob$ with respect to $\prob$ for $\qprob\in\M$.  The relative entropy of $\qprob$ with respect to $\prob$ on $\G_T$ is given by
\begin{equation}\label{E:rel_ent}
\relent{\qprob}{\prob} \dfn \espalt{}{Z^{\qprob}_T \log\left(Z^{\qprob}_T\right)},
\end{equation}
and we define
\begin{equation}\label{E:M_rel_ent}
\tM \dfn \cbra{Q\in\M\such \relent{\qprob}{\prob} < \infty}.
\end{equation}
For now we assume $\tM\neq\emptyset$ but later on we will enforce this assumption via a requirement on the model coefficients (c.f. Assumptions \ref{A:opt_main_ass_inc}, \ref{A:opt_main_ass_com} below).

Next, let $\pi$ be a $\filtg$ predictable process such that the stochastic integral $\int_t^\cdot (\pi_u/S_{u-})dS_u$ is well defined. $\pi_u$ represents the dollar amount invested in $S$ at time $u$. Denote by $\We^{\pi,w}$ the resultant (self-financing) wealth process with initial value $w$ at $t$.  For $s\geq t$, $\We^{\pi,w}$ has dynamics
\begin{equation}\label{E:wealth_dynamics}
\begin{split}
d\We^{\pi,w}_s &= \pi_s \frac{dS_s}{S_{s-}};\\
&= \pi_s1_{s\leq \delta}\left(\left(\mu-\gamma\right)(X_s)ds + \left(\sigma\rho\right)(X_s)'dW_s + \left(\sigma\sqrt{1-\rho'\rho}\right)(X_s)dW^0_s\right) - \pi_s dM_s.
\end{split}
\end{equation}
Note that upon default at $\delta$ the investor losses her dollar position $\pi_{\delta}$, and after this point, there is no change in wealth.  We then say $\pi$ is admissible, and write $\pi\in\mathcal{A}$, if $\We^{\pi,w}$ is a $\qprob$ super-martingale for all $\qprob\in\tM$.

In addition to trading in the stock, the investor owns a non-traded claim which is received at the horizon $T$, contingent upon the survival of the stock.  We consider claims of the form $1_{\delta > T}\phi(X_T)$, so that the payoff may depend on the factor process. However, our interest primarily lies when either  $\phi\equiv 0$ (no claim) or $\phi\equiv 1$ (defaultable bond).  We make the common assumption that $\phi$ is bounded, along with a certain regularity requirement:
\begin{assumption}\label{A:phi}
$\phi\in C^{2,\beta}\left(E;\reals\right)$ is bounded with
\begin{equation}\label{E:phi_bounds}
\ul{\phi}\dfn 0 \wedge \inf_{x\in E}\phi(x);\qquad \ol{\phi}\dfn 0\vee \sup_{x\in E}\phi(x).
\end{equation}
\end{assumption}

The investor's preferences are described by the exponential utility function
\begin{equation}\label{E:U}
U(w) \dfn -e^{-\alpha w},
\end{equation}
where $\alpha>0$ is the absolute risk aversion.  The investor trades in $S$ in order to maximize her expected utility of terminal wealth, including her position in the contingent claim.  Thus, for a given wealth $w$ we define (recall: the investment window starts at $t$ and the factor process starts at $x$):
\begin{equation}\label{E:util_funct}
u(t,x;w,\phi) \dfn \sup_{\pi\in\mathcal{A}}\espalt{}{-e^{-\alpha\left(\We^{\pi,w}_T + 1_{\delta>T}\phi(X_T)\right)}};\quad u(t,x;\phi)\dfn u(t,x;0,\phi).
\end{equation}
It is clear for exponential utility that $u(t,x;w,\phi) = e^{-aw}u(t,x;\phi)$.  As such, until Section \ref{S:indiff_px} we consider $w=0$. Lastly, we write $G(t,x;\phi)$ as the certainty equivalent, defined by
\begin{equation}\label{E:CE}
G(t,x;\phi) \dfn -\frac{1}{a}\log\left(-u(t,x;\phi)\right).
\end{equation}

\subsection{Main result}\label{SS:main_results}

Before presenting the main result, we must introduce one last piece of notation.  For $y>0$ define $\theta(y)$ as the unique solution to
\begin{equation}\label{E:theta_def}
\theta(y) e^{\theta(y)} = y.
\end{equation}
$\theta$ is known as the  "Product-Log" (Mathematica) or "Lambert-W" (Matlab) function.  Further properties of $\theta$ are given in Lemma \ref{L:theta_lem} below.  With this notation, for a smooth function $f$ on $[0,T]\times E$, define
\begin{equation}\label{E:theta_f}
\theta_f(s,y) \dfn \theta\left(\frac{\gamma(y)}{\sigma^2(y)} e^{\frac{\mu(y)}{\sigma^2(y)} + \alpha f(s,y) - \frac{\alpha}{\sigma(s)}\nabla f(s,y)'a(y)\rho(y)}\right).
\end{equation}

A heuristic derivation using the dynamic programming principal indicates $G$ from \eqref{E:CE} should solve the semi-linear, degenerate elliptic partial differential equation (PDE) or Hamilton-Jacoby-Bellman (HJB) equation (here, we suppress the function arguments $(t,x)$ and recall $L$ from \eqref{E:L_def})
\begin{equation}\label{E:G_PDE}
\begin{split}
0 & = G_t + LG - \frac{\alpha}{2}\nabla G'A\nabla G + \frac{\sigma^2}{2\alpha}\left(\frac{2\gamma}{\sigma^2} + \left(\frac{\mu}{\sigma^2} - \frac{\alpha}{\sigma}\nabla G'a\rho\right)^2 -\theta_G^2 - 2\theta_G\right);\ t\in (0,T),x\in E\\
\phi & = G(T,\cdot;\phi);\ x\in E
\end{split}
\end{equation}

The dynamic programming principal also suggests that if $G$ solves the above PDE then for any starting point $(t,x)$ the optimal trading strategy is

\begin{equation}\label{E:opt_pi}
\hpi_s = \hpi(s,X_s) = \hpi(s,X^{t,x}_s);\quad \hpi\dfn \frac{1}{\alpha}\left(\frac{\mu}{\sigma^2} - \frac{\alpha}{\sigma}\nabla G(\cdot;\phi)'a\rho - \theta_{G(\cdot;\phi)}\right);\quad t\leq s\leq T.
\end{equation}

\begin{remark}
We do not derive the HJB equation as it is standard.  Furthermore, we will use direct methods to a) yield a solution $G$ to \eqref{E:G_PDE} and b) verify that it is the certainty equivalent, and $\hpi$ the optimal policy.  As such, we do not require the dynamic programming principal to hold.
\end{remark}

We must enforce one more restriction to solve the optimal investment problem.  Here, we split into cases depending on the correlation $\rho$.  The first is when the market absent default is "strictly incomplete" in that $\rho'\rho$ is bounded below $1$, while the second  places no restrictions on $\rho'\rho$.  We make this split because in the former case, the main result goes through under a very mild, and simple to state, condition upon the market price of risk $(\mu-\gamma)/\sigma$.  In the latter case, the condition we must assume is more complicated to formulate, but none-the-less holds in many models of interest: see Section \ref{S:examples}.

First we consider when $\rho'\rho$ is bounded below $1$, and recall the notation of Remark \ref{R:mart_prob}:

\begin{assumption}\label{A:opt_main_ass_inc}
$\ol{\rho}\dfn \sup_{x\in E}\rho'\rho(x) < 1$ and there exists an $\eps>0$ so that for each $n$
\begin{equation*}
\sup_{x\in \ol{E}_n} \espalt{x}{e^{\eps\int_0^T \left(\frac{\mu-\gamma}{\sigma}\right)(X_u)^2 du}} = C(n) < \infty.
\end{equation*}
\end{assumption}

Next, we consider when $\rho'\rho$ is not bounded below $1$:

\begin{assumption}\label{A:opt_main_ass_com}
There are no restrictions on $\ol{\rho}$. However,
\begin{enumerate}[(A)]
\item There is a solution to the martingale problem on $E$ for the operator
\begin{equation*}
L_0 \dfn \frac{1}{2}\tr\left(AD^2\right) + \left(b-\frac{\mu-\gamma}{\sigma}a\rho\right)'\nabla.
\end{equation*}
With $\prob_0 = \cbra{\prob_0^{x}}_{x\in E}$ denoting the resultant solution, there is some $\eps > 0$ so that for each $n$:
\begin{equation*}
\sup_{x\in \ol{E}_n} \espalt{\prob_0^x}{e^{\eps\int_0^T \left(\frac{\mu-\gamma}{\sigma}\right)(X_u)^2 du}} = C(n) < \infty.
\end{equation*}
\item For some $p>1$ there is a solution the martingale problem on $E$ for the operator
\begin{equation*}
L_{p} \dfn \frac{1}{2}\tr\left(A D^2\right) + \left(b+(p-1)\frac{\mu-\gamma}{\sigma}a\rho\right)'\nabla.
\end{equation*}
With $\prob_{p} = \cbra{\prob_{p}^x}_{x\in E}$ denoting the resultant solution, we have for each $n$ that
\begin{equation*}
\sup_{x\in \ol{E}_n} \espalt{\prob_{p}^x}{e^{\frac{1}{2}p(p-1)\int_0^T \left(\frac{\mu-\gamma}{\sigma}\right)(X_u)^2 du}} = C(n) < \infty.
\end{equation*}
\end{enumerate}

\end{assumption}

\begin{remark} Note that each condition in Assumption \ref{A:opt_main_ass_inc} allows for unbounded $(\mu-\gamma)/\sigma$. To gain intuition for why we split into two cases, note that if $\ol{\rho} < 1$ then one may obtain a measure $\qprob_0\in\M$ via the density process
\begin{equation*}
Z_{0,s} = \EN\left(-\int_t^\cdot \left(\frac{\mu-\gamma}{\sigma\sqrt{1-\rho'\rho}}\right)(X_u)dW^0_u\right)_s;\quad t\leq s\leq T.
\end{equation*}
Since $X$ is independent of $W^0$, we see that $\qprob_0\in\tM$ when Assumption \ref{A:opt_main_ass_inc} holds: in fact, this assumption is much stronger than what is actually needed, but is used in the delicate unwinding procedure of Proposition \ref{P:unwind_prob_A} below.  However, when $\ol{\rho} =1$ we are no longer able to "put" the market price of risk to the independent Brownian motion $W^0$: here we instead consider the density process
\begin{equation*}
Z_{0,s} = \EN\left(-\int_t^\cdot \left(\frac{\mu-\gamma}{\sigma}\rho'\right)(X_u)dW_u\right)_s;\quad t\leq s\leq T.
\end{equation*}
Since $X$ is not independent of $W$ we must enforce the integrability requirements of Assumption \ref{A:opt_main_ass_com} to obtain our results.  But, as is shown in Section \ref{S:examples}, Assumption \ref{A:opt_main_ass_inc}, while seemingly complicated, is very easy to check, and holds under mild parameter restrictions for many common models.
\end{remark}

With all the assumptions in place, we state the main result. For a definition of the parabolic H\"{o}lder space $H_{2+\beta,(0,T)\times E,\textrm{loc}}$ see Section \ref{SS:moll_notation} below, but for now we remark that $H_{2+\beta,(0,T)\times E,\textrm{loc}}\subset C^{1,2}((0,T)\times E)$.

\begin{theorem}\label{T:main_result}
Let Assumptions \ref{A:region}, \ref{A:factor}, \ref{A:intensity}, \ref{A:asset_coeff}, \ref{A:phi} and either \ref{A:opt_main_ass_inc} or \ref{A:opt_main_ass_com} hold.  Then, the certainty equivalent $G$ from \eqref{E:CE} is in $H_{2+\beta,(0,T)\times E, \textrm{loc}}$ and solves the PDE in \eqref{E:G_PDE}.  Furthermore, for each $0\leq t\leq T, x\in E$, the optimal trading strategy is $\hpi\in\mathcal{A}$ from \eqref{E:opt_pi} and the process $\hat{Z} = \hat{Z}^{t,x}$ defined by
\begin{equation}\label{E:hat_Z}
\hat{Z}_s \dfn e^{-\alpha\left(\We^{\hpi}_s - G(t,x;\phi) + 1_{\delta > s}G(s,X_s;\phi)\right)};\qquad t\leq s\leq T,
\end{equation}
is the density process of a measure $\hat{\qprob}\in \tM$ which solves the dual problem given in \eqref{E:dual_util_funct} below.
\end{theorem}

\section{Examples}\label{S:examples}

\subsection{OU factor process with constant default intensity}\label{SS:OU}  Consider when $X$ has dynamics $dX_t = -bX_t dt + dW_t$ ($b\in\reals$), taking values in $E = \reals$,  for which we set $E_n = (-n,n)$. The correlation $\rho \in [-1,1]$ is constant.  For the asset dynamics we take $\sigma(x) = \sigma > 0$ and  $\mu(x) = \sigma\left(\mu_1 + \mu_2x\right), \mu_1,\mu_2\in\reals$. The default intensity is $\gamma(x) = \sigma\gamma, \gamma > 0$. Clearly, Assumptions \ref{A:region}, \ref{A:factor}, \ref{A:intensity} and \ref{A:asset_coeff} hold.  The claim is either $\phi\equiv 0$ or $\phi\equiv 1$ so that Assumption \ref{A:phi} holds as well.  As for the more complicated Assumptions \ref{A:opt_main_ass_inc}, \ref{A:opt_main_ass_com} we note that
\begin{enumerate}[(1)]
\item $((\mu-\gamma)^2/\sigma^2)(x) \leq 2(\mu_1-\gamma_1)^2 + 2\mu_2^2 x^2$
\item Under each of the operators in Assumptions \ref{A:opt_main_ass_inc}, \ref{A:opt_main_ass_com}, $X$ is an OU process.
\item $\esp\bra{e^{\int_0^T kX_t^2dt}} \leq (1/T)\int_0^T \esp\bra{e^{kT X_t^2}}dt$ for $k\in\reals$.
\item For $k>0$ small enough and $X_0=x$, $\esp\bra{e^{kT X_t^2}} \leq C(T)e^{C(T)x^2}$ for $t\leq T$, since $X_t$ is normally distributed with mean and variance continuous in $t$.
\end{enumerate}
Therefore, it is easy to see that both Assumptions \ref{A:opt_main_ass_inc}, \ref{A:opt_main_ass_com}, hold with no additional parameter restrictions.

\subsection{CIR factor process with affine default intensity}\label{SS:CIR} Consider when $X$ has dynamics $dX_t = \kappa(\theta - X_t)dt + \xi\sqrt{X_t}dW_t$.  We assume $\kappa > 0, \kappa\theta \geq \xi^2/2$ so that $E=(0,\infty)$ ($E_n =((1/n),n)$), $b(x) = \kappa(\theta-x)$, $A(x) = \xi^2 x$.  The correlation $\rho\in [-1,1]$ is again constant. For the asset dynamics we take $\sigma(x) = \sigma\sqrt{x}, \sigma > 0$ and $\mu(x) = \sigma(\mu_1 + \mu_2 x), \mu_1,\mu_2\in\reals$.  Lastly, we assume $\gamma(x) = \sigma(\gamma_1 + \gamma_2 x), \gamma_1,\gamma_2 > 0$. Thus, the model falls into the "extended affine" class: see \cite{dai2000saa,MR1994043,filipovic2001general} amongst others.  As with the OU case, Assumptions \ref{A:region}, \ref{A:factor}, \ref{A:intensity} and \ref{A:asset_coeff} clearly hold.  As for Assumptions \ref{A:opt_main_ass_inc}, \ref{A:opt_main_ass_com} we have the following:

\begin{lemma}\label{L:CIR_goes_through}
Assume $\kappa\theta > (1/2)\xi^2$ and additionally, if $\rho = 1$\footnote{The case $\rho = -1$ is similar.}, then $\mu_1-\gamma_1 < (1/\xi^2)(\kappa\theta-(1/2)\xi^2)$ and $\mu_2-\gamma_2 > -\kappa/\xi^2$. Then Assumptions \ref{A:opt_main_ass_inc}, \ref{A:opt_main_ass_com} hold.
\end{lemma}

\begin{proof}[Proof of Lemma \ref{L:CIR_goes_through}]
We first claim that if $X$ is a general CIR process with dynamics $dX_t = \tilde{\kappa}\left(\tilde{\theta}-X_t\right)dt + \tilde{\xi}\sqrt{X_t}dW_t, X_0 = x$, then provided $\tilde{\kappa} > 0$ and  $\tilde{\kappa}\tilde{\theta} > (1/2)\tilde{\xi}^2$  we have for $0 < A < (\tilde{\kappa}\tilde{\theta}-(1/2)\tilde{\xi}^2)/(2\tilde{\xi}^2)$ and $0 < B < \tilde{\kappa}^2/(2\tilde{\xi}^2)$ that
\begin{equation*}
\espalt{x}{e^{\int_0^T\left(\frac{A}{X_t} + BX_t\right)dt}} \leq \left(\frac{C e}{D}\right)^C x^{-C} e^{Dx + \lambda T},
\end{equation*}
where $C = (\tilde{\kappa}\tilde{\theta} - (1/2)\tilde{\xi}^2)/\tilde{\xi}^2\left(1 - \sqrt{1 - 2\tilde{\xi}^2A/(\tilde{\kappa}\tilde{\theta}-(1/2)\tilde{\xi}^2)^2}\right)$, $D = (\tilde{\kappa}/\tilde{\xi}^2)\left(1-\sqrt{1-2\tilde{\xi}^2B/\tilde{\kappa}^2}\right)$,
and $\lambda = \tilde{\kappa}C + \tilde{\kappa}\tilde{\theta}D - \tilde{\xi^2}{CD}$. Indeed, note that $C,D>0$ and set $f(x) = x^{-C}e^{Dx}$, $\tilde{L}$ as the second order operator associated to $X$.  Then $\tilde{L}f(x) + (A/x + Bx)f = \lambda f$ and by \ito's formula we see that
\begin{equation*}
\espalt{x}{e^{\int_0^T\left(\frac{A}{X_t} + BX_t\right)dt}} = x^{-C} e^{Dx + \lambda T}\espalt{x}{M_T X_T^C e^{-DX_T}} \leq \left(\frac{eC}{D}\right)^C x^{-C} e^{Dx+\lambda T},
\end{equation*}
where $M_\cdot = \EN\left(\int_0^\cdot \left(-C/X_t + D\right)\tilde{\xi}\sqrt{X_t}dW_t\right)_\cdot$, and where the last inequality holds since $C,D>0$ imply $x^{C}e^{-Dx} \leq (eC/D)^C$, and since $M$ is a super-martingale. Next, note that $((\mu-\gamma)^2/\sigma^2)(x) = (\mu_1-\gamma_1)^2/x + 2(\mu_1-\gamma_1)(\mu_2-\gamma_2) + (\mu_2-\gamma_2)^2x$, and that under each of the operators in Assumption \ref{A:opt_main_ass_inc}, \ref{A:opt_main_ass_com} $X$ is still a CIR process.  The parameter restrictions imply $X$ does not explode under each of the operators (for $p>1$ small enough), and a straight-forward calculation yields Assumptions \ref{A:opt_main_ass_inc}, \ref{A:opt_main_ass_com}.
\end{proof}

\section{Indifference Pricing}\label{S:indiff_px}

As an immediate application of Theorem \ref{T:main_result}, we can obtain the indifference price for a defaultable bond, taking into account the investor's ability to trade in the defaultable asset.  Here, we assume the buyer holds notional $q > 0$ of the defaultable bond, and seek to identify the (per unit notional) indifference price.  Formally, we fix $0\leq t\leq T, x\in E$ and define the (buyer's average) indifference $p(t,x;w,q)$ via the balance equation
\begin{equation}\label{E:abstract_indiff_px}
u(t,x;w,0) = u(t,x;w-qp(t,x;w,q),q).
\end{equation}
Above, we have used \eqref{E:util_funct} and written $0$ for $\phi\equiv 0$ and $q$ for $\phi\equiv q$. Qualitatively, \eqref{E:abstract_indiff_px} states that at the price $p(t,x;w,q)$ the buyer of the claim is indifferent between not owning the claim and paying $qp(t,x;w,q)$ to purchase $q$ units of the claim.  For exponential utility it is well known the initial wealth $w$ does not affect the indifference price so we write $p(t,x;w,q) = p(t,x;q)$ where
\begin{equation}\label{E:exp_indiff_px}
p(t,x;q) = -\frac{1}{\alpha q}\log\left(\frac{u(t,x;q)}{u(t,x;0)}\right) = \frac{1}{q}\left(G(t,x;q)-G(t,x;0)\right),
\end{equation}
where $G$ is as in \eqref{E:CE}.  Thus, Theorem \ref{T:main_result} implies

\begin{proposition}\label{P:indiff_px}
Let Assumptions \ref{A:region}, \ref{A:factor}, \ref{A:intensity}, \ref{A:asset_coeff} and either of Assumptions \ref{A:opt_main_ass_inc} or \ref{A:opt_main_ass_com} hold.  Then for a fixed $0\leq t\leq T, x\in E$ the per unit buyer's indifference price for owning $q$ units notional of a defaultable bond is given by $p(t,x;q)$ from \eqref{E:exp_indiff_px}.  For a fixed $q$ the indifference price is in $H_{2+\beta,[0,T]\times E,\textrm{loc}}$.
\end{proposition}

\section{The Pricing of Dynamic Default Insurance}\label{S:insurance} We now consider the problem of finding a fair price for dynamic insurance against default, which takes into account market incompleteness as well as investor preferences.  The results of this section are meant to complement those in \cite{MR2359373}, where a similar optimal investment problem is considered, with the difference being that upon default of $S$, it is implicitly assumed the investor is protected \footnote{Indeed, in \cite[Section 2.1]{MR2359373} the author's state "For simplicity, we assume she (the investor) receives full pre-default market value on her stock holdings on liquidation, though one might extend to consider some loss, or jump downwards in the stock price at default time."}. Specifically, if the investor owns a dollar amount $\pi_\delta$ at $\delta$, she does not lose this dollar amount. Rather, the default only indicates that further trading is not possible.

Presently, we describe the method for obtaining a fair price of dynamic default insurance.  In the model of Section \ref{SS:opt_invest} with no defaultable bond, we assume the investor has two potential strategies:
\begin{enumerate}[(1)]
\item The investor chooses a strategy $\pi\in\mathcal{A}$ as above, and does not purchase default insurance, so that upon default at $\delta$ she will lose her dollar position $\pi_\delta$ in the stock. Thus, for a fixed $0\leq t\leq T, x\in E$ her indirect utility is given by $u(t,x;0)$ from \eqref{E:util_funct}, with certainty equivalent $G(t,x;0)$ from \eqref{E:CE}.
\item The investor again chooses a strategy $\pi^d$ in the stock. However, she enters into a contract which offers dynamic protection against default.  Prior to default, this contract costs a continuous payment rate of $\pi^d f$ where $f$ is the (to-be-determined) per-dollar cost of the protection.  At default, the investor does not lose her dollar-position $\pi^d_\delta$, but is no longer able to trade in the underlying asset.  In this scenario, for a given strategy $\pi^d$, the wealth process evolves according to (see \eqref{E:wealth_dynamics} for comparison)
    \begin{equation}\label{E:wealth_dynamics_alt}
    d\We^{\pi,d}_{s} = \pi^d_s1_{\delta\leq s}\left(\left((\mu-\gamma)(X_s)ds - f_s\right)ds + (\sigma\rho)(X_s)'dW_s + (\sigma\sqrt{1-\rho'\rho})(X_s)dW^0_s\right).
    \end{equation}
    Here, $\pi^d\in\mathcal{A}^d$ is admissible if $\pi^d$ is $\filtg$ predictable, such that the relevant integrals are well defined, and such that $\We^{\pi,d}$ is a $\qprob^d$ super-martingale for all $\qprob^d\in \tM_d$, the class of equivalent local martingale measures with finite relative entropy in this market.  She then seeks to solve the optimal investment problem
    \begin{equation}\label{E:util_funct_al}
    u^d(t,x;0) \dfn \sup_{\pi\in\mathcal{A}^d}\espalt{}{-e^{-\alpha \We^{\pi,d}_T}},
    \end{equation}
    and we denote by $G^d(t,x;0) \dfn -(1/\alpha)u^d(t,x;0)$ the certainty equivalent.
\end{enumerate}

The goal is to find $f$ so that the investor is indifferent between the two alternatives; i.e. $u(t,x;0) = u_d(t,x;0)$ for all starting points $(t,x)$.  In light of Theorem \ref{T:main_result}, we can immediately identify $f$ by equating the PDEs for $G,G^d$.  Indeed, following the arguments in \cite[Proposition 2.1]{MR2359373}, under the a-priori assumption that $f_t = f(t,X_t)$ is functionally determined, the HJB equation for $G^d$ on $[0,T]\times E$ is
\begin{equation}\label{E:G_d_PDE}
\begin{split}
0 & = G^d_t + LG^d - \frac{\alpha}{2}(\nabla G^d)'A\nabla G^d + \frac{\sigma^2}{2\alpha}\left(\frac{2\gamma}{\sigma^2}\left(1-e^{\alpha G^d}\right) + \left(\frac{\mu-f}{\sigma^2} - \frac{\alpha}{\sigma}(\nabla G^d)'a\rho\right)^2\right);\\
0 &= G^d(T,\dot); x\in E.
\end{split}
\end{equation}
The corresponding optimal trading strategy for $t\leq s\leq T$ is
\begin{equation}\label{E:opt_pi_d}
\hpi^d_s  = \hpi^d(s,X_s) = \hpi^d(s,X^{t,x}_s);\quad \hpi^d = \frac{1}{\alpha}\left(\frac{\mu-f}{\sigma^2} - \frac{\alpha}{\sigma}(\nabla G^d)'a\rho\right).
\end{equation}

Now, consider when Theorem \ref{T:main_result} holds, so that $G=G(\cdot;0)$ solves \eqref{E:G_PDE} with $\phi \equiv 0$. Upon comparison with \eqref{E:G_d_PDE} we see that $G$ will also solve \eqref{E:G_d_PDE} provided
\begin{equation*}
\frac{2\gamma}{\sigma^2}\left(1-e^{\alpha G}\right) + \left(\frac{\mu-f}{\sigma^2} - \frac{\alpha}{\sigma}\nabla G'a\rho\right)^2 = \frac{2\gamma}{\sigma^2} + \left(\frac{\mu}{\sigma^2} - \frac{\alpha}{\sigma}\nabla G'a\rho\right)^2 - \theta_G^2 - 2\theta_G.
\end{equation*}
This has two (real) solutions
\begin{equation*}
f_{\pm} = \sigma^2\left(\frac{\mu}{\sigma^2} - \frac{\alpha}{\sigma}\nabla G'a\rho \pm \sqrt{\left(\frac{\mu}{\sigma^2} - \frac{\alpha}{\sigma}\nabla G'a\rho\right)^2 - \left(\theta_G^2 + 2\theta_G - \frac{2\gamma}{\sigma^2}e^{\alpha G}\right)}\right),
\end{equation*}
as part $(3)$ of Lemma \ref{L:theta_lem} below at $x=((\mu/\sigma^2)-(\alpha/\sigma)\nabla G'a\rho)$ and $y=(\gamma/\sigma^2)e^{\alpha G}$ shows the term within the square root is non-negative.  Since the investor is paying for the default insurance, we take the perspective that she seeks the lowest possible cost and hence set the price of default insurance as
\begin{equation}\label{E:default_insurance}
f \dfn  \sigma^2\left(\frac{\mu}{\sigma^2} - \frac{\alpha}{\sigma}\nabla G'a\rho  - \sqrt{\left(\frac{\mu}{\sigma^2} - \frac{\alpha}{\sigma}\nabla G'a\rho\right)^2 - \left(\theta_G^2 + 2\theta_G - \frac{2\gamma}{\sigma^2}e^{\alpha G}\right)}\right)\footnote{ A second argument in favor of using $f_{-}$ follows by inspecting the optimal $\hpi^d$ in the protected market from \eqref{E:opt_pi_d}.  Here
\begin{equation*}
\alpha \hpi^d_{\pm} = \mp \sqrt{\left(\frac{\mu}{\sigma^2} - \frac{\alpha}{\sigma}\nabla G'a\rho\right)^2 - \left(\theta_G^2 + 2\theta_G - \frac{2\gamma}{\sigma^2}e^{\alpha G}\right)}.
\end{equation*}
Thus, if $f_+ >0$ is used, the investor is short the stock while also paying for default insurance, a highly questionable situation.}.
\end{equation}
For $f$ defined above, we have the following proposition:
\begin{proposition}\label{P:default_insurance}
Let Assumptions \ref{A:region}, \ref{A:factor}, \ref{A:asset_coeff}, \ref{A:intensity} and either of \eqref{A:opt_main_ass_inc}, \ref{A:opt_main_ass_com} hold.  Let $G$ from \eqref{E:CE} be the certainty equivalent to \eqref{E:util_funct} which solves the PDE in \eqref{E:G_PDE}.  Define $f$ as in \eqref{E:default_insurance} and consider the optimal trading strategy function $\hpi$ from \eqref{E:opt_pi}.  Then we have the following facts regarding $f$:
\begin{enumerate}[(1)]
\item $f\leq \gamma e^{\alpha(G+\hpi)} = \gamma^{\hat{\qprob}}$, the default intensity function of $\delta$ under $\hat{\qprob}$ from Theorem \ref{T:main_result}. Here, equality is achieved if and only if $\hat{\pi} = 0$.
\item $f$ has the same sign (including $0$) as $\gamma e^{\alpha(2\hpi + G)}/(2\sigma^2) + e^{\alpha\hpi} - 1$.  In particular, $f>0$ when $\hpi > 0$.
\end{enumerate}
\end{proposition}

\begin{proof}[Proof of Proposition \ref{P:default_insurance}]
With $x=(\mu/\sigma^2)-(\alpha/\sigma)\nabla G'a\rho$ and $y=(\gamma/\sigma^2)e^{\alpha G}$ we see that $f/\sigma^2 = x - \sqrt{x^2 + 2y - \theta(ye^x)^2 - 2\theta(ye^x)}$.  Now, \eqref{E:opt_pi} implies $x = \alpha\hpi + \theta(ye^x)$, and writing $z=\theta(ye^x)$, we have by definition of $\theta$ that $ze^z = ye^x = y e^{\alpha\hpi + z}$. Thus, $z = ye^{\alpha\hpi}$ and $x=\alpha\hpi + ye^{\alpha\hpi}$.  Plugging in for $x,z$ in \eqref{E:default_insurance} gives
\begin{equation}\label{E:default_insurance_alt}
\frac{f}{\sigma^2} = h\left(\alpha\hpi,\frac{\gamma}{\sigma^2}e^{\alpha G}\right) ;\quad h(l,y) \dfn l + ye^{l} - \sqrt{l^2 + 2y(le^{l} +1 - e^l)}.
\end{equation}
Straightforward analysis shows for $l\in\reals, y>0$ that $l - \sqrt{l^2 + 2y(le^{l} +1 - e^l)}\leq 0$ and so $f/\sigma^2 \leq  (\gamma/\sigma^2)e^{\alpha(G+\hpi)}$ and $(1)$ follows since $\gamma e^{\alpha(G+\hpi)}$ is the default intensity of $\delta$ under $\hat{\qprob}$ as can be seen from Lemma \ref{L:BR_Girsanov} and \eqref{E:Z_form} below.  As for $(2)$, analysis shows there is a unique $l_0 = l_0(y) < 0$ so that $h(l_0,y) = 0$.  For $l>l_0$, $h(l,y)>0$ and for $l<l_0$, $h(l,y) < 0$.  Thus, $f$ has the same sign as $l-l_0$. Simple calculation shows both $l_0 = \log((\sqrt{1+2y}-1)/y)$ and that $l-l_0$ has the same sign as $(y/2)e^{2l} + e^l -1$. Plugging in $l=\alpha\hpi$ and $y=(\gamma/\sigma^2)e^{\alpha G}$ gives $(2)$.

\end{proof}

\subsubsection{Discussion}\label{SSS:default_insurance_discussion}

To interpret Proposition \ref{P:default_insurance}, consider when $\hpi > 0$. Here, the investor will lose money upon default, and hence has motivation to pay for default insurance so that $f>0$. Interestingly however, there is a threshold $\hpi_0<0$, where as long as $\hpi_0 < \hpi < 0$ the profit from immediate default does not outweigh the value of default protection if future positions rise: thus the investor is willing to pay for insurance.  Below the level $\hpi_0$ the profit from default is so significant that the investor will need to be compensated for giving this up and hence $f<0$.

A full calculation of $f = f(t,x)$ requires knowledge of $G,\nabla G$, and can be difficult to calculate, especially when $X$ is multi-dimensional.  However, for $t\approx T$ we can use \eqref{E:default_insurance_alt} to provide a simple relationship between $f$ and $\hpi$. Indeed, since $G(T,\cdot) = 0$, if we substitute $G(t,\cdot)\approx 0$ into \eqref{E:default_insurance_alt} we see for $t\approx T$ that
\begin{equation}\label{E:default_insurance_approx}
\frac{f}{\sigma^2}  = \alpha\hpi + \frac{\gamma}{\sigma^2}e^{\alpha\hpi} - \sqrt{(\alpha\hpi)^2 + \frac{2\gamma}{\sigma^2}\left(\alpha\hpi e^{\alpha \hpi} + 1 - e^{\alpha\hpi}\right)}.
\end{equation}
Consider the models of Sections \ref{SS:OU}, \ref{SS:CIR}, where in the latter model we enforce the affine condition by setting $\mu_1=\gamma_1=0$. In each case, $\gamma/\sigma^2$ is constant, and we can view $f/\sigma^2$ solely as a function of $\alpha\hpi$. Figure \ref{F:Ins_Fig} shows the relationship between $f/\sigma^2$ and $\alpha\hpi$ in this setting, along with the theoretical upper bound from $(1)$ of Proposition \ref{P:default_insurance}.

\begin{figure}
\epsfig{file=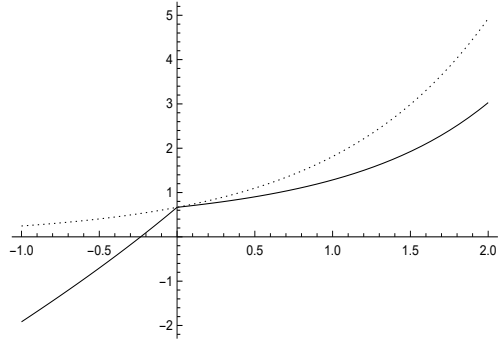,height=4.5cm,width=6.5cm}
\caption{$f/\sigma^2$ as a function of $\alpha\hpi$ for $t\approx T$ and constant $\gamma/\sigma^2= 2/3$.  The solid line is the value of $f/\sigma^2$ from \eqref{E:default_insurance_approx} while the dashed line is the theoretical upper bound from part (1) of Proposition \ref{P:default_insurance}.} \label{F:Ins_Fig}
\end{figure}

\section{A Numerical Application}\label{S:num_app}  Consider the CIR model of Section \ref{SS:CIR} but restricted to the affine, rather than extended affine, case.  In particular we assume $\mu_1 = \gamma_1 = 0$.  We also assume $\rho\neq\pm 1$ so that Theorem \ref{T:main_result} holds under the assumption $\kappa\theta > (1/2)\xi^2$.  The goal is to compute the indifference price of \eqref{E:exp_indiff_px} for $q$ units notional of a defaultable bond, as well as per unit fair price of dynamic default insurance given in \eqref{E:default_insurance}.  To obtain these values, we solve the PDE \eqref{E:G_PDE} for $\phi\equiv 0$ and $\phi\equiv q$ using the semi-linear PDE solver "pdepe" from Matlab, which can solve such PDEs in one spatial dimension\footnote{Our code is available upon request by contacting the second author at "scottrob@bu.edu".}.  For $X$ and $\rho$ we use the same parameter values as in \cite[Section 6]{MR2352905}. For the instantaneous return and variance we assume at the long run mean $\theta$ of $X$ a variance of $\sigma^2\theta = .09$ and mean of $\sigma\mu_2\theta = .10$. For the default intensity we assume at the long run mean of $\theta$ we have a probability of default within one year is $3\%$, which corresponds to $e^{-\sigma\gamma\theta} = .97$.  Lastly, we use a risk aversion of $3$. This yields
\begin{equation}\label{E:CIR_num_param}
\begin{split}
\mu_1 &=0;\ \mu_2 = 1.3608;\ \sigma = 1.2247;\ \rho = -0.53;\\
\kappa &= 0.25;\ \theta = 0.06;\ \xi = 0.1;\\
\gamma_1 &= 0;\ \gamma_2 = 0.4145;\ \alpha = 3.
\end{split}
\end{equation}
The conditions of Lemma \ref{L:CIR_goes_through} are satisfied so Theorem \ref{T:main_result}, as well as Propositions \ref{P:indiff_px}, \ref{P:default_insurance} go through.  At time $t=0$, Figure \ref{F:def_ins_indiff_px} shows the dynamic default insurance price $f(0,x)$ and the indifference price $p(0,x;q)$ as a function of the underlying state variable $x$.  For the dynamic insurance price we have also plotted the upper bound $\gamma^{\hat{\qprob}}$ from Proposition \ref{P:default_insurance} and the default intensity under $\prob$, which in this case is also the intensity under the minimal martingale measure: see \cite{MR1108430,MR1353193}. Here, we see that the insurance price is increasing with the state variable, as intuition would dictate since the physical measure default intensity is linear in $x$.

For the indifference prices, we have plotted the price as a function of the notional size $q$ and state variable $x$ for $q=1,3,5,10$. As expected the price decreases in $q$ and increases in $x$ (recall: this is the price at which we would buy the defaultable bond).

\begin{figure}
\epsfig{file=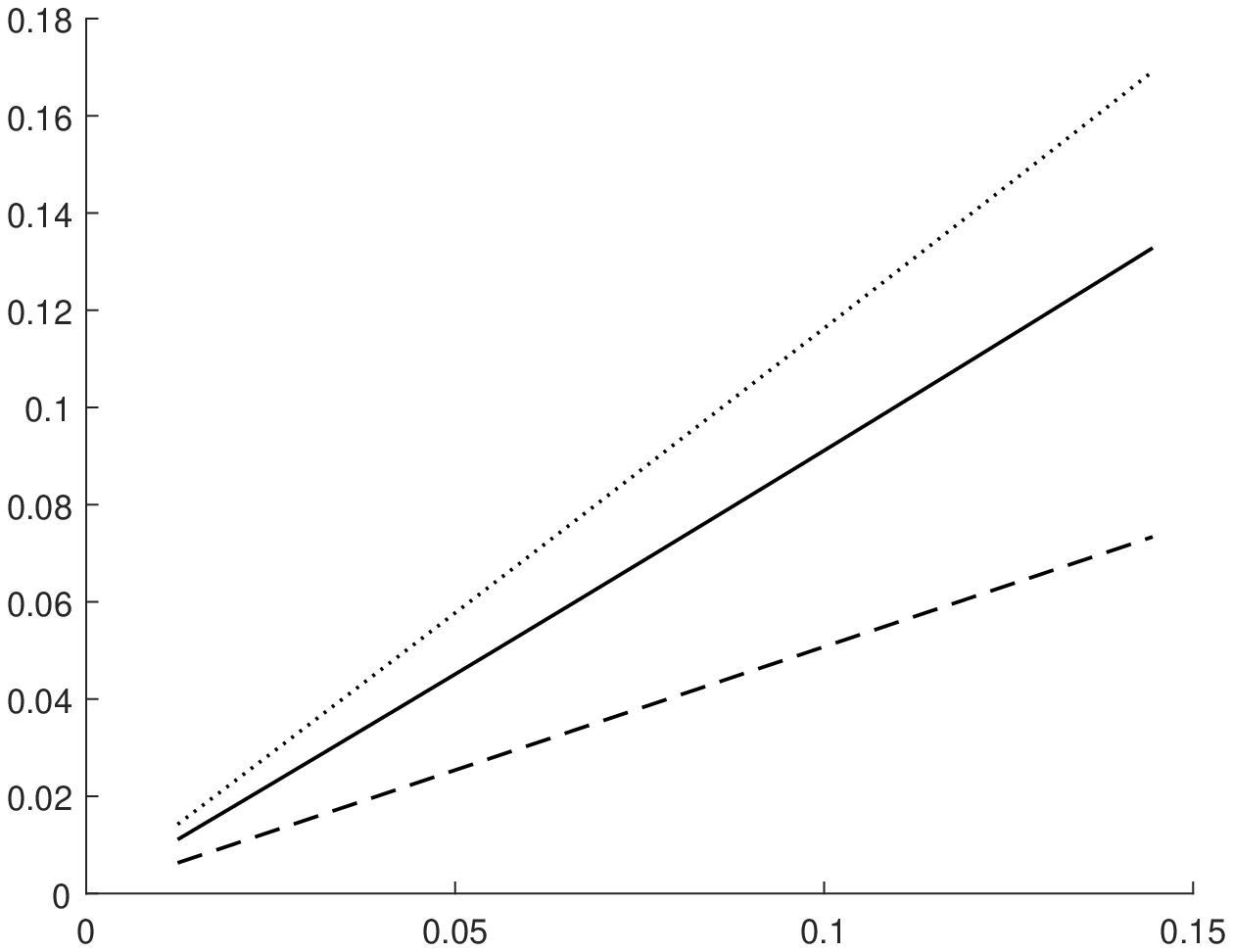,height=4.5cm,width=6.5cm}\qquad \epsfig{file=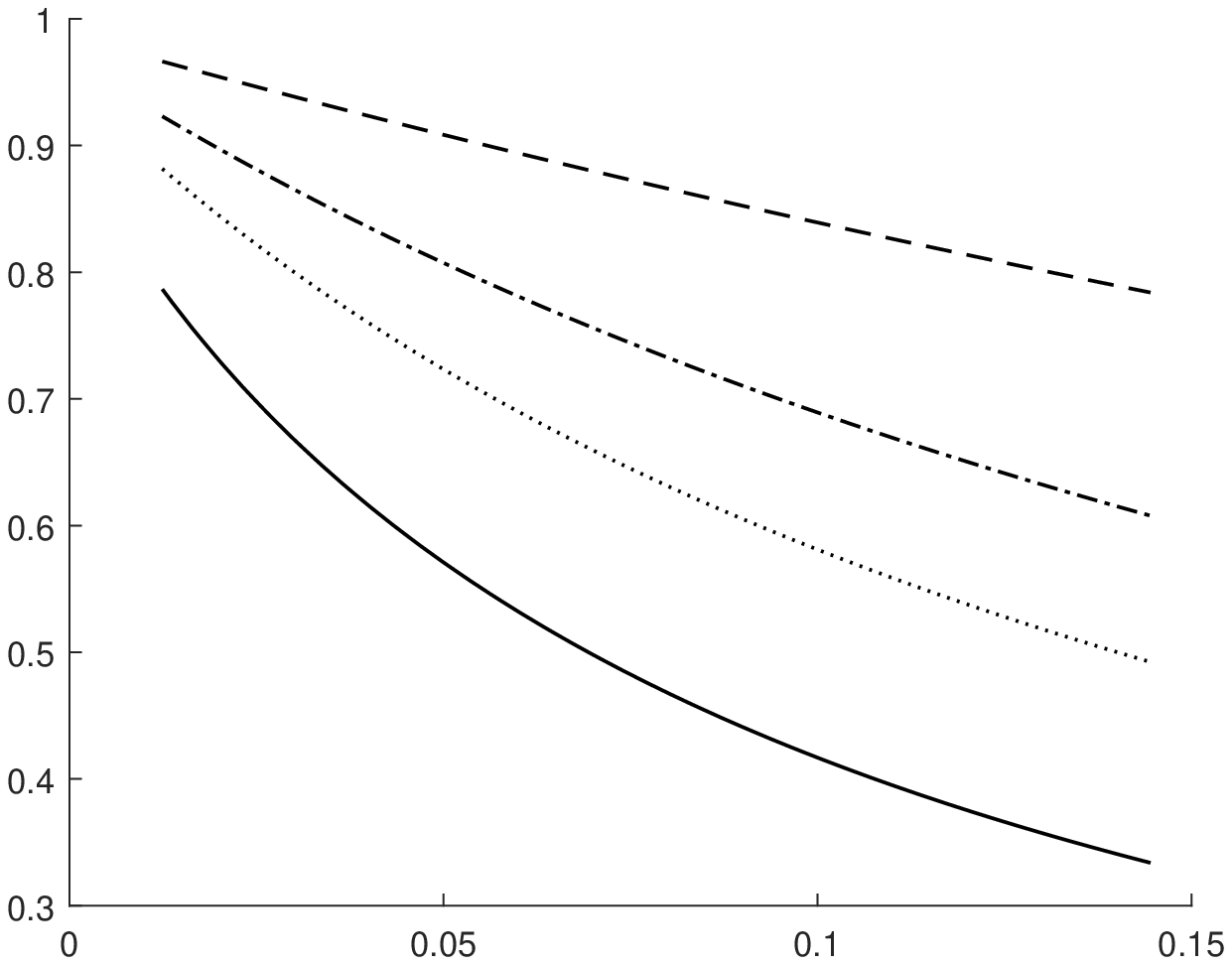,height=4.5cm,width=6.5cm}
\caption{The dynamic default insurance price (left plot) and defaultable bond indifference price (right)  at time $0$ as a function of the underlying state variable. The parameters are as in \eqref{E:CIR_num_param} and the horizon is $T=1$ year. In the insurance plot, the solid line corresponds to $f$ from \eqref{E:default_insurance}. The dotted line is the upper bound $\gamma^{\hat{\qprob}}$ from Proposition \ref{P:default_insurance}, and the dashed line is the default intensity $\gamma$ under $\prob$ (which is also the intensity under the minimal martingale measure). In the indifference price plot, the prices are given for $q=1$ (dash), $q=3$ (dot-dash), $q=5$ (dot) and $q=10$ (solid) notional. The state variable ranges from the 2.5\% to 97.5\% quantiles of the invariant distribution for $X$.} \label{F:def_ins_indiff_px}
\end{figure}

\section{Proof of Theorem \ref{T:main_result}}\label{S:proof}

The proof of Theorem \ref{T:main_result} is lengthy due to the facts that we are working in general domains, not assuming uniform ellipticity of the factor process, and not restricting the model coefficients to be bounded.   The outline for proving Theorem \ref{T:main_result} consists of the following steps:

\begin{enumerate}[(1)]
\item Identify a local version of the PDE in \eqref{E:G_PDE}, where we are able to use the theory of semi-linear elliptic equations with quadratic growth in both the solution and its gradient, to prove existence of solutions with bounded gradient.
\item Associate to the local PDE a local optimal investment problem, and show that any solution to the local PDE is the certainty equivalent of the local optimal investment problem, and hence solutions are unique.
\item Show that solutions to the local PDE are locally uniformly bounded, and hence there exists a solution to the full PDE.
\item Show that this solution to the full PDE is the certainty equivalent to the full optimal investment problem, and identify the optimal trading strategy and equivalent local martingale measure.
\end{enumerate}

Throughout this section, as the function $\phi$ is fixed we write $G = G(\cdot;\phi)$ for the PDE in \eqref{E:G_PDE}. Also, Assumptions \ref{A:region}, \ref{A:factor}, \ref{A:intensity}, \ref{A:asset_coeff} and \ref{A:phi} are in force.  As for Assumptions \ref{A:opt_main_ass_inc}, \ref{A:opt_main_ass_com}, they are only required for Propositions \ref{P:unwind_prob_A}, \ref{P:unwind_prob_B} below and their use will be made explicit.

\subsection{Mollifiers and function spaces}\label{SS:moll_notation}

We first introduce mollifiers in order to define the local PDE. To this end, we claim that without loss of generality we can re-index the sub-domains $E_n$ of Assumption \ref{A:region} so that for each $n$ there exists a function $\chi_n \in C^{\infty}(E;\reals)$ such that
\begin{equation}\label{E:moll_facts}
0\leq \chi_n \leq 1;\quad \chi_n = 1 \textrm{ on } E_{n-1};\quad \chi_n \textrm{ is supported on } \ol{E}_n;\quad \chi_n > 0 \textrm{ on }E_n.
\end{equation}
Indeed, set $\eps_n = (1/3)\textrm{dist}(\partial E_n, \partial E_{n-1})$. Then set $E'_n = \cbra{x\in E \such \textrm{dist}(x,\ol{E}_{n-1}) < \eps_n}$ and $\tilde{E}_n = \cbra{x \in E \such \textrm{dist}(x,\ol{E}_{n-1}) < 2\eps_n}$.  Lastly, set $\chi_n = \eta_{\eps_n} \star 1_{E'_n}$ where $\eta_{\eps}$ is the standard mollifier (see \cite[Appendix C]{MR1625845}). Then $\chi_n\in C^{\infty}(E;\reals)$ with $0\leq \chi_n \leq 1$.  Also, we have $\chi_n = 1$ on $E_{n-1}\supseteq \tilde{E}_{n-1}$, $\chi_n$ is supported in $\ol{\tilde{E}}_n$ and $\chi_n > 0$ on $\tilde{E}_n$.  So, \eqref{E:moll_facts} is satisfied on $\tilde{E}_n$ which also satisfies Assumption \ref{A:region}.  Thus, we can relabel $E_n = \tilde{E}_n$.

Next, we introduce the function spaces where our PDE solutions will lie. We use the notation of \cite[Chapters 1.3,4.1]{MR1465184}.  Namely, let $Q$ denote a region in $\reals^{1+m}$ and write $\ol{X} = (t,x), t\in\reals, x\in \reals^m$ for a typical point in $Q$. The parabolic distance between $\ol{X_1},\ol{X_2}$ is given by $\rho(\ol{X}_1,\ol{X_2}) \dfn \max\cbra{|t_1-t_2|^{1/2}, |x_1-x_2|}$, and for a given function $f$ on $Q$ and $\beta\in (0,1]$ we define
\begin{equation}\label{E:parabolic_holder_norms}
\begin{split}
|f|_{0,Q} &\dfn \sup_{\ol{X}\in Q}|f(\ol{X})|;\\
[f]_{\beta,Q} &\dfn \sup_{\ol{X}_1,\ol{X}_2\in Q,\ol{X}_1\neq \ol{X_2}} \frac{|f(\ol{X}_1)-f(\ol{X}_2))|}{\rho(\ol{X}_1,\ol{X}_2)^\beta};\\
\langle f \rangle_{\beta,Q} &\dfn \sup_{\ol{X}_1\in Q}\ \ \sup_{t_2\neq t_1, (t_2,x_1)\in Q} \frac{|f(t_1,x)-f(t_2,x)|}{|t_1-t_2|^{\beta/2}}.
\end{split}
\end{equation}
Next, for a given non-negative integer $k$, define the $|\cdot|_{k+\beta}$ norm via
\begin{equation}\label{E:parabolic_holder_k_beta}
|f|_{k+\beta, Q} \dfn \sum_{|\alpha| + 2j = k} [D^\alpha_x D^j_t f]_{\beta,Q} + \sum_{|\alpha| + 2j = k-1} \langle D^\alpha_x D^j_t f\rangle_{1+\beta,Q} + \sum_{|\alpha| + 2j\leq k} |D^\alpha_x D^j_t f|_{0,Q}.
\end{equation}
Here $a$ is a multi-index with norm $|a|$. $D^a_x$ is the derivative with respect to $x$ determined by $a$ and $D^j_t$ is the $j^{th}$ derivative with respect to $t$.  The parabolic H\"{o}lder space $H_{k+\beta,Q}$ is the Banach space of all functions $f$ on $Q$ with $|f|_{k+\beta,Q} < \infty$.  When $Q$ takes the special form $Q=(0,T)\times E$, the space $H_{k+\beta,Q,\textrm{loc}}$ is the set of functions $f$ which are in $H_{k+\beta,(0,T)\times K}$ for all bounded regions $K$ with $\ol{K}\subset E$. We pay special attention to when $k=2$ so that
\begin{equation}\label{E:parabolic_holder_2_beta}
|f|_{2+\beta, Q} = \sum_{|\alpha| + 2j = 2} [D^\alpha_x D^j_t f]_{\beta,Q} + \sum_{|\alpha| = 1} \langle D^\alpha_x D^j_t f\rangle_{1+\beta,Q} + \sum_{|\alpha| + 2j\leq 2} |D^\alpha_x D^j_t f|_{0,Q},
\end{equation}
and $k=0$ where $|f|_{0+\beta, Q}  = |f|_{\beta,Q} = [f]_{\beta,Q} + |f|_{0,Q}$.

\subsection{The local PDE and optimal investment problem}\label{SS:local}

With the mollifiers $\chi_n$ in place, consider a localized version of \eqref{E:G_PDE} on $(0,T)\times E_n$:
\begin{equation}\label{E:G_PDE_local}
\begin{split}
0 & = G^n_t + LG^n - \frac{\alpha}{2}\nabla (G^n)'A\nabla G^n + \frac{\sigma^2\chi_n}{2\alpha}\left(\frac{2\gamma}{\sigma^2} + \left(\frac{\mu}{\sigma^2} - \frac{\alpha}{\sigma}\nabla (G^n)'a\rho\right)^2 -\theta_{G^n}^2 - 2\theta_{G^n}\right);\\
\chi_n\phi & = G^n(T,\cdot).
\end{split}
\end{equation}

To conform to the notation in \cite{MR1465184} we reverse time, defining $v^n(t,x)\dfn G^n(T-t,x)$, $\Omega_n \dfn (0,T)\times E_n$, and $\Gamma_n$ as the parabolic boundary of $\Omega_n$.  Additionally, we write the PDE for $v^n$ as
\begin{equation}\label{E:lieb_PDE_form}
\begin{split}
0 &= Pv^n \dfn -v^n_t + \frac{1}{2}\tr\left(AD^2 v^n\right) + \check{a}^n(x,v^n, \nabla v^n);\ (t,x) \in \Omega_n\\
\chi_n\phi & = v^n;\ (t,x)\in \Gamma_n.
\end{split}
\end{equation}
In the above we have defined $(\chi_n\phi)(t,x) \dfn \chi_n(x)\phi(x)$ for $(t,x)\in \Gamma_n$.\footnote{Comparing with \cite[Equation (12.2)]{MR1465184} we have $a^{ij} = (1/2)A^{ij}$ and $a=\check{a}^n$.} Also, for $x\in E_n,z\in\reals,p\in\reals^d$ we have set (recall $\theta$ from \eqref{E:theta_def}):
\begin{equation}\label{E:lieb_notation}
\begin{split}
\theta(x,z,p) & \dfn \theta\left(\frac{\gamma(x)}{\sigma^2(x)}e^{\frac{\mu(x)}{\sigma^2(x)} +\alpha z - \frac{\alpha}{\sigma(x)}p'a(x)\rho(x)}\right);\\
\check{a}^n(x,z,p) & \dfn b(x)'p - \frac{\alpha}{2}p'A(x) p + \frac{\sigma^2(x)\chi_n(x)}{2\alpha}\left(\frac{2\gamma(x)}{\sigma^2(x)} + \left(\frac{\mu(x)}{\sigma^2(x)} - \frac{\alpha}{\sigma(x)}p'a(x)\rho(x)\right)^2\right)\\
&\quad - \frac{\sigma^2(x)\chi_n(x)}{2\alpha}\left(\theta\left(x,z,p\right)^2 - 2\theta\left(x,z,p\right)\right).
\end{split}
\end{equation}

With this notation, the following is an almost immediate consequence of \cite[Theorem 12.16]{MR1465184}.

\begin{proposition}\label{P:local_pde_exist}

There exists a solution $v^n\in H_{2+\beta,\Omega_n}$ to \eqref{E:lieb_PDE_form} and hence a solution $G^n\in H_{2+\beta,\Omega_n}$ to \eqref{E:G_PDE_local}.

\end{proposition}

\begin{proof}[Proof of Proposition \ref{P:local_pde_exist}]

The result will follow from \cite[Theorem 12.16]{MR1465184} once the requisite hypothesis are met.  To this end, the fact that $\partial E_n$ is $C^{1,\beta}$ implies that $\Gamma_n \in H_{1+\beta}$.  Furthermore, that $\phi\in C^{2,\beta}(E;\reals)$ and $\chi_n\in C^{\infty}(E;\reals)$ implies  $\chi_n\phi\in H_{1+\beta,\Omega_n}$.  Next, Lemma \ref{L:lieb_a_bounds} below implies \cite[Equation (12.26)]{MR1465184} and the local ellipticity of $A$ yields \cite[Equation (12.25a)]{MR1465184}. $A\in C^{2,\beta}(E;\mathbb{S}^d_{++})$ implies $A^{ij} \in H_{1,K}$ (i.e. is Lipschitz) for all bounded subsets $K$ of $E_n$.  The coefficient regularity also implies  $\check{a}^n(x,z,p)$ is in $H_{\beta,K}$ for all bounded subsets $K$ of $E_n\times \reals\times\reals^d$.

Regarding \cite[Equation (12.27)]{MR1465184}, first note \cite[Equation (12.26)]{MR1465184}, which follows from Lemma \ref{L:lieb_a_bounds}, implies an a-priori maximum principal for solutions to the PDE in \eqref{E:lieb_PDE_form}. Indeed, using Lemma \ref{L:lieb_a_bounds} in conjunction with \cite[Theorem 9.5]{MR1465184} applied to both $v^n,-v^n$ it follows that any solution $v^n$ to \eqref{E:lieb_PDE_form} satisfies the bound $\sup_{\Omega_n} |v^n| \leq e^{(1+C(n))T}\left(\sup_{\ol{E}_n}|\phi| + C(n)^{1/2}\right)$. Thus, any solution $v^n$ lies in a compact interval $[z_1,z_2]$ of $\reals$.  In view of this, Lemma \ref{L:lieb_a_Op} below shows that $|\check{a}^n|$ is on the order of $|p|^2$.  It is also clear that $A^{ij}_z = A^{ij}_p = 0$ and that $A^{ij}_x$ is independent of $p$ and hence on the order of $1$.  Therefore the resultant hypotheses of \cite[Theorem 12.16]{MR1465184} are met and there exists a solution $v^n$ to the PDE in \eqref{E:lieb_PDE_form}.   Now, \cite[Theorem 12.16]{MR1465184} yields a solution $v^n \in H^{-1-\beta}_{2+\beta}$ which is defined in \cite[Chapter 4]{MR1465184}.  However, since in \eqref{E:lieb_PDE_form} it follows that $\chi_n\phi$ (the boundary term) satisfies the \emph{compatibility condition of the first order}:
\begin{equation*}
P\chi_n\phi = 0;\qquad \cbra{0}\times \partial E_n,
\end{equation*}
since $\chi_n$ vanishes on $\partial E_n$.  Thus, as remarked at the end of \cite[Theorem 12.16]{MR1465184} (c.f. \cite[Theorems 5.14,8.2]{MR1465184}), it follows that $v^n \in H_{2+\beta,\Omega_n}$.  This gives the result.

\end{proof}

\begin{remark}\label{R:gradient_bound}

We record that $|G^{n}|_{2+\beta,\Omega_n} < \infty$ implies $\sup_{0\leq t\leq T, x\in \ol{E}_n} |\nabla G^n(t,x)| \leq C(n)$ for some constant $C(n)$. This will be used in the proof of Proposition \ref{P:opt_invest_local} below.

\end{remark}

We next show $G^n$ is the certainty equivalent for a localized version of the optimal investment problem in Section \ref{SS:opt_invest}.  Indeed, fix $0\leq t\leq T, x\in E$ and consider $n$ large enough so that $x\in E_n$.  The factor process $X = X^{t,x}$ is the same as in Remark \ref{R:mart_prob}.   Next, define a localized default time $\delta^n$ via
\begin{equation}\label{E:delta_n}
\delta^n \dfn \inf\cbra{s\geq t \such \int_t^s \left(\chi_n\gamma\right)(X_u)du = -\log(U)},
\end{equation}
and the localized default indicator process and its compensator via $H^n_s \dfn 1_{\delta^n \leq s}$ and $M^n_s \dfn H^n_s - \int_t^{s\wedge\delta^n}\left(\chi_n\gamma\right)(X_u)du$ for $s\geq t$. Set $\filtg^n$ in a similar manner to \eqref{E:G_def} and note that $M^n$ is a $\filtg^n$ martingale (c.f. \cite[Theorems 1.51, 4.48]{MR2273672}. The asset price process $S^n$ defined by $S^n_t = S_0$ and for $s\geq t$:
\begin{equation}\label{E:S_n_dynamics}
\frac{dS^n_s}{S^n_{s-}}= 1_{s\leq \delta^n}\left(\left(\chi_n(\mu-\gamma)\right)(X_u) du + \left(\chi_n\sigma\rho\right)(X_u)'dW_u + \left(\sqrt{\chi_n}\sigma\sqrt{1-\chi_n\rho'\rho}\right)(X_u)dW^0_u\right) - dM^n_s.
\end{equation}
Having localized the default intensity and asset dynamics, we next localize the optimal investment problem to stop when $X$ exits $E_n$.  To this end define
\begin{equation}\label{E:tau_n}
\tau^n \dfn \inf\cbra{s\geq t \such X_s \in \partial E_n}.
\end{equation}
Set $\M^n$ as the class of equivalent local martingale measures on $\G^n_{T\wedge\tau^n}$ (these are the measures so that $S^n$ stopped at $\tau^n$ is a $\filtg^n$ local martingale on $[t,T]$) and let $\tM^n$ denote the subset with finite relative entropy with respect to $\prob$ on $\G^n_{T\wedge\tau^n}$. Denote by $\mathcal{A}^n$ the class of $\filtg^n$ predictable trading strategies $\pi^n$ so that $\pi^n_\cdot 1_{\cdot\leq \tau^n}/S^n_{\cdot-}$ is $S^n$ integrable on $[t,T]$ and such that the resultant wealth process $\We^{\pi^n}$ is a $\qprob^n$ super-martingale for all $\qprob^n\in \tM^n$. Here, $\We^{\pi^n}$ has dynamics for $s\geq t$:
\begin{equation}\label{E:wealth_dynamics_n}
\begin{split}
d\We^{\pi^n,w}_s &= 1_{s\leq \tau^n} \pi^n_s\frac{dS^n_s}{S^n_{s-}};\\
&= 1_{s\leq\tau^n\wedge\delta^n}\pi^n_s\left(\left(\chi_n(\mu-\gamma)\right)(X_u) du + \left(\chi_n\sigma\rho\right)(X_u)'dW_u + \left(\sqrt{\chi_n}\sigma\sqrt{1-\chi_n\rho'\rho}(X_u)\right)dW^0_u\right)\\
&\qquad - 1_{s\leq \tau^n}\pi^n_s dM^n_s,
\end{split}
\end{equation}

For the starting point $(t,x)$, the localized optimal investment problem is
\begin{equation}\label{E:util_funct_n}
u^n(t,x) \dfn \sup_{\pi^n\in\mathcal{A}^n} \espalt{}{-e^{-\alpha\left(\We^{\pi^n}_{T\wedge\tau^n} + 1_{\delta^n > T\wedge\tau^n} \left(\chi_n\phi\right)(X_{T\wedge\tau^n})\right)}}.
\end{equation}

Before identifying the certainty equivalent $-(1/\alpha)\log(-u^n(t,x))$ with $G^n(t,x)$ from Proposition \ref{P:local_pde_exist}, we present two supplementary results concerning the structure of the local martingale measures in this setting. The first result is given in \cite[Proposition 5.3.1]{bielecki2004credit}:

\begin{lemma}\label{L:BR_Girsanov}
For any measure $\qprob^n\sim \prob$ on $\G^n_{T\wedge\tau^n}$ there is the representation
\begin{equation*}
\frac{d\qprob}{d\prob}\bigg|_{\G^n_{T\wedge\tau^n}} = \EN\left(\int_0^\cdot (A^n_u)'dW_u + \int_0^\cdot B^n_u dW^0_u + \int_0^\cdot C^n_u dM^n_u\right)_{T\wedge\tau^n},
\end{equation*}
where $A^n,B^n$ and $C^n$ are $\filtg^n$ predictable processes.  Additionally, $(W^{\qprob^n},W^{0,\qprob^n})$ is a $\qprob^n$ Brownian motion stopped at $T\wedge\tau^n$ where $W^{\qprob^n}_\cdot \dfn W_{\cdot} - \int_0^{\cdot} A^n_u du$, $W^{0,\qprob^n}_\cdot \dfn W^0_{\cdot} - \int_0^{\cdot} B^n_u du$, and $M^{\qprob^n}_\cdot\dfn M^n_{\cdot} - \int_t^{\cdot\wedge\delta^n} \left(\chi_n\gamma\right)(X_u)C^n_udu$  is a $\qprob^n$ martingale stopped at $T\wedge\tau^n$.
\end{lemma}

The second lemma characterizes when a measure $\qprob^n\sim\prob$ on $\G^n_{T\wedge\tau^n}$ is in $\M^n$:

\begin{lemma}\label{L:mkt_px_of_risk_n}
Let $\qprob^n\sim\prob$ on $\G^n_{T\wedge\tau^n}$, and let $A^n,B^n,C^n$ be as in Lemma \ref{L:BR_Girsanov}.  Then $\qprob^n\in\M^n$ if and only if for $\prob\times\textrm{leb}_{[t,T]}$ almost all $(\omega,u)$:
\begin{equation}\label{E:mkt_px_of_risk_n}
1_{t\leq u\leq \delta^n\wedge\tau^n\wedge T}\left( \left(\chi_n(\mu-\gamma)\right)(X_u) + \left(\chi_n\sigma\rho\right)(X_u)'A^n_u + \left(\sqrt{\chi_n}\sigma\sqrt{1-\chi_n\rho'\rho}\right)(X_u)B^n_u - \left(\chi_n\gamma(X_u)\right)C^n_u\right) = 0.
\end{equation}
\end{lemma}

\begin{proof}[Proof of Lemma \ref{L:mkt_px_of_risk_n}]
Using the dynamics for $S^n$ in \eqref{E:S_n_dynamics} in conjunction with Lemma \ref{L:BR_Girsanov} it follows that under $\qprob^n\in\M^n$ the asset $S^n$ has dynamics on $[t,T\wedge\tau^n]$:
\begin{equation}\label{E:S_n_dynamics_q}
\begin{split}
\frac{dS^n_u}{S^n_{u-}} &= 1_{u\leq \delta^n}\left(\left(\chi_n(\mu-\gamma)\right)(X_u) + \left(\chi_n\sigma\rho\right)(X_u)'A^n_u + \left(\sqrt{\chi_n}\sigma\sqrt{1-\chi_n\rho'\rho}\right)(X_u)B^n_u - \left(\chi_n\gamma(X_u)\right)C^n_u\right)du\\
& + 1_{u\leq \delta^n}\left(\left(\chi_n\sigma\rho\right)(X_u)'dW^{\qprob^n}_u + \left(\sqrt{\chi_n}\sigma\sqrt{1-\chi_n\rho'\rho}\right)(X_u)dW^{0,\qprob^n}_u\right) - dM^{\qprob^n}_u.
\end{split}
\end{equation}
Since continuous local martingales with finite variation paths must be constant (c.f. \cite[Ch IV, Prop (1.2)]{MR1725357}) the result follows.
\end{proof}

A heuristic use of the dynamic programming principle shows that the PDE for the certainty equivalent to $u^n$ is the same as in \eqref{E:G_PDE_local}.  The following proposition shows that indeed, $G^n$ from Proposition \ref{P:local_pde_exist} is the certainty equivalent.

\begin{proposition}\label{P:opt_invest_local}

There is a unique solution $G^n\in H_{2+\beta,\Omega_n}$ to the PDE in \eqref{E:G_PDE_local} which takes the form $G^n(t,x) = -(1/\alpha)\log\left(-u^n(t,x)\right)$, for $u^n$ defined in \eqref{E:util_funct_n}.  For the localized optimal investment problem, the optimal trading strategy is given by
\begin{equation}\label{E:opt_pi_n}
\hpi^n_s  = \hpi^n(s,X_s) = \hpi^n(s,X^{t,x}_s);\quad \hpi^n \dfn \frac{1}{\alpha}\left(\frac{\mu}{\sigma^2} - \frac{\alpha}{\sigma}(\nabla G^n(\cdot))'a\rho - \theta_{G^n(\cdot)}\right);\quad t\leq s\leq T\wedge\tau_n.
\end{equation}
The optimal martingale density process is given by $\hz^n = \hz^{n,(t,x)}$ where
\begin{equation}\label{E:hat_Z_n}
\hz^n_s = e^{-\alpha\left(\We^{\hpi^n}_s - G^n(t,x) + 1_{s\wedge\tau^n < \sigma^n} G^n(s\wedge\tau^n,X_{s\wedge\tau^n}) \right)};\quad t\leq s\leq T\wedge\tau^n.
\end{equation}

\end{proposition}

\begin{proof}[Proof of Proposition \ref{P:opt_invest_local}]
Write $\hwe^n = \We^{\hpi^n}$ and note that below, $C(n)$ is a constant which may change from line to line. Since $G^n$ solves \eqref{E:G_PDE_local}  at $T\wedge\tau^n$ we have
\begin{equation*}
\begin{split}
\hz^n_{T\wedge\tau^n} &= e^{\alpha G^n(t,x)} \times e^{-\alpha\left(\hwe^n_{T\wedge\tau^n} + 1_{T\wedge\tau^n< \delta^n}(\chi_n\phi)(X_{T\wedge\tau^n})\right)},
\end{split}
\end{equation*}
so $\hwe^n, \hz^n$ satisfy the first order conditions for optimality. From the well known utility maximization results for exponential utility (see \cite{MR1891730,MR1891731,MR1743972,MR1920099}) the result will follow provided
\begin{enumerate}[(1)]
\item $\hwe^n$ is a $\qprob^n$ super-martingale for all $\qprob^n\in\tM^n$.
\item $\hz^n_{T\wedge\tau^n} = d\hqprob^n/d\prob \big|_{\G^n_{T\wedge\tau^n}}$ for some $\hqprob^n \in \tM^n$.
\item $\hwe^n$ is a $\hqprob^n$ martingale.
\end{enumerate}

Parts $(1)$ and $(3)$ follow immediately from Lemmas \ref{L:BR_Girsanov}, \ref{L:mkt_px_of_risk_n} and Remark \ref{R:gradient_bound}.  Indeed, using \eqref{E:S_n_dynamics_q} we have
\begin{equation*}
\begin{split}
d\hwe^n_s &= 1_{u\leq \tau^n\wedge\delta^n}\hpi^n(u,X_u)\left(\left(\chi_n\sigma\rho\right)(X_u)'dW^{\qprob^n}_u + \left(\sqrt{\chi_n}\sigma\sqrt{1-\chi_n\rho'\rho}\right)(X_u)dW^{0,\qprob^n}_u\right)\\
&\qquad  - 1_{u\leq\tau^n}\hpi^n(u,X_u)dM^{\qprob^n}_u.
\end{split}
\end{equation*}
Thus,
\begin{equation*}
\left[\hwe^n, \hwe^n\right]_{T\wedge\tau^n} = \int_t^{T\wedge\tau^n\wedge\delta^n}\left(\hpi^n(u,X_u)\right)^2 \chi_n\sigma^2(X_u)du + 1_{\delta^n\leq T\wedge\tau^n}\left(\hpi^n(\delta^n,X_{\delta^n})\right)^2.
\end{equation*}
Since $|\hpi^n|\leq C(n)$ on $[t,T]\times \ol{E}_n$ (see Remark \ref{R:gradient_bound}) and $\sigma,\chi_n$ are bounded on $\ol{E}_n$ it follows that $\qprob^n$ almost surely $[\hwe^n]_{T\wedge\tau^n} \leq C(n)$, and hence $\hwe^n$ is a $\qprob^n$ martingale (c.f. \cite[Theorems 1.51,4.48]{MR2273672}).  This gives $(1)$ and $(3)$ as well, provided we can show $(2)$. To this end, a lengthy but straight-forward calculation using \ito's formula in Lemma \ref{L:ITO_result} below shows $\hz^n$ has dynamics:
\begin{equation}\label{E:Z_n_form}
\begin{split}
&\frac{d\hz^n_s}{\hz^n_{s-}} = 1_{s\leq\tau^n\wedge\delta^n}\left((A^n_s)'dW_s + B^n_s dW^0_s\right) + 1_{s\leq \tau^n}C^n_s dM^n_s;\\
&\qquad A^n_s = -\alpha \left(\hpi^n\chi_n\sigma \rho + a\nabla G^n\right)(s,X_s);\quad B^n_s = -\alpha \left(\hpi^n\sqrt{\chi_n}\sigma\sqrt{1-\chi_n\rho'\rho}\right)(s,X_s);\\
&\qquad C^n_s = \left(e^{\alpha\left(\hpi^n + G^n\right)}-1\right)(x,X_s).
\end{split}
\end{equation}
Remark \ref{R:gradient_bound} and \eqref{E:opt_pi_n} show that almost surely $|A^n_s|, |B^n_s|, |C^n_s|\leq C(n)$. Also, note that $A^n, B^n, C^n$ are $\filt^W$ predictable, and there is some $\eps_n>0$ so that $C^n_s > -(1-\eps_n)$. Therefore, by \cite[Theorem 9]{MR2574236} it follows that $\hz^n$ is a strictly positive martingale\footnote{Note that $\int_t^\cdot1_{s\leq \tau^n}C^n_s dM^n_s$ is a martingale since it has bounded quadratic variation: see \cite[Theorems 1.51 and 4.48]{MR2273672}. Clearly, the Brownian stochastic integrals are martingales.}. It remains to show that $\hqprob^n$ defined by $\hz^n_{T\wedge\tau^n}$ is in $\tM^n$. First, that $\hqprob^n\in\M^n$ follows immediately from \eqref{E:functional_mkt_px_of_risk_n} in Lemma \ref{L:ITO_result} below in conjunction with Lemma \ref{L:mkt_px_of_risk_n}. For the relative entropy, note that by definition of $\hz^n$ in \eqref{E:hat_Z_n}, and \eqref{E:G_PDE_local} we obtain
\begin{equation*}
\hz^n_{T\wedge\tau^n}\log\left(\hz^n_{T\wedge\tau^n}\right) = -\alpha \hz^n_{T\wedge\tau^n}\hwe^n_{T\wedge\tau^n} - \alpha \hz^n_{T\wedge\tau^n} 1_{T\wedge\tau^n<\delta^n}(\chi_n\phi)(X_{T\wedge\tau^n}) + \alpha \hz^n_{T\wedge\tau^n}G^n(t,x).
\end{equation*}
Since $0\leq\chi_n\leq 1$, $\phi$ is bounded, $\hz^n$ is a martingale, and $G^n$ is deterministic, the fact that $\hz^n\hwe^n$ is a martingale implies the desired result.  Thus, $\espalt{}{\hz^n_{T\wedge\tau^n}\log\left(\hz^n_{T\wedge\tau^n}\right)} < \infty$ so that $\hqprob^n\in\tM^n$.

\end{proof}

\subsection{Unwinding the localization: analytic results}\label{SS:unwind_anal}

We now provide two analytic results for unwinding the localization.  The first uses the maximum principal to obtain a uniform lower bound for solutions $G^n$ to \eqref{E:G_PDE_local}.  The second proves existence of solutions to \eqref{E:G_PDE} provided locally uniform upper bounds for $G^n$.

\begin{proposition}\label{P:G_n_global_min}

Let $G^n$ denote the unique solution to \eqref{E:G_PDE_local} from Propositions \ref{P:local_pde_exist} and \ref{P:opt_invest_local}.  Recall the definition of $\ul{\phi}$ from Assumption \ref{A:phi}.  Then, for each $n$, $\inf_{0\leq t\leq T, x\in\ol{E}_n} G^n(t,x) \geq \ul{\phi}$.
\end{proposition}

\begin{proof}[Proof of Proposition \ref{P:G_n_global_min}]
First, assume $G^n$ has a minimum in $[0,T)\times E_n$ at $(t_0,x_0)$. If $t_0>0$ then $G^n_t(t_0,x_0) = 0$.  If $t_0 = 0$ then $G^n_t(t_0,x_0)\geq 0$.  Also, $\nabla G^n(t_0,x_0) = 0$.  By the ellipticity of $A$ in $E_n$, \eqref{E:G_PDE_local} implies at $(t_0,x_0)$:
\begin{equation*}
\begin{split}
0 &\geq \frac{\chi_n\sigma^2}{2\alpha}\left(\frac{2\gamma}{\sigma^2} + \frac{\mu^2}{\sigma^4} - \theta^2\left(\frac{\gamma}{\sigma^2}e^{\frac{\mu}{\sigma^2} + \alpha G^n}\right) - 2\theta\left(\frac{\gamma}{\sigma^2}e^{\frac{\mu}{\sigma^2} + \alpha G^n}\right)\right);\\
&\geq \frac{\chi_n\sigma^2}{2\alpha}\left(\frac{2\gamma}{\sigma^2} + \frac{\mu^2}{\sigma^4} - \frac{2\gamma}{\sigma^2}e^{\alpha G^n} - \frac{\mu^2}{\sigma^4}\right);\\
&=\frac{\chi_n\gamma}{\alpha}\left(1-e^{\alpha G^n}\right).
\end{split}
\end{equation*}
Above the second inequality uses Lemma \ref{L:theta_lem} below at $y=(\gamma/\sigma^2)e^{\alpha G^n}$ and $x=\mu/\sigma^2$. Since $x_0\in E_n$ and $\chi_n, \gamma > 0$ in $E_n$ we see $0 \geq 1 - e^{\alpha G^n}$ which implies $G^n \geq 0 \geq \ul{\phi}$. But, we already know $G^n(T,\cdot) = \chi_n\phi \geq \ul{\phi}$ and $G^n(t,\partial E_n) = 0 \geq \ul{\phi}$.  Thus, the result follows.

\end{proof}

The next proposition is significantly more involved. Though it can be deduced from \cite[Theorem V.8.1]{MR0241822}, for transparency we offer a detailed proof using the results in \cite{MR1465184, MR0181836}.

\begin{proposition}\label{P:G_n_bounds_give_soln}
Let $G^n$ be the unique solution to \eqref{E:G_PDE_local} from Propositions \ref{P:local_pde_exist} and \ref{P:opt_invest_local}.  Assume for each $k\in\nats$ that
\begin{equation}\label{E:G_n_loc_unif_ub}
\sup_{n\geq k+1}\sup_{0\leq t\leq T, x\in \ol{E}_k} G^n(t,x) = C(k) < \infty.
\end{equation}
Then there exists a solution $G$ to \eqref{E:G_PDE}. In particular, there is a subsequence (still labeled $n$) such that $G^n$ converges to $G$ in $H_{2+\beta,(0,T)\times E,\textrm{loc}}$.
\end{proposition}

\begin{proof}[Proof of Proposition \ref{P:G_n_bounds_give_soln}]
In what follows, $C(k)$ is a constant which depends on all model quantities in $[0,T]\times\ol{E}_k$, but may change from line to line. Note that Proposition \ref{P:G_n_global_min} and \eqref{E:G_n_loc_unif_ub} imply
\begin{equation}\label{E:G_n_loc_unif}
\sup_{n\geq k+1}\sup_{0\leq t\leq T, x\in \ol{E}_k} |G^n(t,x)| <\infty.
\end{equation}
\subsubsection*{Step 1:} Use \cite[Theorem 11.3(b)]{MR1465184} and \eqref{E:G_n_loc_unif} to conclude that
\begin{equation}\label{E:G_n_gradient_loc_unif}
\sup_{n\geq k+2}\sup_{0\leq t\leq T,x\in \ol{E}_k} |\nabla G^n(t,x)| < \infty.
\end{equation}
To show this, we recall that $G^n(t,x) = v^n(T-t,x)$ where $v^n$ solves the PDE in \eqref{E:lieb_PDE_form} with $\check{a}^n$ therein defined in \eqref{E:lieb_notation}.  Next, we define the Bernstein function (c.f. \cite[Equation (8.3)]{MR1465184}):
\begin{equation}\label{E:bernstein}
\EN(x,p) \dfn \frac{1}{2}p'A(x)p;\quad x\in E, p\in\reals^d.
\end{equation}
Note that by assumption we have that $\EN(x,p) \geq \lambda_k p'p > 0$ for $x\in \ol{E}_k$ and some $\lambda_k>0$.  Next, we define the differential operators from \cite[Chapter 11]{MR1465184} which act on functions $f(x,z,p)$ via
\begin{equation}\label{E:bernstein_delta}
\delta(p)[f](x,z,p) \dfn f_z(x,z,p) + \frac{1}{p'p} p'\nabla_x f(x,z,p);\quad \ol{\delta}(p)[f](x,z,p) \dfn p'\nabla_p f(x,z,p).
\end{equation}
For the domain $\Omega_k = (0,T)\times E_k$ the quantities $A,B,C$ from \cite[Equation (11.7)]{MR1465184} become\footnote{$v$ of \cite[Equation (11.7)]{MR1465184} is equal to $p'p$: see right after \cite[Equation (11.2)]{MR1465184}.}
\begin{equation}\label{E:lieb_A_B_C}
\begin{split}
A_k(x,z,p) &\dfn \frac{1}{\EN(x,p)}\left(\frac{p'p}{8\lambda_k}\sum_{i,j=1}^d \left(\ol{\delta}(p)[A^{ij}](x,z,p)\right)^2 + \left(\ol{\delta}(p)-1\right)[\EN](x,z,p)\right);\\
B_k(x,z,p) &\dfn \frac{1}{\EN(x,p)}\left(\delta(p)[\EN](x,z,p) + \left(\ol{\delta}(p)-1\right)[\check{a}^n](x,z,p)\right);\\
C_k(x,z,p) &\dfn \frac{1}{\EN(x,p)}\left(\frac{p'p}{8\lambda_k} \sum_{i,j=1}^d \left(\delta(p)[A^{ij}](x,z,p)\right)^2 + \delta(p)[\check{a}^n](x,z,p)\right).
\end{split}
\end{equation}
From Lemma \ref{L:lieb_A_B_C} below we see that the quantities $A^\infty_k,B^\infty_k,C^\infty_k$ of \eqref{E:lieb_A_B_C_inf} are finite with $A^\infty_k = 1$ and $C^\infty_k = 0$. Lastly, as noted right after \cite[Equation (11.4)]{MR1465184} we take $a^{ij}_{*} = a^{ij}, f_j =0$,
and use these values to verify \cite[Equations (11.17abc)]{MR1465184}.  After all these preparations, we are ready to invoke \cite[Theroem 11.3(b)]{MR1465184}.  First of all, note that \cite[Equations (11.17ac)]{MR1465184} hold.  As for \cite[Equation (11.17b)]{MR1465184}, take $\theta=1$ therein, and note that there is a constant $\Lambda_k>0$ so that $\EN(x,p)\leq \Lambda_k p'p$ in $\ol{E}_k$ (this is the $\Lambda$ of \cite[Equation (11.17b)]{MR1465184}). Next, define
\begin{equation}\label{E:lieb_D}
D_k(x,z,p) \dfn \frac{1}{\EN(x,p)}\left(\Lambda_k p'p + |p|\left(|\nabla_p \EN(x,p)| + |\nabla_p \check{a}^n(x,z,p)|\right)\right).
\end{equation}
From Lemma \ref{L:lieb_A_B_C} below we see the quantity $D^\infty_k$ of \eqref{E:lieb_D_inf} is finite and hence \cite[Equation (11.17b)]{MR1465184} holds for all $Q=(0,T)\times B(x_0,R)$ with $x_0 \in E_{k}$ and $R>0$ sufficiently small.  In fact, let $x_0\in E_{k-1}$ and set $R = \textrm{dist}(E_{k-1},\partial E_k)$.  For such $x_0,R$ we have \cite[Equation (11.17b)]{MR1465184} and as such, for $n\geq k+1$ we deduce from \cite[Theorem 11.3(b)]{MR1465184} that
\begin{equation*}
|\nabla G^n(t,x_0)| = |\nabla v^n(x_0,T-t)| \leq c_3\left(1+\left(\frac{\textrm{osc}_{\Omega_k} v^n}{R}\right)\right).
\end{equation*}
Here, $c_3$ depends on $A^k_\infty,B^k_\infty,C^k_\infty$ and on $\sup_{BQ}|\nabla v^n|$.  However, as defined in \cite[Chapter 2.1]{MR1465184}, $BQ = \cbra{0}\times B(x_0,R)$ so that $\sup_{BQ}|\nabla v^n| \leq \sup_{\ol{E}_k}|\nabla\phi|$, since for $n\geq k+1$ we have $v^n = \phi$ on $\cbra{0}\times \ol{E}_k$. Lastly, the quantity $\textrm{osc}_{\Omega_k} v^n$ is defined in \cite[Section 4.1]{MR1465184} but is bounded above in our case by
\begin{equation*}
\sup_{\Omega_k} v^n - \inf_{\Omega_k} v^n = \sup_{\Omega_k} G^n - \inf_{\Omega_k} G^n \leq C(k)-\ul{\phi},
\end{equation*}
where the last equality comes from \eqref{E:G_n_loc_unif_ub} and Proposition \ref{P:G_n_global_min}. Putting all this together and using the uniform continuity of $\nabla G^n$ on $\Omega_k$ we obtain $\sup_{n\geq k+1}\sup_{0\leq t\leq T,x\in\ol{E}_{k-1}}|\nabla G^n(t,x)| \leq C(k)$. If we re-index this by moving $k$ to $k-1$ we obtain \eqref{E:G_n_gradient_loc_unif}.

\subsubsection*{Step 2:} Use \eqref{E:G_n_gradient_loc_unif}  and \cite[Theorem 4, Ch 7, Section 2]{MR0181836} to show
\begin{equation}\label{E:G_n_holder_gradient_loc_unif}
\sup_{n\geq k+3}\left(\left[G^n\right]_{\beta,\Omega_k} + \left[\nabla G^n\right]_{\beta,\Omega_k}\right) \leq C(k),
\end{equation}
where $[\cdot]_{\beta,\Omega_k}$ is defined in \eqref{E:parabolic_holder_norms}. Indeed, fix $k\geq 2$ and $n\geq k+2$.  Now, instead of applying \cite[Theorem 4, Ch 7, Section 2]{MR0181836} directly to $v^n(t,x) = G^n(T-t,x)$ we will apply it to a truncated version of $v^n$ given by
\begin{equation}\label{E:truncated_v_n}
\check{v}^n\dfn \chi_k\left(v^n - \phi\right)(t,x);\quad (t,x)\in\Omega_k.
\end{equation}
Note that $\check{v}^n\in H_{2+\beta,\Omega_k}$ and that $\check{v}^n = 0$ on $\Gamma_k$. Using the PDE for $v^n$ in \eqref{E:lieb_PDE_form} it follows that $\check{v}^n$ solves the PDE
\begin{equation}\label{E:check_v_n_pde}
-\check{v}^n_t + \frac{1}{2}\tr\left(AD^2\check{v}^n\right) = f,
\end{equation}
where
\begin{equation}\label{E:holder_f_def}
f \dfn -\chi_k \check{a}^n + v^n\frac{1}{2}\tr\left(AD^2\chi_k\right) + (\nabla v^n)'A\nabla\chi_k - \frac{1}{2}\tr\left(AD^2(\chi_k\phi)\right).
\end{equation}
Note that $f$ vanishes at $\cbra{0}\times\partial E_k$, and moreover, since
\begin{enumerate}[(1)]
\item $|\chi_k|_{0,\Omega_k} + \sum_{i=1}^d |\partial_{x_i}\chi_k|_{0,\Omega_k} + \sum_{i,j=1}^d |\partial^2_{x_i,x_j}\chi_k|_{0,\Omega_k} \leq C(k)$ by construction.
\item $\ul{\phi} \leq v^n \leq C(k)$ by Proposition \ref{P:G_n_global_min} and \eqref{E:G_n_loc_unif_ub}.
\item $|\nabla v^n|_{0,\Omega_k} \leq C(k)$ by \eqref{E:G_n_gradient_loc_unif}.
\item $\phi\in C^{2,\beta}(E)$ hence $\phi\in H_{2+\beta,\Omega_k}$ by assumption,
\end{enumerate}
we know that
\begin{equation}\label{E:f_unif_bound}
|f|_{0,\Omega_k} \leq C(k).
\end{equation}
Thus, by \cite[Theorem 4, Ch 7, Section 2]{MR0181836} we obtain $[\check{v}^n]_{\beta,\Omega_k} + [\nabla\check{v}^n]_{\beta,\Omega_k} \leq C(k)|f|_{0,\Omega_k}$. Now, since $v^n = \phi + \check{v}^n$ in $\Omega_{k-1}$  the triangle inequality implies
\begin{equation*}
\begin{split}
[v^n]_{\beta,\Omega_{k-1}} + [\nabla v^n]_{\beta,\Omega_{k-1}} &\leq [\phi]_{\beta,\Omega_{k-1}} + [\nabla\phi]_{\beta,\Omega_{k-1}} + [\check{v}^n]_{\beta,\Omega_{k-1}} + [\nabla\check{v}^n]_{\beta,\Omega_{k-1}};\\
&\leq |\phi|_{2+\beta,\Omega_{k-1}} + C(k)|f|_{0,\Omega_{k}};\\
&\leq C(k).
\end{split}
\end{equation*}
This held for $n\geq k+2$. Replacing $k-1$ with $k$ we see that for $n\geq k+3$ we have
\begin{equation*}
[G^n]_{\beta,\Omega_k} + [\nabla G^n]_{\beta,\Omega_{k}} = [v^n]_{\beta,\Omega_k} + [\nabla v^n]_{\beta,\Omega_{k}} \leq C(k).
\end{equation*}
which is what we wanted to show.

\subsubsection*{Step 3}: Use \eqref{E:G_n_holder_gradient_loc_unif}  and \cite[Theorem (5.14)]{MR1465184} to show
\begin{equation}\label{E:G_n_holder_full_loc_unif}
\sup_{n\geq k+4} |G^n|_{2+\beta,\Omega_k} \leq C(k).
\end{equation}
Let $n\geq k+3$. We retain the notation $\check{v}^n$ from the previous step.  Since $\check{v}^n$ satisfies the linear parabolic PDE in \eqref{E:check_v_n_pde} in $\Omega_k$ with boundary condition $\check{v}^n = 0$ on $\Gamma_k$, it follows by the well-known existence results on linear parabolic PDE (see, for example, \cite[Theorem 5.14]{MR1465184}) that
\begin{equation*}
|\check{v}^n|_{2+\beta,\Omega_k} \leq C(k)|f|_{\beta,\Omega_k} = C(k)\left(|f|_{0,\Omega_k} + [f]_{\beta,\Omega_k}\right).
\end{equation*}
By \eqref{E:f_unif_bound} we already know that $|f|_{0,\Omega_k} \leq C(k)$.  However, from \eqref{E:holder_f_def}, \eqref{E:G_n_holder_gradient_loc_unif} and the regularity of the model coefficients and $\chi_k$ it is easily seen that $[f]_{\beta,\Omega_k} \leq C(k)$.  This yields that $|\check{v}^n|_{2+\beta,\Omega_k} \leq C(k)$. Since $\check{v}^n = v^n$ in $\Omega_{k-1}$ we see that
\begin{equation*}
|G^n|_{2+\beta,\Omega_{k-1}} = |v^n|_{2+\beta,\Omega_{k-1}} \leq |\check{v}^n|_{2+\beta,\Omega_k} \leq C(k).
\end{equation*}
Thus, replacing $k-1$ with $k$ it follows from $n\geq k+4$ that $|G^n|_{2+\beta,\Omega_k} \leq C(k)$, which is precisely \eqref{E:G_n_holder_full_loc_unif}.

\subsubsection*{Step 4:} Use \eqref{E:G_n_holder_full_loc_unif} to obtain the solution $G$ to \eqref{E:G_PDE}.  This proof is standard and follows a diagonal sub-sequence argument.  Indeed, fix an integer $k_0$.  By \eqref{E:G_n_holder_full_loc_unif} applied to $k=k_0$ we may extract a sub-sequence $G^{n_{k_0}}$ which converges in $|\cdot|_{2+\beta,\Omega_{k_0}}$ to a function $G^{k_0}$.  Clearly, $G^{k_0}$ solves \eqref{E:G_PDE} in $\Omega_{k_0}$ with $G^{k_0}(T,\cdot) = \phi$ on $E_{k_0}$. Then for $k=k_0+1$ we may take a further subsequence $G^{n_k}$ which converges in $|\cdot|_{2+\beta,\Omega_k}$ to a function $G^k$ satisfying the PDE in $\Omega_k$.  By construction $G^{k_0} = G^{k}$ in $\Omega_{k_0}$ and hence setting $G$ as this common function, $G$ is well defined in $\Omega_k$.  Repeating this process for $k_0,k_0+1,...$ we obtain a function $G$ which solves the full PDE in $(0,T)\times E$ with correct boundary condition.  It is clear that $G\in H_{2+\beta,\Omega_k}$ for each $k$ and hence $G\in H_{2+\beta,(0,T)\times E,\textrm{loc}}$.  This finished the proof.

\end{proof}

\subsection{Unwinding the localization: probabilistic results}\label{SS:unwind_prob}  We now provide probabilistic results complementary to the analytic ones in Section \ref{SS:unwind_anal}.  The first lemma is similar in spirit to Lemmas \ref{L:BR_Girsanov}, \ref{L:mkt_px_of_risk_n}, but contains an additional statement concerning the dual problem to \eqref{E:util_funct}.  The dual problem is (recall we are starting at $t\leq T$ and the factor process satisfies $X_t = x$):
\begin{equation}\label{E:dual_util_funct}
v(t,x;\phi)\dfn \inf_{\qprob\in\tM}\left(\espalt{}{Z^{\qprob}_T\log\left(Z^{\qprob}_T\right)} + \alpha\espalt{}{Z^{\qprob}_T 1_{\delta>T}\phi(X_T)}\right);\qquad Z^{\qprob}_T =  \frac{d\qprob}{d\prob}\bigg|_{\G_T}.
\end{equation}

\begin{lemma}\label{L:dual_opt_structure}
Assume that $\tM\neq\emptyset$ and hence there is a unique optimizer $\hat{\qprob}\in\tM$ to the right hand side of \eqref{E:dual_util_funct} (c.f., \cite[Theorem 1.1]{MR2489605}. Then $Z^{\hat{\qprob}}$ must be of the form for $t\leq s\leq T$:
\begin{equation}\label{E:Z_Q_structure}
Z^{\hat{\qprob}}_s = \EN\left(\int_t^{\cdot} 1_{u\leq T\wedge\delta} A_u'dW_u + \int_t^\cdot 1_{u\leq T\wedge\delta} B_u dW^0_u + \int_{t}^\cdot 1_{u\leq T\wedge\delta} C_u dM_u\right)_s,
\end{equation}
where $A, B, C$ are $\filt^{W,W^0}$ predictable processes such that
\begin{equation}\label{E:mkt_px_of_risk}
0 = (\mu-\gamma)(X_u) + (\sigma\rho)(X_u)'A_u + (\sigma\sqrt{1-\rho'\rho})(X_u)B_u - \gamma(X_u)C_u,
\end{equation}
for $\prob\times\textrm{leb}_{[t,T]}$ almost every $(\omega,u)$.
\end{lemma}

\begin{proof}[Proof of Lemma \ref{L:dual_opt_structure}]
The same arguments as in Lemmas \ref{L:BR_Girsanov}, \ref{L:mkt_px_of_risk_n} show that if $\qprob\in\M$ then for $t\leq s\leq T$:
\begin{equation*}
Z^{\qprob}_s = \EN\left(\int_t^{\cdot} \tilde{A}_u'dW_u + \int_t^\cdot \tilde{B}_u dW^0_u + \int_{t+}^\cdot \tilde{C}_u dM_u\right)_s,
\end{equation*}
where $\tilde{A},\tilde{B},\tilde{C}$ are $\filtg$ predictable and where for $\prob\times\textrm{leb}_{[t,T]}$ almost every $(\omega,s)$
\begin{equation*}
0 = 1_{u\leq T\wedge\delta}\left((\mu-\gamma)(X_u) + \sigma\rho(X_u)'\tilde{A}_u + \sigma\sqrt{1-\rho'\rho}(X_u)\tilde{B}_u - \gamma(X_u)\tilde{C}_u\right).
\end{equation*}
We first claim that for $\qprob\in\tM$, $\qprob$ cannot be optimal for the dual problem unless $\tilde{A},\tilde{B},\tilde{C}$ are all stopped at $T\wedge\delta$.  Indeed, $\qprob\in\tM$ implies that $Z^\qprob$ as above is a martingale and $\espalt{}{Z^\qprob_T\log\left(Z^\qprob_T\right)} < \infty$. Thus, $Z^\qprob \log(Z^\qprob)$ is a sub-martingale. As such, since $1_{\delta>T}\phi(X_T) = 1_{\delta>T}\phi(X_{T\wedge\delta})$ is in $\G_{T\wedge\delta}$ we have by the sub-martingale property and optional sampling theorem that
\begin{equation*}
\espalt{}{Z^{\qprob}_T\log\left(Z^{\qprob}_T\right)} + \alpha\espalt{}{Z^{\qprob}_T 1_{\delta>T}\phi(X_T)} \geq \espalt{}{Z^{\qprob}_{T\wedge\delta}\log\left(Z^{\qprob}_{T\wedge\delta}\right)} + \alpha\espalt{}{Z^{\qprob}_{T\wedge\delta} 1_{\delta>T}\phi(X_{T\wedge\delta})}.
\end{equation*}
Thus, any optimizer must lie in the class where $\espalt{}{Z^{\qprob}_T\log\left(Z^{\qprob}_T\right)} = \espalt{}{Z^{\qprob}_{T\wedge\delta}\log\left(Z^{\qprob}_{T\wedge\delta}\right)}$,  which implies that almost surely $Z^{\qprob}_T\log\left(Z^{\qprob}_T\right) = Z^{\qprob}_{T\wedge\delta}\log\left(Z^{\qprob}_{T\wedge\delta}\right)$.
Hence, by uniqueness of the optimizer it must hold that the associated $\tilde{A},\tilde{B},\tilde{C}$ are of the form
\begin{equation}
\tilde{A}_\cdot 1_{\cdot\leq T\wedge\delta}, \ \tilde{B}_\cdot 1_{\cdot\leq T\wedge\delta}, \ \tilde{C}_\cdot 1_{\cdot\leq T\wedge\delta}.
\end{equation}
Now, so far, $\tilde{A},\tilde{B},\tilde{C}$ need only be $\filtg$ predictable.  But, as shown in \cite[Chapter 5]{bielecki2004credit}, due to the specific structure of $\filtg$ we have that $\tilde{A},\tilde{B},\tilde{C}$ coincide with $\filt^{W,W^0}$ predictable process $A,B,C$ on the interval $[t,T\wedge \delta)$. It therefore follows that $\prob\times\textrm{leb}_{[t,T]}$ almost surely
\begin{equation*}
\begin{split}
0 &= 1_{u <  T\wedge\delta}\left((\mu-\gamma)(X_u) + (\sigma\rho)(X_u)'A_u + (\sigma\sqrt{1-\rho'\rho})(X_u)B_u - \gamma(X_u)C_u\right);\\
&= 1_{u \leq  T\wedge\delta}\left((\mu-\gamma)(X_u) + (\sigma\rho)(X_u)'A_u + (\sigma\sqrt{1-\rho'\rho})(X_u)B_u - \gamma(X_u)C_u\right)\\
&\qquad - 1_{u =  T\wedge\delta}\left((\mu-\gamma)(X_u) + (\sigma\rho)(X_u)'A_u + (\sigma\sqrt{1-\rho'\rho})(X_u)B_u - \gamma(X_u)C_u\right).\\
\end{split}
\end{equation*}
Define the $\filt^{W,W^0}$ predictable process
\begin{equation*}
Y_u \dfn \left((\mu-\gamma)(X_u) + (\sigma\rho)(X_u)'A_u + (\sigma\sqrt{1-\rho'\rho})(X_u)B_u - \gamma(X_u)C_u\right);\quad t\leq u\leq T.
\end{equation*}
For any $\eps>0$ we have $\prob\times\textrm{leb}_{[t,T]}\bra{1_{u=T\wedge\delta}Y_u \geq \eps} \leq \prob\times\textrm{leb}_{[t,T]}\bra{u=T\wedge\delta}$. But, for each $u\in [t,T)$ we know $\prob\bra{u=T\wedge\delta} = \espalt{}{\condespalt{}{1_{u=\delta}}{\F^{W,W^0}_\infty}} =0$, since $\delta$ has a density conditional on $\F^{W,W^0}_\infty$.  As $\cbra{u=T}$ has Lebesgue measure $0$ we see that $\prob\times\textrm{leb}_{[t,T]}\bra{1_{u=T\wedge\delta}Y_u \geq \eps} = 0$ for all $\eps>0$.  A similar argument shows $\prob\times\textrm{leb}_{[t,T]}\bra{1_{u=T\wedge\delta}Y_u \leq -\eps} = 0$ and hence $\prob\times\textrm{leb}_{[t,T]}$ almost surely
\begin{equation}\label{E:mkt_px_cond}
\begin{split}
0 &= 1_{u \leq  T\wedge\delta}\left((\mu-\gamma)(X_u) + (\sigma\rho)(X_u)'A_u + (\sigma\sqrt{1-\rho'\rho})(X_u)B_u - \gamma(X_u)C_u\right)
\end{split}
\end{equation}
Now with $Y$ as above assume there is an open interval $(a,b)\subset [t,T]$, a set $A\in\F^{W,W^0}$, and a constant $\eps > 0$ so that $Y_u\geq \eps$ on $A\times(a,b)$ and $\prob\bra{A} > 0$. We then have
\begin{equation*}
\cbra{1_{u\leq T\wedge\delta}, Y_u\geq \eps} \supseteq (A\times (a,b)) \cap (\cbra{\delta>T}\times [t,T]) = (A\cap\cbra{\delta > T})\times (a,b).
\end{equation*}
Clearly, $\textrm{leb}_{[t,T]}(a,b) = (b-a)/(T-t) > 0$.  Additionally, $\prob\bra{A\cap\cbra{T>\delta}} = \espalt{}{1_A e^{-\int_0^T \gamma(X_u)du}} > 0$, where the last inequality follows since $\prob\bra{A}>0$ and $\gamma$ is continuous.  This contradicts \eqref{E:mkt_px_cond}.  A similar argument for $Y\leq -\eps$ shows that $\prob\times\textrm{leb}_{[t,T]}$ almost surely that $Y = 0$, finishing the result.
\end{proof}

We next present a probabilistic counterpart to Proposition \ref{P:G_n_bounds_give_soln}, which yields a candidate dual optimizer as well as an upper bound for the certainty equivalent.

\begin{proposition}\label{P:unwind_prob_ub}
Let $G^n$ be the unique solution to \eqref{E:G_PDE_local} from Propositions \ref{P:local_pde_exist} and \ref{P:opt_invest_local}.  As in Proposition \ref{P:G_n_bounds_give_soln}, assume for each $k\in\nats$ that
\begin{equation}\label{E:G_n_loc_unif_ub_prob}
\sup_{n\geq k+1}\sup_{0\leq t\leq T, x\in \ol{E}_k} G^n(t,x) = C(k) < \infty.
\end{equation}
Let $G$ denote the solution to \eqref{E:G_PDE} from Proposition \ref{P:G_n_bounds_give_soln} and recall the subsequence (still labeled $n$) such that $G^n$ converges to $G$ in $H_{2+\beta,(0,T)\times E,\textrm{loc}}$. Fix $t\in[0,T], x\in E$, and for this $G$, define $\hpi$ as in \eqref{E:opt_pi} and $\hz$ as in \eqref{E:hat_Z}.  Then
\begin{enumerate}[(1)]
\item $\hz$ is the density process of a measure $\hqprob\in\tM$.
\item $\We^{\hpi}$ is a $\hqprob$ \underline{sub}-martingale.
\item For $u(t,x;\phi)$ as in \eqref{E:util_funct}
\begin{equation}\label{E:util_funct_ub}
-\frac{1}{\alpha}\log\left(-u(t,x)\right) \leq G(t,x).
\end{equation}
\end{enumerate}
\end{proposition}

\begin{proof}[Proof of Proposition \ref{P:unwind_prob_ub}]
Take $k,n\geq k+1$ large enough so that $x\in\ol{E}_k$. From Proposition \ref{P:opt_invest_local}, $G^n$ is the certainty equivalent to \eqref{E:util_funct_n}.  Additionally, the process $\hz^n$ from Proposition \ref{P:opt_invest_local} defines a measure $\hqprob^n\in\tM^n$ which solves the dual problem (similar to \eqref{E:dual_util_funct}). Thus,
\begin{equation}\label{E:bdd_rel_ent}
\begin{split}
\espalt{}{\hz^n_{T\wedge\tau^n}\log\left(\hz^n_{T\wedge\tau^n}\right)} &= \alpha G^n(t,x) - \alpha \espalt{}{\hz^n_{T\wedge\tau^n}1_{\delta^n>T\wedge\tau^n}(\chi_n\phi)(X_{T\wedge\tau^n})} \leq \alpha(C(k) -\ul{\phi}),
\end{split}
\end{equation}
where the last equality follows from \eqref{E:G_n_loc_unif_ub_prob}. This shows that $\cbra{\hz^n_{T\wedge\tau^n}}_{n\geq k+1}$ is uniformly integrable.  Next, with $X=X^{t,x}$ we know almost surely that $\hpi^n(s,X_s) \rightarrow \hpi(s,X_s), t\leq s\leq T$. $\hpi$ determines the wealth process $\hwe$ with dynamics
\begin{equation*}
\begin{split}
d\hwe_s &= \hpi(s,X_s)1_{s\leq \delta}\left(\mu(X_s)ds + (\sigma\rho)(X_s)'dW_s + (\sigma\sqrt{1-\rho'\rho})(X_s)dW^0_s\right) - \hpi(s,X_s)dH_s.
\end{split}
\end{equation*}
Recall the optimal wealth process $\hwe^{n}$ from Proposition \ref{P:opt_invest_local} with dynamics
\begin{equation*}
\begin{split}
d\hwe^n_s &= \hpi^n(s,X_s)1_{s\leq \tau^n\wedge\delta^n}\left((\chi_n\mu)(X_s)ds + (\chi_n\sigma\rho)(X_s)'dW_s + (\sqrt{\chi_n}\sigma\sqrt{1-\chi_n\rho'\rho})(X_s)dW^0_s\right)\\
&\qquad -\hpi^n(s,X_s)1_{s\leq\tau^n} dH^n_s.
\end{split}
\end{equation*}
For $u\in [t,T]$ write
\begin{equation*}
\begin{split}
\mathbf{A}^n_u &\dfn \hpi(u,X_u)1_{u\leq \delta}\mu(X_u) - \hpi^n(u,X_u)1_{u\leq \tau^n\wedge\delta^n}(\chi_n\mu)(X_u);\\
\mathbf{B}^n_u &\dfn \hpi(u,X_u)1_{u\leq \delta}(\sigma\rho)(X_u) - \hpi^n(u,X_u)1_{u\leq \tau^n\wedge\delta^n}(\chi_n\sigma\rho)(X_u);\\
\mathbf{C}^n_u &\dfn \hpi(u,X_u)1_{u\leq \delta}(\sigma\sqrt{1-\rho'\rho})(X_u) - \hpi^n(u,X_u)1_{u\leq \tau^n\wedge\delta^n}(\sqrt{\chi_n}\sigma\sqrt{1-\chi_n\rho'\rho})(X_u).
\end{split}
\end{equation*}
Note that
\begin{equation*}
\begin{split}
\int_t^T \left|1_{u\leq \delta} - 1_{u\leq\tau^n\wedge\delta^n}\chi_n(X_u)\right|du &\leq \int_t^T 1_{u\leq\tau^{n-1}}\left|1_{u\leq\delta} - 1_{u\leq \delta^n}\right|\\
&\qquad  + \int_t^T 1_{u>\tau^{n-1}}\left|1_{u\leq\delta} - 1_{u\leq\tau^n\wedge\delta^n}\chi_n(X_u)\right|du;\\
&\leq \delta^n\wedge T - \delta\wedge T + 2\max\cbra{T-\tau^{n-1},0}.
\end{split}
\end{equation*}
The last inequality follows from Lemma \ref{L:delta_n_delta} below which shows $\delta^n\geq\delta$.  Thus, by Lemma \ref{L:delta_n_delta} and Assumption \ref{A:factor} we have almost surely that $\lim_{n\uparrow\infty} \int_t^T \left|1_{u\leq \delta} - 1_{u\leq\tau^n\wedge\delta^n}\chi_n(X_u)\right|du = 0$. From here, the facts that $X$ is continuous, the model coefficients are  continuous functions, and $G^n\rightarrow G$ in $H_{2+\beta,(0,T)\times E,\textrm{loc}}$ imply that almost surely
\begin{equation*}
\lim_{n\uparrow\infty} \int_t^T |\mathbf{A}^n_u|du = 0;\quad \lim_{n\uparrow\infty}\int_t^T |\mathbf{B}^n_u|^2du = 0;\quad \lim_{n\uparrow\infty} \int_t^T |\mathbf{C}^n_u|^2du = 0.
\end{equation*}
This shows that the integrals with respect to $du,dW_u$, and $dW^0_u$ in $\hwe^n$ converge in probability to the respective integrals in $\hwe$ uniformly on $[t,T]$ (see \cite[Proposition 3.2.26]{MR1121940}).  Lastly, set
\begin{equation*}
\begin{split}
\check{M}_s &\dfn \int_t^s \hpi(u,X_u)dH_u = 1_{\delta\leq s}\hpi(\delta,X_{\delta});\qquad \check{M}^n_s \dfn \int_t^s 1_{u\leq\tau^n} \hpi^n(u,X_u)dH^n_u = 1_{\delta^n\leq s\wedge\tau^n}\hpi^n(\delta^n,X_{\delta^n}).
\end{split}
\end{equation*}
Lemma \ref{L:delta_n_delta} implies $1_{\delta^n\leq s\wedge\tau^n} \rightarrow 1_{\delta < s}$ almost surely, and clearly $\hpi^n(\delta^n,X_{\delta^n})\rightarrow \hpi(\delta,X_{\delta})$ almost surely as well.  Thus, we have almost surely $\lim_{n\uparrow\infty} |\check{M}^n_s - M_s| = 1_{\delta = s} \hpi(s,X_s)$. But, as shown in the proof of Lemma \ref{L:dual_opt_structure}, $\prob\bra{\delta=s} = 0$ so that $\check{M}^n_s \rightarrow \check{M}_s$ almost surely for each $s\in [t,T]$.  Putting the above facts together gives $\hwe^n_T \rightarrow \hwe_T$ in Probability. Next, as in \eqref{E:hat_Z} define
\begin{equation}\label{E:hat_Z_a}
\hz_s \dfn  e^{-\alpha (\hwe_s - G(t,x) + 1_{s<\delta} G(s,X_s))};\quad t\leq s\leq T.
\end{equation}
In light of the proceeding we have by construction of $\hwe^n$:
\begin{equation*}
\begin{split}
\hz^n_{T\wedge\tau^n} = \hz^n_T &= e^{-\alpha(\hwe^n_T - G^n(t,x) + 1_{T\wedge\tau^n < \delta^n} (\chi_n\phi)(X_{T\wedge\tau^n}))} \rightarrow e^{-\alpha(\hwe_T - G(t,x) + 1_{T\leq \delta}\phi(X_T))};\\
&= \hz_T e^{\alpha 1_{T=\delta}\phi(X_T)}.
\end{split}
\end{equation*}
where the limit is in Probability. Again, since $\prob\bra{T=\delta} = 0$, we see $\hz^n_{T\wedge\tau^n}\rightarrow \hz_T$ in Probability. We have already shown that $\espalt{}{\hz^n_{T\wedge\tau^n}} = 1$ and  $\cbra{\hz^n_{T\wedge\tau^n}}_{n}$ is uniformly integrable.  This shows that $\espalt{}{\hz_T} = 1$.  Furthermore, a lengthly calculation using \ito's formula, exactly mirroring that in Lemma \ref{L:ITO_result} below\footnote{Indeed, the result may be recovered formally by setting $\chi_n\equiv 1, \tau^n\equiv\infty, \delta^n\equiv \delta$ therein.}, shows $\hz$ has dynamics
\begin{equation}\label{E:Z_form}
\begin{split}
\frac{d\hz_s}{\hz_{s-}} &= 1_{s\leq\delta}\left(A_s'dW_s + B_s dW^0_s\right) + C_s dM_s;\\
A_s &= -\alpha \left(\hpi\sigma\rho + a\nabla G\right)(s,X_s);\quad B_s = -\alpha \hpi\sigma\sqrt{1-\rho'\rho}(s,X_s);\quad C_s = \left(e^{\alpha\left(\hpi + G\right)(s,X_s)}-1\right).
\end{split}
\end{equation}
Also, as can be deduced from the calculations leading to \eqref{E:functional_mkt_px_of_risk_n} in Lemma \ref{L:ITO_result}, the facts that $G$ solves the PDE in \eqref{E:G_PDE} and $\hpi$ is as in \eqref{E:opt_pi} prove that $\prob\times\textrm{leb}_{[t,T]}$ almost surely
\begin{equation*}
0 = 1_{u\leq \delta\wedge T}\left((\mu-\gamma)(X_u) + \sigma\rho(X_u)'A_u + \sigma\sqrt{1-\rho'\rho}(X_u)B_u - \gamma(X_u)C_u\right),
\end{equation*}
so that $\hqprob$ defined by $d\hqprob/d\prob|_{\G_T} = \hz_T$ is in $\M$.  Fatou's Lemma and $\hz^n_{T\wedge\tau^n}\rightarrow \hz_T$ in Probability imply
\begin{equation*}
\espalt{}{\hz_T\log(\hz_T)} \leq \liminf_{n\uparrow\infty} \espalt{}{\hz^n_{T\wedge\tau^n}\log\left(\hz^n_{T\wedge\tau^n}\right)} \leq \alpha(C(k)-\ul{\phi}).
\end{equation*}
The last inequality follows from \eqref{E:bdd_rel_ent}.  This shows that $\hqprob\in\tM$ and gives statement $(1)$. Continuing, from \eqref{E:hat_Z_a} it follows that
\begin{equation}\label{E:nice_lb_trick}
-\alpha \hwe_s \hz_s + \alpha \hz_s G(t,x) = \hz_s\log(\hz_s) + \alpha 1_{s<\delta} \hz_s G(s,X_s) \geq -\frac{1}{e} + \alpha\ul{\phi},
\end{equation}
where the inequality follows via Proposition \ref{P:G_n_global_min} and $G^n \rightarrow G$.  Since $\hqprob\in\tM$ we see that $-\alpha \hwe_s\hz_s + \alpha \hz_s G(t,x)$ is a local martingale bounded from below, hence super-martingale.  Thus, since $\hz$ is a martingale we see that $\hwe\hz$ is a sub-martingale.  This gives statement $(2)$. The sub-martginality implies $\espalt{}{\hwe_T \hz_T} \geq 0$.  Therefore, using the well-known duality results we obtain from \eqref{E:util_funct}, \eqref{E:nice_lb_trick} that
\begin{equation}\label{E:big_ub}
\begin{split}
-\frac{1}{\alpha}\log\left(-u(t,x;\phi)\right) &=\inf_{\qprob\in\tM}\left(\frac{1}{\alpha}\espalt{}{Z^{\qprob}_T\log\left(Z^{\qprob}_T\right)} + \espalt{}{Z^{\qprob}_T 1_{\delta > T}\phi(X_T)}\right);\\
&\leq \frac{1}{\alpha}\espalt{}{\hz_T\log(\hz_T)} + \espalt{}{\hz_T 1_{\delta > T}\phi(X_T)};\\
&=G(t,x) - \espalt{}{\hz_T\hwe_T};\\
&\leq G(t,x).
\end{split}
\end{equation}
Thus, \eqref{E:util_funct_ub} holds, finishing the result.

\end{proof}

With the above results we are ready to finish the proof of Theorem \ref{T:main_result}.  Here we split the results according to whether  Assumption \ref{A:opt_main_ass_inc} or Assumption \ref{A:opt_main_ass_com} hold. To make the notation shorter, set
\begin{equation}\label{E:ell_n_def}
\begin{split}
\ell(x) &\dfn \left(\frac{\mu-\gamma}{\sigma}\right)(x);\qquad \ell_n(x) \dfn \frac{\sqrt{\chi_n}}{\sqrt{1-\chi_n\rho'\rho}}\left(\frac{\mu-\gamma}{\sigma}\right)(x),
\end{split}
\end{equation}
and note that Assumptions \ref{A:opt_main_ass_inc}, \ref{A:opt_main_ass_com} essentially concern the exponential integrability of $\int_0^T \ell^2(X_u)du$.

\begin{proposition}\label{P:unwind_prob_A}

Let Assumption \ref{A:opt_main_ass_inc} hold.  Then, the conclusions of Theorem \ref{T:main_result} follow.

\end{proposition}

\begin{proof}[Proof of Proposition \ref{P:unwind_prob_A}]

Recall that we have fixed a starting point $(t,x)$ for the optimal investment problems of Sections \ref{SS:opt_invest} and \ref{SS:local}. Furthermore, according to Proposition \ref{P:opt_invest_local} for each $n\in\nats$ there is a unique function $G^n\in H_{2+\beta,\Omega_n}$ solving \eqref{E:G_PDE_local}, which is the certainty equivalent to \eqref{E:util_funct_n}. Thus, by the standard duality results
\begin{equation}\label{E:standard_dual_local}
\begin{split}
G^n(t,x) &= \inf_{\qprob^n\in\tM^n}\left(\frac{1}{\alpha}\espalt{}{Z^{\qprob^n}_{T\wedge\tau^n}\log\left(Z^{\qprob^n}_{T\wedge\tau^n}\right)} + \espalt{}{Z^{\qprob^n}_{T\wedge\tau^n}1_{\delta^n>T\wedge\tau^n}\chi_n\phi(X_T\wedge\tau^n)}\right)\\
&\leq \frac{1}{\alpha} \inf_{\qprob^n\in\tM^n} \espalt{}{Z^{\qprob^n}_{T\wedge\tau^n}\log\left(Z^{\qprob^n}_{T\wedge\tau^n}\right)} + \ol{\phi}.
\end{split}
\end{equation}
Now, define
\begin{equation}\label{E:ol_Z_n_inc}
\ol{Z}^n_s \dfn \EN\left(\int_t^\cdot -\ell_n(X_u)dW^0_u\right)_s,\ t\leq s\leq T.
\end{equation}
If $\ol{Z}^n$ is the density process of a measure $\ol{\qprob}^n$ then, in the notation of Lemma \ref{L:mkt_px_of_risk_n}, we have $A^n_u \equiv 0$, $B^n_u = -\ell_n(X_u)$ and $C^n_u\equiv 0$ for $u\leq T$. Since $B^n$ is $\filt^{W}$ predictable, it is also $\filtg^n$ predictable, and a direct calculation shows that \eqref{E:mkt_px_of_risk_n} is satisfied. Also
\begin{equation*}
\espalt{}{\ol{Z}^n_{T\wedge\tau^n}} = \espalt{}{\EN\left(\int_t^\cdot B^n_u dW^0_u\right)_{T\wedge\tau^n}} = \espalt{}{\EN\left(\int_t^\cdot 1_{u\leq T\wedge\tau^n}B^n_u dW^0_u\right)_T} = 1,
\end{equation*}
where the last equality follows by conditioning on $\F^W_T$ and noting that $1_{\cdot\leq \tau^n B^n_\cdot}$ is $\filt^W$ predictable and $W^0$ is independent of $W$. Thus, $\ol{Z}^n$ is a $\filtg^n$ martingale, and we see from Lemmas \ref{L:BR_Girsanov} and \ref{L:mkt_px_of_risk_n} that $\ol{\qprob}^n\in\M^n$.  In fact, using the independence of $1_{\cdot\leq T\wedge\tau^n} B^n_\cdot$ and $W^0$ we obtain:
\begin{equation*}
\begin{split}
\espalt{}{\ol{Z}^n_{T\wedge\tau^n}\log\left(\ol{Z}^n_{T\wedge\tau^n}\right)} &= \frac{1}{2}\espalt{}{\int_t^{T\wedge\tau^n} \ell_n^2(X_u)du}\leq \frac{1}{2(1-\ol{\rho})}\espalt{}{\int_t^T \ell^2(X_u)du}.
\end{split}
\end{equation*}
Above, the inequality used that $0\leq\chi_n\leq 1$ and $\sup_{E}\rho'\rho  = \ol{\rho} < 1$.  Now, recall that $X = X^{t,x}$.  Using the Markov property and the solution $\prob^x$ to the martingale problem for $L$ on $E$ from Assumption \ref{A:factor} we have
\begin{equation*}
\espalt{}{\int_t^T \ell^2(X^{t,x}_u)du} = \espalt{x}{\int_0^{T-t}\ell^2(X_u)du} \leq \frac{1}{\eps}\espalt{x}{e^{\eps\int_0^T \ell^2(X_u)du}},
\end{equation*}
where the last inequality used $x\leq (1/\eps) e^{\eps x}$ for $x\geq 0$.  Thus, from \eqref{E:standard_dual_local} and Assumption \ref{A:opt_main_ass_inc} we conclude that for each $k\in\nats$, $0\leq t\leq T$, $x\in \ol{E}_k$ and $n\geq k+1$ we have
\begin{equation*}
G^n(t,x) \leq \frac{1}{2(1-\ol{\rho})\eps\alpha} \sup_{x\in\ol{E}_k}\espalt{x}{e^{\eps\int_0^T \ell^2(X_u)du}} + \ol{\phi} = C(\eps,k).
\end{equation*}
As such, the conclusions of Propositions \ref{P:G_n_bounds_give_soln}, \ref{P:unwind_prob_ub} hold.  As in the latter proposition we denote by $\hpi$ the optimal strategy from \eqref{E:opt_pi}, $\hwe = \We^{\hpi}$ the wealth process and $\hz$ the density process from \eqref{E:hat_Z}.  We now turn to proof the opposite inequality in \eqref{E:util_funct_ub}.  To this end, we use Lemma \ref{L:dual_opt_structure}. Namely, let $\qprob\in\tM$ be of the form in Lemma \ref{L:dual_opt_structure}  (from which any optimizer must lie) where $A,B,C$ are $\filt^{W,W^0}$ predictable process satisfying \eqref{E:mkt_px_of_risk}.  Denote by $Z = Z^\qprob$ the resultant density process. Create the $\filtg^n$ predictable processes $A^n,B^n,C^n$ on $[t,T\wedge\tau^n]$ via
\begin{equation*}
\begin{split}
A^n_u &= A_u1_{u\leq \delta^n \wedge T\wedge \tau^{n-1}};\\
B^n_u &= B_u 1_{u\leq \delta^n \wedge T\wedge \tau^{n-1}} - \ell_n(X_u)1_{\delta^n\wedge T\wedge \tau^{n-1} < u \leq \delta^n\wedge T\wedge\tau^n};\\
C^n_u &= C_u1_{u\leq \delta^n \wedge T\wedge \tau^{n-1}}\\
\end{split}
\end{equation*}
A straight-forward calculation shows the market price of risk equation in \eqref{E:mkt_px_of_risk_n} is satisfied. Then, create the $\filtg^n$ local martingale $Z^n$ by
\begin{equation*}
Z^n_\cdot \dfn \EN\left(\int_t^\cdot (A^n_u)'dW_u + \int_t^\cdot B^n_u dW^0_u + \int_{t}^\cdot C^n_u dM^n_u\right)_{\cdot\wedge T\wedge \tau^n}.
\end{equation*}
We now show that $Z^n$ is a martingale.  First, Lemma \ref{L:Z_n_to_Z} below proves the technical facts which essentially hold because $\chi_n = 1$ on $E_{n-1}$:
\begin{equation}\label{E:Z_n_to_Z}
\delta^n\wedge T \wedge \tau^{n-1} = \delta\wedge T\wedge\tau^{n-1};\qquad Z^n_{\delta^n\wedge T\wedge\tau^{n-1}} = Z_{\delta\wedge T\wedge\tau^{n-1}}.
\end{equation}
Next, to simplify notation, define
\begin{equation}\label{E:a_n_b_n_def}
a_n\dfn \delta\wedge T\wedge \tau^{n-1} = \delta^n\wedge T\wedge\tau^{n-1};\qquad b_n\dfn\delta^n\wedge T\wedge \tau^n,
\end{equation}
and note that $a_n\rightarrow \delta\wedge T$, $b_n\rightarrow \delta\wedge T$. Also note that we have defined $a_n$ in terms of $\delta$ but the second equality above comes from Lemma \ref{L:Z_n_to_Z}. With this notation, \eqref{E:Z_n_to_Z} and \eqref{E:ol_Z_n_inc}, we have using iterative conditioning on $\G^n_{a_n}$, and the independence of $W_0$ and $(W,\tau^n,\delta^n)$:
\begin{equation*}\
\begin{split}
\espalt{}{Z^n_{T\wedge\tau^n}} &= \espalt{}{Z^n_{a_n}\frac{\ol{Z}^n_{b_n}}{\ol{Z}^n_{a_n}}} = \espalt{}{Z^n_{a_n}}= \espalt{}{Z_{a_n}} =1.
\end{split}
\end{equation*}
The last equality follows since by hypothesis $Z$ is a martingale (as a density process for $\qprob\in\M$) and the optional sampling theorem. Thus, $Z^n$ is the density of a martingale measure $\qprob^n\in\M^n$.   As for the relative entropy:
\begin{equation}\label{E:Z_n_rel_ent_equality}
\begin{split}
&\espalt{}{Z^n_{T\wedge\tau^n}\log\left(Z^n_{T\wedge\tau^n}\right)}\\
&\quad = \espalt{}{Z^n_{a_n} \log\left(Z^n_{a_n}\right) \frac{\ol{Z}^n_{b_n}}{\ol{Z}^n_{a_n}}} + \espalt{}{Z^n_{a_n}\frac{\ol{Z}^n_{b_n}}{\ol{Z}^n_{a_n}}\log\left(\frac{\ol{Z}^n_{b_n}}{\ol{Z}^n_{a_n}}\right)};\\
&\quad = \espalt{}{Z_{a_n}\log\left(Z_{a_n}\right)} + \espalt{}{Z_{a_n}\frac{1}{2}\int_{a_n}^{b_n} \ell_n^2(X_u)du}.
\end{split}
\end{equation}
As $Z$ is a $\filtg$ martingale, $a_n\leq \delta\wedge T$ and $\espalt{}{Z_T\log(Z_T)} = \espalt{}{Z_{\delta\wedge T}\log\left(Z_{\delta\wedge T}\right)} < \infty$, the first term is bounded from above by $\espalt{}{Z_{\delta \wedge T}\log\left(Z_{\delta\wedge T}\right)}$. For the second term, we use inequality $xy\leq (1/K)e^{Kx} + (1/K)y\log(y)$ for $x\in \reals, y>0, K >0$.  This gives
\begin{equation}\label{E:rel_ent_aaa}
\begin{split}
\espalt{}{Z_{a_n}\frac{1}{2}\int_{a_n}^{b_n} \ell_n^2(X_u)du} & \leq \frac{1}{K}\espalt{}{Z_{a_n}\log\left(Z_{a_n}\right)} + \frac{1}{K}\espalt{}{e^{\frac{K}{2}\int_{a_n}^{b_n} \ell_n^2(X_u)du}};\\
&\leq \frac{1}{K}\espalt{}{Z_{\delta\wedge T}\log\left(Z_{\delta\wedge T}\right)} + \frac{1}{K}\espalt{}{e^{\frac{K}{2}\int_{t}^{T} \ell_n^2(X_u)du}};\\
\end{split}
\end{equation}
From \eqref{E:ell_n_def} and Assumption \ref{A:opt_main_ass_inc}, we obtain for $K=2(1-\ol{\rho})\eps$ that
\begin{equation*}
\begin{split}
\espalt{}{e^{\frac{K}{2}\int_{t}^{T} \ell_n^2(X_u)du}} &\leq \espalt{x}{e^{\eps\int_0^T\ell^2(X_u)du}}.
\end{split}
\end{equation*}
So, using Assumption \ref{A:opt_main_ass_inc} we see that for $t\leq T, x\in \ol{E}_k$
\begin{equation}\label{E:Z_n_rel_ent_ub}
\espalt{}{Z^n_{T\wedge\tau^n}\log\left(Z^n_{T\wedge\tau^n}\right)} \leq \left(1+\frac{1}{2(1-\ol{\rho})\eps}\right)\espalt{}{Z_{\delta\wedge T}\log\left(Z_{\delta\wedge T}\right)} + \frac{1}{2(1-\ol{\rho})\eps}\sup_{x\in\ol{E}_k} \espalt{x}{e^{\eps\int_0^T\ell^2(X_u)du}}.
\end{equation}
This shows that $Z^n\in\tM^n$ and hence from Proposition \ref{P:opt_invest_local} and \eqref{E:standard_dual_local}
\begin{equation}\label{E:G_n_ub_1}
G^n(t,x) \leq \frac{1}{\alpha}\espalt{}{Z^n_{T\wedge\tau^n}\log\left(Z^n_{T\wedge\tau^n}\right)} + \alpha\espalt{}{Z^n_{T\wedge\tau^n}1_{\delta^n>T\wedge\tau^n}(\chi_n\phi)(X_{T\wedge\tau^n})}.
\end{equation}
Continuing, from \eqref{E:Z_n_to_Z} we see $Z^n_{T\wedge\tau^n} = Z_{a_n}\ol{Z}^n_{b_n}/\ol{Z}^n_{a_n}$. Since $|b_n-a_n|\rightarrow 0, \tau^n\uparrow\infty$ it is clear that $Z^n_{T\wedge\tau^n}\rightarrow Z_{\delta\wedge T}$ almost surely.  Furthermore, \eqref{E:Z_n_rel_ent_ub} implies $\cbra{Z^n_{T\wedge\tau^n}}$ is uniformly integrable.  Thus, by the dominated convergence theorem:
\begin{equation*}
\lim_{n\uparrow\infty}\espalt{}{Z^n_{T\wedge\tau^n}1_{\delta^n>T\wedge\tau^n}(\chi_n\phi)(X_{T\wedge\tau^n})} = \espalt{}{Z_{T\wedge\delta}1_{\delta\geq T}\phi(X_{T})} = \espalt{}{Z_{T\wedge\delta}1_{\delta >  T}\phi(X_{T})},
\end{equation*}
where the last equality holds since $\prob\bra{\delta = T} = 0$. As for $\espalt{}{Z^n_{T\wedge\tau^n}\log\left(Z^n_{T\wedge\tau^n}\right)}$, come back to \eqref{E:Z_n_rel_ent_equality}:
\begin{equation*}
\begin{split}
&\espalt{}{Z^n_{T\wedge\tau^n}\log\left(Z^n_{T\wedge\tau^n}\right)}= \espalt{}{Z_{a_n}\log\left(Z_{a_n}\right)} + \espalt{}{Z_{a_n}\frac{1}{2}\int_{a_n}^{b_n} \ell_n^2(X_u)du}.
\end{split}
\end{equation*}
Fatou's lemma and the sub-martingale property of $Z\log(Z)$ imply $\lim_{n\uparrow\infty} \espalt{}{Z_{a_n}\log\left(Z_{a_n}\right)} = \espalt{}{Z_{\delta\wedge T}\log\left(Z_{\delta\wedge T}\right)}$. As for the second term, since $|b_n-a_n|\rightarrow 0$ we have almost surely that $Z_{a_n}\int_{a_n}^{b_n} \ell_n^2(X_u)du \rightarrow 0$. Furthermore, as shown in \eqref{E:rel_ent_aaa} for $K=2(1-\ol{\rho})\eps$, this term is bounded from above by
\begin{equation*}
\frac{1}{2(1-\ol{\rho})\eps}Z_{a_n}\log\left(Z_{a_n}\right) + \frac{1}{2(1-\ol{\rho})\eps}e^{\eps\int_{t}^{T}\ell^2(X_u)du}.
\end{equation*}
But, this term is uniformly integrable since it converges in probability and in $L^1$ as argued above.  Thus $\cbra{Z^n_{T\wedge\tau^n}\log(Z^n_{T\wedge\tau^n})}$ is uniformly integrable and hence we see from \eqref{E:G_n_ub_1} that
\begin{equation*}
G(t,x) = \lim_{n\uparrow\infty} G^n(t,x) \leq \frac{1}{\alpha}\espalt{}{Z_{\delta\wedge T}\log\left(Z_{\delta\wedge T}\right)} + \espalt{}{Z_{\delta\wedge T}1_{\delta>T}\phi(X_T)}.
\end{equation*}
Now, this result holds for all $\qprob\in\tM$ for the form in Lemma \ref{L:dual_opt_structure}, but as argued there-in, the dual problem is obtained by minimizing over this class.  Thus,
\begin{equation}\label{E:big_lb}
G(t,x) \leq \inf_{\qprob\in\tM}\left(\frac{1}{\alpha}\espalt{}{Z^{\qprob}_T\log\left(Z^\qprob_T\right)} +\espalt{}{Z^\qprob_T 1_{\delta > T}\phi(X_T)}\right) = -\frac{1}{\alpha}\log\left(-u(t,x;\phi)\right).
\end{equation}
This, combined with \eqref{E:util_funct_ub} in Proposition \ref{P:unwind_prob_ub} proves $G$ is the certainty equivalent for the optimal investment problem. Now, at this point we have proved from the above and Proposition \ref{P:unwind_prob_ub}
\begin{enumerate}[(1)]
\item That $G$ is the certainty equivalent.
\item The existence of a measure $\hqprob\in\tM$  with density process $\hz$ which is optimal for the dual problem.  Note: this follows from part (1) of Proposition \ref{P:unwind_prob_ub} and \eqref{E:big_lb}, the latter of which shows the inequalities in \eqref{E:big_ub} are equalities.
\item The existence of a trading strategy $\hpi$ so that for the resultant wealth process $\hwe$, $e^{-\alpha\left(\hwe_T + 1_{\delta > T}\phi(X_T)\right)} = e^{-\alpha G(t,x)}\hz_T$, so that $\hwe$ achieves the optimal utility.  Furthermore, since the inequalities in \eqref{E:big_ub} are all equalities, $\espalt{\hqprob}{\hwe_T} = 0$.
\end{enumerate}

The last thing to show is that $\hpi\in\mathcal{A}$, which will follow if $\hwe$ is a $\qprob$ super-martingale for all $\qprob\in\tM$.  This is hard to show directly, so at this point we appeal to the well-known duality results for exponential utility with locally bounded semi-martingales.  Namely, as $\tM\neq \emptyset$ there exists an optimal $\check{\pi}$ to the primal problem (c.f. \cite{MR2489605}), and since $\hqprob$ solves the dual problem, we already know from the necessary and sufficient conditions for optimality that with probability one:
\begin{equation*}
e^{-\alpha\left(\We^{\check{\pi}}_T + 1_{\delta > T}\chi(X_T)\right)} = e^{-\alpha G(t,x)}\hz_T.
\end{equation*}
This shows that $\hwe_T = \We^{\check{\pi}}_T$ almost surely.  Next, from part (2) of Proposition \ref{P:unwind_prob_ub} we know $\hwe$ is a $\hqprob$ sub-martingale.  But, $\espalt{\hqprob}{\hwe} = 0$ along with the sub-martingale property imply that $\hwe$ is a $\hqprob$ martingale. The abstract theory tells us that $\We^{\check{\pi}}$ is also martingale under $\hqprob$.  This gives, for $t\leq s\leq T$ that $\hwe_s = \We^{\check{\pi}}_s$ almost surely, and hence by right continuity they are indistinguishable on $[t,T]$.  Lastly, the abstract theory implies $\We^{\check{\pi}}$ is a $\qprob$ sub-martingale for all $\qprob\in\tM$ and so the same is true for $\hwe$.  Thus, $\hpi\in \mathcal{A}$ and the proof is complete.

\end{proof}

Lastly, we turn to the case of Assumption \ref{A:opt_main_ass_com}.

\begin{proposition}\label{P:unwind_prob_B}

Let Assumption \ref{A:opt_main_ass_com} hold.  Then, the conclusions of Theorem \ref{T:main_result} follow.

\end{proposition}

\begin{proof}[Proof of Proposition \ref{P:unwind_prob_B}]

The proof is similar to that of Proposition \ref{P:unwind_prob_A} and hence we just show the differences, appealing to former proof to fill in the steps.   As such, there are two things to show/do:
\begin{enumerate}[(1)]
\item For $0\leq t\leq T, x\in E$, and $n$ large enough, create a measure $\ol{\qprob}^n = \ol{\qprob}^{(t,x),n}\in \tM^n$ with relative entropy on $\G^n_{T\wedge\tau^n}$ which is bounded locally in $x$, uniformly in $n$. This will enable us to invoke Propositions \ref{P:G_n_bounds_give_soln}, \ref{P:unwind_prob_ub}.
\item For $0\leq t\leq T, x\in E$ and $\qprob = \qprob^{t,x}\in \tM$, appropriately adjust $\qprob$ in the random interval $(T\wedge\tau^{n-1},T\wedge\tau^n]$ to create a measure $\qprob^n\in\tM^n$ and then show the relative entropy associated to this adjustment vanishes as $n\uparrow\infty$. This will establish the upper bound \eqref{E:big_lb} for $G$, from which the remaining theorem statements follow exactly as in the last paragraphs of the proof of Proposition \ref{P:unwind_prob_A}.
\end{enumerate}
We first consider $(1)$. By part $(A)$ of Assumption \ref{A:opt_main_ass_com}, for $0\leq t\leq T, x\in E$, there is a unique solution to the SDE $dX^{t,x}_s = \left(b-\ell a\rho\right)(X^{t,x}_s)ds + a(X^{t,x}_s)dW_s, s\geq t$, $X^{t,x}_t = x$, and hence the process
\begin{equation*}
Z_{0,s}\dfn \EN\left(-\int_t^\cdot \ell\rho(X_u)'dW_u\right)_s;\quad t\leq s\leq T,
\end{equation*}
is an $\filt^{W,W^0}$ martingale which defines a measure $\prob_0$ on $\F^{W,W^0}_T$.  Note also that we have since $\ell\rho(x)1_{x\in E_n}$ is bounded that
\begin{equation}\label{E:com_temp_calc}
\begin{split}
\espalt{}{Z_{0,T}\log\left(Z_{0,T}\right)} &\leq \liminf_{n\uparrow\infty} \espalt{}{Z_{0,T\wedge\tau^n}\log\left(Z_{0,T\wedge\tau^n}\right)};\\
\espalt{}{Z_{0,T\wedge\tau^n}\log\left(Z_{0,T\wedge\tau^n}\right)} &= \frac{1}{2}\espalt{\prob_0}{\int_t^{T\wedge\tau^n} \ell^2\rho'\rho(X_u)du} \leq \frac{1}{2}\espalt{\prob_0}{\int_t^T \ell^2(X_u)du} \leq \frac{2}{\eps}\espalt{\prob^x_{0}}{e^{\eps\int_0^T \ell^2(X_u)du}}.
\end{split}
\end{equation}
The last inequality follows by the Markov property of $X$ under $\cbra{\prob^x_0}_{x\in E}$, and the estimate $(1/2)x \leq (2/\eps)e^{\eps x}, x>0$ for $\eps>0$.  Next, let $n$ be large enough so $x\in E_n$, recall the definition of $a_n,b_n$ in \eqref{E:a_n_b_n_def}, and define the $\filtg^n$ predictable processes for $t\leq u\leq T$ via $A^n_u = -1_{u\leq b_n}\ell\rho(X_u)$ and $B^n_u\equiv 0, C^n_u\equiv 0$.  It is clear that the market price of risk equations \eqref{E:mkt_px_of_risk_n} are satisfied and that the resultant density process $\ol{Z}^n_{0}$ satisfies $\ol{Z}^n_{0,s} = Z_{0,s}$ for $t\leq s\leq b_n$.  Since $|A^n_u| \leq |\ell(X_u)|1_{u\leq b_n} \leq C(n)$, by construction of $\filtg^n$ we know $\ol{Z}^n_0$ is a $\filtg^n$ martingale, and hence defines a measure $\ol{\qprob}^n \in \M^n$ by $d\ol{\qprob}^n / d\prob \big|_{\G^n_{T\wedge\tau^n}} =  \ol{Z}^n_{0,T\wedge\tau^n} =  \ol{Z}^n_{0,b_n} = Z_{0,b_n}$.   Furthermore,
\begin{equation*}
\begin{split}
&\espalt{}{\ol{Z}^n_{0,T\wedge\tau^n}\log\left(\ol{Z}^n_{0,T\wedge\tau^n}\right)} = \espalt{}{Z_{0,b_n}\log\left(Z_{0,b_n}\right)};\\
&\quad = \espalt{}{\int_0^T 1_{u\leq T\wedge\tau^n} Z_{0,u\wedge T\wedge \tau^n}\log(Z_{0,u\wedge T\wedge\tau^n})(\chi_n\gamma)(X_{u\wedge T\wedge\tau^n}) e^{-\int_0^{u\wedge T\wedge\tau^n} (\chi_n\gamma)(X_v)dv}du}\\
&\qquad\qquad  + \espalt{}{Z_{0,T\wedge\tau^n}\log\left(Z_{0,T\wedge\tau^n}\right)e^{-\int_0^{T\wedge\tau^n}(\chi_n\gamma)(X_v)dv}};\\
&\quad \leq \espalt{}{Z_{0,T\wedge\tau^n}\log\left(Z_{0,T\wedge\tau^n}\right)};\\
&\quad \leq \frac{2}{\eps}\espalt{\prob^x_{0}}{e^{\eps\int_0^T \ell^2(X_u)du}}.
\end{split}
\end{equation*}
Above, the second equality follows from the conditional density of $\delta^n$.  The first inequality holds since $1_{u\leq T\wedge\tau^n}$ is in $\F^{W,W^0}_{u\wedge T\wedge\tau^n}$ and since \eqref{E:com_temp_calc} implies $Z_0\log(Z_0)$ is a $\filt^{W,W^0}$ sub-martingale. The last inequality also follows from \eqref{E:com_temp_calc}.  Thus, $\ol{\qprob}^n\in\tM^n$, and by \eqref{E:standard_dual_local} and part (A) of Assumption \ref{A:opt_main_ass_com} we conclude that for each $k\in\nats$, $0\leq t\leq T$, $x\in \ol{E}_k$ and $n\geq k+1$
\begin{equation*}
G^n(t,x) \leq  \sup_{x\in\ol{E}_k} \frac{2}{\eps} \espalt{\prob^x_{0}}{e^{\eps\int_0^T \ell^2(X_u)du}} + \ol{\phi} = C(\eps,k).
\end{equation*}
As such, part $(1)$ above holds and we may apply Propositions \ref{P:G_n_bounds_give_soln}, \ref{P:unwind_prob_ub}.  As before, we denote by $\hpi$ the optimal strategy from \eqref{E:opt_pi}, $\hwe = \We^{\hpi}$ the wealth process and $\hz$ the density process from \eqref{E:hat_Z}.  We now turn to $(2)$ above which will yield \eqref{E:big_lb}.  To this end, let $\qprob\in\tM$ be of the form in Lemma \ref{L:dual_opt_structure} where $A,B,C$ are $\filt^{W,W^0}$ predictable process satisfying \eqref{E:mkt_px_of_risk}.  Denote by $Z = Z^\qprob$ the resultant density process. Create the $\filtg^n$ predictable processes $A^n,B^n,C^n$ on $[t,T\wedge\tau^n]$ via $A^n_u = A_u1_{u\leq a_n} - \ell\rho(X_u)1_{a_n < u \leq b_n}$, $B^n_u = B_u 1_{u\leq a_n}$ and $C^n_u = C_u1_{u\leq a_n}$. Again, the market price of risk equation \eqref{E:mkt_px_of_risk_n} is satisfied. Then, create the process $Z^n$ by
\begin{equation*}
Z^n_\cdot =\EN\left(\int_t^\cdot (A^n_u)'dW_u + \int_t^\cdot B^n_u dW^0_u + \int_{t}^\cdot C^n_u dM^n_u\right)_{\cdot \wedge T\wedge \tau^n}.
\end{equation*}
To show $Z^n$ is a martingale, we use Lemma \ref{L:Z_n_to_Z}, iterative conditioning on $\G^n_{a_n}$, that $\ell\rho (x)1_{x\in E_n}$ is bounded, and that $Z$ is a $\filtg$ martingale (since $\qprob\in \tM$) to deduce
\begin{equation*}
\begin{split}
\espalt{}{Z^n_{T\wedge\tau^n}} &= \espalt{}{Z^n_{a_n}\EN\left(\int_t^\cdot-\ell\rho(X_u)'dW_u\right)_{a_n}^{b_n}}= \espalt{}{Z_{a_n}} =1.
\end{split}
\end{equation*}
Thus, $Z^n$ is the density of a martingale measure $\qprob^n\in\M^n$.  Next, we note that $Z_0$ is in fact a $\filtg^n$ martingale as well, since $W$ is a $\filtg^n$ martingale, $\ell\rho$ is $\filtg^n$ predictable and $\espalt{}{Z_{0,s}} = 1$ for $t\leq s$.  Thus, we can extend $\prob_0$ to $\G_{T\wedge\tau^n}$ via $Z_{0,T\wedge\tau^n}$. As for the relative entropy, again, using that $\ell\rho(x)1_{x\in E_n}$ is bounded and the conditional Bayes' formula:
\begin{equation}\label{E:Z_n_rel_ent_equality_com}
\begin{split}
&\espalt{}{Z^n_{T\wedge\tau^n}\log\left(Z^n_{T\wedge\tau^n}\right)}\\
&\quad = \espalt{}{Z^n_{a_n}\log(Z^n_{a_n}) \frac{Z_{0,b_n}}{Z_{0,a_n}}} + \espalt{}{Z^n_{a_n} \frac{Z_{0,b_n}}{Z_{0,a_n}}\log\left(\frac{Z_{0,b_n}}{Z_{0,a_n}}\right)};\\
&\quad = \espalt{}{Z_{a_n}\log\left(Z_{a_n}\right)} + \espalt{}{Z_{a_n}\condespalt{\prob_0}{\log\left(\frac{Z_{0,b_n}}{Z_{0,a_n}}\right)}{\G^n_{a_n}}};\\
&\quad = \espalt{}{Z_{a_n}\log\left(Z_{a_n}\right)} + \espalt{}{Z_{a_n}\condespalt{\prob_0}{\frac{1}{2}\int_{a_n}^{b_n} \ell^2\rho'\rho(X_u)du}{\G^n_{a_n}}}.\\
\end{split}
\end{equation}
As before, the first term is bounded above by $\espalt{}{Z_{\delta \wedge T}\log\left(Z_{\delta\wedge T}\right)}$. As for the second term, we again use $xy\leq (1/K)e^{Kx} + (1/K)y\log(y)$ for $x\in \reals, y>0, K >0$ which gives
\begin{equation*}
\begin{split}
\espalt{}{Z_{a_n}\condespalt{\prob_0}{\frac{1}{2}\int_{a_n}^{b_n} \ell^2\rho'\rho(X_u)du}{\G^n_{a_n}}} &\leq \frac{1}{K}\espalt{}{Z_{a_n}\log\left(Z_{a_n}\right)} + \frac{1}{K}\espalt{}{e^{\frac{K}{2}\condespalt{\prob_0}{\int_{a_n}^{b_n} \ell^2\rho'\rho(X_u)du}{\G^n_{a_n}}}}.
\end{split}
\end{equation*}
The first term is again bounded by $(1/K)\espalt{}{Z_{\delta\wedge T}\log\left(Z_{\delta\wedge T}\right)}$. For the second term:
\begin{equation*}
\begin{split}
\espalt{}{e^{\frac{K}{2}\condespalt{\prob_0}{\int_{a_n}^{b_n} \ell^2\rho'\rho(X_u)du}{\G^n_{a_n}}}} &\leq \espalt{}{\condespalt{\prob_0}{e^{\frac{K}{2}\int_{a_n}^{b_n} \ell^2\rho'\rho(X_u)du}}{\G^n_{a_n}}};\\
&=\espalt{\prob_0}{\check{Z}_{0,a_n} \condespalt{\prob_0}{e^{\frac{K}{2}\int_{a_n}^{b_n} \ell^2\rho'\rho(X_u)du}}{\G^n_{a_n}}},
\end{split}
\end{equation*}
where $\check{Z}_{0,s} = \EN\left(\int_t^\cdot \ell\rho(X_u)'dW^{\prob_0}_u\right)_s = d\prob/d\prob_0\big|_{\G^n_s}$ (cf Lemma \ref{L:BR_Girsanov}). Continuing, we have for the $p>1$ of $(B)$ in Assumption \ref{A:opt_main_ass_com}:
\begin{equation*}
\begin{split}
\espalt{\prob_0}{\check{Z}_{0,a_n} \condespalt{\prob_0}{e^{\frac{K}{2}\int_{a_n}^{b_n} \ell^2\rho'\rho(X_u)du}}{\G^n_{a_n}}} & = \espalt{\prob_0}{\check{Z}_{0,a_n}e^{\frac{K}{2}\int_{a_n}^{b_n} \ell^2\rho'\rho(X_u)du}};\\
&\leq \espalt{\prob_0}{\left(\check{Z}_{0,a_n}\right)^p}^{1/p}\espalt{\prob_0}{e^{\frac{Kp}{2(p-1)}\int_{a_n}^{b_n} \ell^2\rho'\rho(X_u)du}}^{(p-1)/p}.
\end{split}
\end{equation*}
For the right-most term above, we use part (A) of Assumption \ref{A:opt_main_ass_com} and take $K=2\eps(p-1)/p$. Then for $k$ large enough so $x\in\ol{E}_k$ and $n\geq k+1$ we have
\begin{equation*}
\begin{split}
\espalt{\prob_0}{e^{\frac{Kp}{2(p-1)}\int_{a_n}^{b_n} \ell^2\rho'\rho(X_u)du}} &\leq \espalt{\prob^x_0}{e^{\eps\int_0^T \ell^2(X_u)du}} < C(\eps, k).
\end{split}
\end{equation*}
As for the other term, we have by part $(B)$ of Assumption \ref{A:opt_main_ass_com} and the convexity of $y^p, p>1$ that for $x\in E$ and $k$ large enough so that $x\in \ol{E}_k$ and $n\geq k+1$:
\begin{equation*}
\begin{split}
\espalt{\prob_0}{\left(\check{Z}_{0,a_n}\right)^p}&\leq \espalt{\prob_0}{\left(\check{Z}_{0,T}\right)^p}\\
&= \espalt{\prob_0}{\EN\left(\int_t^\cdot p\ell\rho(X_u)'dW^{\prob_0}_u\right)_Te^{\frac{1}{2}p(p-1)\int_t^T \ell^2\rho'\rho(X_u)du}};\\
&= \espalt{\prob_{p}}{e^{\frac{1}{2}p(p-1)\int_t^T \ell^2\rho'\rho(X_u)du}};\\
&\leq \espalt{\prob^x_{p}}{e^{\frac{1}{2}p(p-1)\int_0^T \ell^2(X_u)du}} \leq C(p,k),
\end{split}
\end{equation*}
where the last inequality follows from $(B)$ of Assumption \ref{A:opt_main_ass_com}. Putting all this together, we obtain that for $x\in \ol{E}_k, n\geq k+1$:
\begin{equation}\label{E:Z_n_long_calc}
\espalt{}{Z^n_{T\wedge\tau^n}\log\left(Z^n_{T\wedge \tau_n}\right)} \leq \left(1+\frac{2p}{2(p-1)\eps}\right)\espalt{}{Z_{\delta\wedge T}\log\left(Z_{\delta\wedge T}\right)} + C(\eps,p, k).
\end{equation}
Thus, $\qprob^n\in\tM^n$ and by Proposition \ref{P:opt_invest_local} and \eqref{E:standard_dual_local}
\begin{equation}\label{E:G_n_ub_1_com}
G^n(t,x) \leq \frac{1}{\alpha}\espalt{}{Z^n_{T\wedge\tau^n}\log\left(Z^n_{T\wedge\tau^n}\right)} + \alpha\espalt{}{Z^n_{T\wedge\tau^n}1_{\delta^n>T\wedge\tau^n}(\chi_n\phi)(X_{T\wedge\tau^n})}.
\end{equation}
Continuing, again from \eqref{E:Z_n_to_Z} we can write $Z^n_{T\wedge\tau^n} = Z_{a_n}Z_{0,b_n}/Z_{0,a_n}$ and hence we know $Z^n_{T\wedge\tau^n}\rightarrow Z_{\delta\wedge T}$ almost surely.  Furthermore, \eqref{E:Z_n_long_calc} implies $\cbra{Z^n_{T\wedge\tau^n}}$ is uniformly integrable. Thus,
\begin{equation*}
\lim_{n\uparrow\infty}\espalt{}{Z^n_{T\wedge\tau^n}1_{\delta^n>T\wedge\tau^n}(\chi_n\phi)(X_{T\wedge\tau^n})} = \espalt{}{Z_{T\wedge\delta}1_{\delta \geq  T}\phi(X_{T})} = \espalt{}{Z_{T\wedge\delta}1_{\delta >  T}\phi(X_{T})},
\end{equation*}
As for $\espalt{}{Z^n_{T\wedge\tau^n}\log\left(Z^n_{T\wedge\tau^n}\right)}$, come back to \eqref{E:Z_n_rel_ent_equality_com}:
\begin{equation*}
\begin{split}
&\espalt{}{Z^n_{T\wedge\tau^n}\log\left(Z^n_{T\wedge\tau^n}\right)} = \espalt{}{Z_{a_n}\log\left(Z_{a_n}\right)} + \espalt{}{Z_{a_n}\condespalt{\prob_0}{\frac{1}{2}\int_{a_n}^{b_n} \ell^2\rho'\rho(X_u)du}{\G^n_{a_n}}}.\\
\end{split}
\end{equation*}
As before, it suffices to show that the second term above vanishes as $n\uparrow\infty$. But, for $K>0$
\begin{equation*}
\begin{split}
0 \leq Z_{a_n}\condespalt{\prob_0}{\frac{1}{2}\int_{a_n}^{b_n} \ell^2\rho'\rho(X_u)du}{\G^n_{a_n}} &\leq \frac{1}{K} Z_{a_n}\log\left(Z_{a_n}\right) + \frac{1}{K}e^{\frac{K}{2}\condespalt{\prob_0}{\int_{a_n}^{b_n} \ell^2\rho'\rho(X_u)du}{\G^n_{a_n}}};\\
&\leq \frac{1}{K} Z_{a_n}\log\left(Z_{a_n}\right) + \frac{1}{K} \condespalt{\prob_0}{e^{\frac{K}{2}\int_{a_n}^{b_n} \ell^2\rho'\rho(X_u)du}}{\G^n_{a_n}}.
\end{split}
\end{equation*}
The first term is uniformly integrable.  As for the second term we have for some $\tilde{p}>1$ and the $p>1$ of Assumption \ref{A:opt_main_ass_com} that
\begin{equation*}
\begin{split}
\espalt{}{\left(\condespalt{\prob_0}{e^{\frac{K}{2}\int_{a_n}^{b_n} \ell^2\rho'\rho(X_u)du}}{\G^n_{a_n}}\right)^{\tilde{p}}} & \leq \espalt{}{\condespalt{\prob_0}{e^{\frac{\tilde{p}K}{2}\int_{a_n}^{b_n} \ell^2\rho'\rho(X_u)du}}{\G^n_{a_n}}};\\
&= \espalt{\prob_0}{\check{Z}_{a_n}e^{\frac{\tilde{p}K}{2}\int_{a_n}^{b_n} \ell^2\rho'\rho(X_u)du}};\\
&\leq \espalt{\prob_0}{\left(\check{Z}_{a_n}\right)^p}^{1/p}\espalt{\prob_0}{e^{\frac{\tilde{p}pK}{2(p-1)}\int_{a_n}^{b_n} \ell^2\rho'\rho(X_u)du}}^{(p-1)/p};\\
\end{split}
\end{equation*}
The first term on the right was already shown to be finite.  For the second term, we take $K>0$ small enough so that for any $\tilde{p}>1$ we have $\tilde{p}pK/(2(p-1)) < \eps$.  This shows that $\condespalt{\prob_0}{e^{(K/2)\int_{a_n}^{b_n} \ell^2\rho'\rho(X_u)du}}{\G^n_{a_n}}$ is bounded in $L^{\tilde{p}}(\prob)$ for all $\tilde{p}>1$ and hence uniformly integrable. Thus,
\begin{equation*}
Z_{a_n}\condespalt{\prob_0}{\frac{1}{2}\int_{a_n}^{b_n} \ell^2\rho'\rho(X_u)du}{\G^n_{a_n}} = Z_{a_n}\frac{Z_{0,a_n}}{Z_{0,b_n}}\log\left(\frac{Z_{0,a_n}}{Z_{0,b_n}}\right),
\end{equation*}
is $\prob$ uniformly integrable, and since it converges to $0$ almost surely, it vanishes in expectation. We thus have $\espalt{}{Z^n_{T\wedge\tau^n}\log(Z^n_{T\wedge\tau^n})} \rightarrow \espalt{}{Z_{\delta\wedge T}\log(Z_{\delta\wedge T})}$. This gives
\begin{equation*}
G(t,x) = \lim_{n\uparrow\infty} G^n(t,x) \leq \frac{1}{\alpha}\espalt{}{Z_{\delta\wedge T}\log\left(Z_{\delta\wedge T}\right)} + \espalt{}{Z_{\delta\wedge T}1_{\delta>T}\phi(X_T)}.
\end{equation*}
Thus, we have shown part $(2)$. From this point on, the proof is same as in Proposition \ref{P:unwind_prob_A}, starting with the sentence right above \eqref{E:big_lb}.
\end{proof}

\appendix

\section{Supporting Lemmas}\label{S:lemmas}

%-------------------------------%

\begin{lemma}\label{L:theta_lem}
Let $\theta$ be defined as in \eqref{E:theta_def}. Then
\begin{enumerate}[(1)]
\item For $y>0$, $y\dot{\theta}(y)(1+\theta(y)) = \theta(y)$.
\item $\lim_{y\uparrow\infty} \frac{\theta(y)}{\log(y)} = 1$.
\item For $x\in\reals$ and $y>0$
\begin{equation}\label{E:theta_lem}
\begin{split}
2y + x^2 - \theta(ye^x)^2 - 2\theta(ye^x) &\geq 0;\\
\sqrt{x^2+2y-\theta(ye^x)^2-2\theta(ye^x)} - |x-\theta(ye^x)| &\geq 0.
\end{split}
\end{equation}
For each of the above inequalities, there is equality if and only if $x=y>0$.
\end{enumerate}
\end{lemma}

\begin{proof}[Proof of Lemma \ref{L:theta_lem}]
$\theta$ is clearly smooth, and $(1)$ follows by direct calculation since $ y = \theta(y)e^{\theta(y)}$. As for $(2)$, it is clear that $\theta(y)\rightarrow\infty$ as $y\uparrow\infty$. l'Hospital's rule then implies
\begin{equation*}
\lim_{y\uparrow\infty}\frac{\theta(y)}{\log(y)} = \lim_{y\uparrow\infty} y\dot{\theta}(y) = \lim_{y\uparrow\infty}\frac{\theta(y)}{1+\theta(y)} = 1,
\end{equation*}
where second to last equality follows from $(1)$. Lastly, for $(3)$, set $h(x,y)$ as the left hand side of  the first equation in \eqref{E:theta_lem}. Then
\begin{equation*}
\begin{split}
h_x(x,y) &= 2x - 2ye^{x}\theta(ye^x)\dot{\theta}(ye^x) - 2ye^x \dot{\theta}(ye^x) = 2(x-\theta(ye^x));\\
h_{xx}(x,y) &= 2 - 2ye^x\dot{\theta}(ye^x) = 2 - 2\frac{\theta(ye^x)}{1+\theta(ye^x)} = \frac{2}{1+\theta(ye^x)} > 0.
\end{split}
\end{equation*}
The last two equalities follow from $(1)$. Thus, for a fixed $y$, $h(x,y)$ has a unique minimum when $x = \theta(ye^x)$ but this can only happen when $x = y$ since by construction $y = \theta(ye^y)$.  In this case we have $2y + y^2 - \theta(ye^y)^2 -2\theta(ye^y) = 0$ which gives the result. As for the second inequality, it is equivalent to showing that $2y - \theta(ye^x)^2 - 2\theta(ye^x)(1-x) \geq 0$. Calculation using (1) shows the partial derivative with respect to $x$ of this function is $2\theta(ye^x)(x-\theta(ye^x))/(1+\theta)$. However, $\theta(ye^x) > 0$ and $x-\theta(ye^x)$ is strictly increasing in $x$ with a $0$ only at $x=y$.  Thus, the left-hand side of the second equation in \eqref{E:theta_lem} is uniquely minimized at $x=y>0$ for a value of $0$.

\end{proof}

%--------------------------------------------%

\begin{lemma}\label{L:lieb_a_bounds}

For $\check{a}^n$ defined in \eqref{E:lieb_notation}, there is a constant $C(n)$ so that $z\check{a}^n(x,z,0) \leq C(n)(1+|z|^2)$ for $x\in\ol{E}_n$.
\end{lemma}

\begin{proof}[Proof of Lemma \ref{L:lieb_a_bounds}]

To alleviate notation, we suppress the $x$ function argument, but leave in the $z$ function argument. Also, $C(n)$ is a constant which may change from line to line. At $p=0$ we have
\begin{equation}\label{E:a_at_p0}
\check{a}^n(x,z,0)  = \frac{\sigma^2\chi_n}{2\alpha}\left(\frac{2\gamma}{\sigma^2} + \frac{\mu^2}{\sigma^4} - \theta\left(\frac{\gamma}{\sigma^2}e^{\frac{\mu}{\sigma^2} +\alpha z}\right)^2 - 2\theta\left(\frac{\gamma}{\sigma^2}e^{\frac{\mu}{\sigma^2} +\alpha z }\right)\right).
\end{equation}
When $z<0$ it follows from Lemma \ref{L:theta_lem} that
\begin{equation*}
\check{a}^n(x,z,0) \geq \frac{\sigma^2\chi_n}{2\alpha}\frac{2\gamma}{\sigma^2}\left(1-e^{\alpha z}\right) \geq 0,
\end{equation*}
and hence $z \check{a}^n(x,z,0) \leq 0$.  For $z>0$ it follows by the strict positivity of $\theta$  that
\begin{equation*}
\check{a}^n(x,z,0) \leq \frac{\sigma^2 \chi_n}{2\alpha}\left(\frac{2\gamma}{\sigma^2} + \frac{\mu^2}{\sigma^4}\right) \leq C(n),
\end{equation*}
since $\chi_n$ is supported on $\ol{E}_n$ and $\gamma,\mu,\sigma^2$ are all bounded on $\ol{E}_n$.  Thus, we have that $z \check{a}^n(x,z,0)\leq C(n)|z| \leq C(n)(1+z^2)$.

\end{proof}

%--------------------------------------------------%

\begin{lemma}\label{L:lieb_a_Op}

For the function $\check{a}^n$ from \eqref{E:lieb_notation}, and any interval $[z_1,z_2]$
\begin{equation*}
\limsup_{|p|\uparrow\infty} \sup_{x\in\ol{E}_n, z\in [z_1,z_2]} \frac{|\check{a}^n(x,z,p)|}{|p|^2} < \infty.
\end{equation*}

\end{lemma}

\begin{proof}[Proof of Lemma \ref{L:lieb_a_Op}]

We again suppress the $x$ function argument, but leave in the $z,p$ arguments. Since $b,A, \sigma,\chi_n,\gamma$ and $\mu$ are all bounded on $\ol{E}_n$ we need only consider the term
\begin{equation*}
\frac{2\gamma}{\sigma^2} + \left(\frac{\mu}{\sigma^2} - \frac{\alpha}{\sigma}p'a\rho\right)^2 - \theta\left(\frac{\gamma}{\sigma^2}e^{\frac{\mu}{\sigma^2} +\alpha z - \frac{\alpha}{\sigma}p'a\rho}\right)^2 - 2\theta\left(\frac{\gamma}{\sigma^2}e^{\frac{\mu}{\sigma^2} +\alpha z - \frac{\alpha}{\sigma}p'a\rho}\right)
\end{equation*}
Since $\theta > 0$ this is bounded above by $2\gamma/\sigma^2 + 2\mu^2/\sigma^4 + (2\alpha/\sigma^2)\rho'A\rho p'p$. By Lemma \ref{L:theta_lem} this is bounded below by $(2\gamma/\sigma^2)\left(1-e^{\alpha z}\right)$. The result readily follows.

\end{proof}

%--------------------------------------------------%

\begin{lemma}\label{L:delta_n_delta}
For $\delta$ defined in \eqref{E:delta} and $\delta^n$ defined in \eqref{E:delta_n} we have that $\delta^n\downarrow \delta$ almost surely.
\end{lemma}

\begin{proof}[Proof of Lemma \ref{L:delta_n_delta}]
Fix $n< m$.  Since $\chi_n,\chi_m,\gamma$ are continuous and $\chi_n\leq \chi_m$ we have $-\log(U) = \int_0^{\delta^m}(\chi_m\gamma)(X_u)du \geq \int_0^{\delta^m} (\chi_n\gamma)(X_u)du$, so that $\delta^m\leq \delta^n$ and hence $\ol{\delta} = \lim_{n\uparrow\infty} \delta^n$ exists almost surely. Also, since $\chi_n\leq 1$ we have $-\log(U) = \int_0^{\delta}\gamma(X_u)du \geq \int_0^{\delta} (\chi_n\gamma)(X_u)du$, which gives $\delta\leq \delta^n$ and hence $\delta\leq \ol{\delta}$.  But
\begin{equation*}
\begin{split}
-\log(U) &= \lim_{n\uparrow\infty} \int_0^{\delta^n}(\chi_n\gamma)(X_u)du =\lim_{n\uparrow\infty}\left(\int_0^{\delta^n} \gamma(X_u)du + \int_0^{\delta^n} ((1-\chi_n)\gamma)(X_u)du\right) \geq \int_0^{\ol{\delta}}\gamma(X_u)du.
\end{split}
\end{equation*}
Thus $\int_0^{\ol{\delta}}\gamma(X_u)du \leq \int_0^{\delta} \gamma(X_u)du$ and hence $\ol{\delta} \leq \delta$, which gives the result.
\end{proof}

%----------------------------------------------%

\begin{lemma}\label{L:lieb_A_B_C}
Let $A_k,B_k,C_k$ be as in \eqref{E:lieb_A_B_C}.  For any constant $C(k)>0$ we have that
\begin{equation}\label{E:lieb_A_B_C_inf}
\begin{split}
A^\infty_k &\dfn \limsup_{|p|\uparrow\infty}\sup_{x\in\ol{E}_k, z\in[-C(k),C(k)]} |A_k(x,z,p)| = 1.\\
B^\infty_k &\dfn \limsup_{|p|\uparrow\infty}\sup_{x\in\ol{E}_k, z\in[-C(k),C(k)]} |B_k(x,z,p)| < \infty.\\
C^\infty_k &\dfn \limsup_{|p|\uparrow\infty}\sup_{x\in\ol{E}_k, z\in[-C(k),C(k)]} |C_k(x,z,p)| = 0.
\end{split}
\end{equation}
\end{lemma}

\begin{proof}[Proof of Lemma \ref{L:lieb_A_B_C}]
Since $A^{ij}$ does not depend upon $p$ we have that $\ol{\delta}(p)[A^{ij}](x,z,p) = 0$.  Additionally, we have $\ol{\delta}(p)[\EN](x,z,p) = 2p'A(x)p$ so that $(\ol{\delta}(p)-1)[\EN](x,z,p) = p'A(x)'p = \EN(x,p)$.  This shows that $A_k(x,z,p) = 1$ which trivially gives the result. As for $B_k$ we first have, since $\EN(x,p)$ does not depend on $z$ that
\begin{equation*}
\frac{1}{\EN(x,p)}\delta(p)[\EN](x,z,p) = \frac{1}{p'p\ p'Ap}p'\nabla_x\left(p'A(x)p\right).
\end{equation*}
By the ellipticity and regularity of $A$ on $\ol{E}_k$, this term is on the order of $1/|p|$ so that
\begin{equation*}
\lim_{|p|\uparrow\infty}\sup_{x\in\ol{E}_k}\frac{1}{\EN(x,p)}\delta(p)[\EN](x,z,p) = 0.
\end{equation*}
Next, we must evaluate $(\ol{\delta}(p) -1)[\check{a}^n](x,z,p)$.  By grouping terms according to $p$ in \eqref{E:lieb_notation} we obtain
\begin{equation}\label{E:check_a_n_good}
\begin{split}
\check{a}^n(x,z,p) &= p'\left(b-\frac{\chi_n\mu}{\sigma}\rho'a\right)(x) - \frac{\alpha}{2}p'\left(A-\chi_n a\rho\rho'a'\right)(x)p + \frac{\chi_n}{\alpha}\left(\gamma + \frac{\mu^2}{2\sigma^2}\right)(x)\\
&\qquad - \frac{\chi_n\sigma^2}{2\alpha}(x)\left(\theta(x,z,p)^2 + 2\theta(x,z,p)\right).
\end{split}
\end{equation}
This implies that
\begin{equation}\label{E:check_a_n_p}
\begin{split}
\ol{\delta}(p)[\check{a}^n](x,z,p) = p'\left(b-\frac{\chi_n\mu}{\sigma}\rho'a\right)(x)-\alpha p'\left(A-\chi_n a\rho\rho'a'\right)(x)p + \chi_n\sigma \theta(x,z,p)p'a\rho,
\end{split}
\end{equation}
where the last term follows by Lemma \ref{L:theta_lem} which implies that for any smooth function $f(p)$ and constant $K>0$ that $\nabla_p\left(\theta(Ke^{f(p)})^2 + 2\theta(Ke^{f(p)})\right) = 2\theta(Ke^{f(p)})\nabla_p f(p)$. Here, we apply this to $K = (\gamma(x)/\sigma^2(x))e^{\mu(x)/\sigma^2(x) + \alpha z}$ and $f(p) = -(\alpha/\sigma(x))p'a(x)\rho(x)$.  We therefore have
\begin{equation*}
\begin{split}
&\frac{1}{\EN(x,p)}\left(\ol{\delta}(p)-1\right)[\check{a}^n](x,z,p)\\
&\quad = \frac{2}{p'A(x)p}\left(-\frac{\alpha}{2}p'\left(A-\chi_n a\rho\rho'a'\right)(x)p -\frac{\chi_n}{\alpha}\left(\gamma + \frac{\mu}{2\sigma^2}\right)(x)\right)\\
&\qquad + \frac{2}{p'A(x)p}\left(\chi_n\sigma(x)p'a\rho(x)\theta(x,z,p) + \frac{\chi_n \sigma^2}{2\alpha}(x)\left(\theta(x,z,p)^2-2\theta(x,z,p)\right)\right).
\end{split}
\end{equation*}
By the ellipticity of $A$ and coefficient regularity assumptions, the first term above is on the order of $1$, uniformly on $\ol{E}_n$, as $|p|\uparrow\infty$.  From the definition of $\theta(x,z,p)$ in \eqref{E:lieb_notation} and Lemma \ref{L:theta_lem} we deduce that $\theta(x,z,p)\approx O(|p|)$ uniformly for $x\in\ol{E}_n, z\in [-C(k),C(k)]$ and hence the second term above is also on the order of $1$, uniformly over $x\in\ol{E}_k,z\in[-C(k),C(k)]$. Thus
\begin{equation*}
\begin{split}
&\limsup_{|p|\uparrow\infty}\sup_{x\in\ol{E}_k,z\in [-C(k),C(k)]} |B_k(x,z,p)|\\
&\quad \leq \limsup_{|p|\uparrow\infty}\sup_{x\in\ol{E}_k}\frac{1}{\EN(x,p)}\delta(p)[\EN](x,z,p) + \limsup_{|p|\uparrow\infty}\sup_{x\in\ol{E}_k,z\in [-C(k),C(k)]}\frac{1}{\EN(x,p)}\left(\ol{\delta}(p)-1\right)[\check{a}^n](x,z,p) < \infty.
\end{split}
\end{equation*}
Lastly, we must consider $C_k$. We first evaluate
\begin{equation*}
\delta(p)[A^{ij}](x,z,p) = \frac{p'\nabla_x A^{ij}(x)}{p'p}.
\end{equation*}
Clearly, this is on the order of $1/|p|$ so that
\begin{equation}\label{E:C_k_first_part}
\limsup_{|p|\uparrow\infty}\sup_{x\in\ol{E}_n}\frac{1}{\EN(x,p)}\frac{p'p}{8\lambda_k}\sum_{i,j=1}^d \left(\delta(p)[A^{ij}]{x,z,p}\right)^2 = 0.
\end{equation}
Next, we must evaluate
\begin{equation*}
\delta(p)[\check{a}^n](x,z,p) = \check{a}^n_z(x,z,p) + \frac{p'\nabla_x\check{a}^n(x,z,p)}{p'p}.
\end{equation*}
From \eqref{E:check_a_n_good} and Lemma \ref{L:theta_lem} we see that $\check{a}^n_z(x,z,p) = -\chi_n\sigma^2(x)\theta(x,z,p)$, which by Lemma \ref{L:theta_lem} is on the order of $|p|$.  Next, we have from \eqref{E:check_a_n_good} that
\begin{equation*}
\begin{split}
\nabla_x\check{a}^n(x,z,p) &= \sum_{i=1}^d p_i \nabla_x\left(b-\frac{\chi_n\mu}{\sigma}\rho'a\right)^i(x) - \frac{\alpha}{2}\sum_{i,j=1}^d p_i p_j \nabla_x\left(A-\chi_n a\rho\rho'a'\right)^{ij}(x)\\
&\quad + \nabla_x\left(\frac{\chi_n}{\alpha}\left(\gamma + \frac{\mu^2}{2\sigma^2}\right)\right)(x) - \left(\theta^2(x,z,p)+2\theta(x,z,p)\right)\nabla_x\left(\frac{\chi_n\sigma^2}{2\alpha}\right)(x)\\
&\quad - \frac{\chi_n\sigma^2}{2\alpha}(x)\nabla_x\left(\theta^2(x,z,p)+2\theta(x,z,p)\right).
\end{split}
\end{equation*}
The terms on the first two lines above are on the order of $|p|^2$.  Using \eqref{E:lieb_notation} and Lemma \ref{L:theta_lem} we obtain
\begin{equation*}
\nabla_x\left(\theta^2(x,z,p)+2\theta(x,z,p)\right) = 2\theta(x,z,p)\left(\nabla_x\left(\log\left(\frac{\gamma}{\sigma^2}\right)+\frac{\mu}{\sigma^2}\right)(x) - \alpha\sum_{i=1}^d p_i \nabla_x\left(\frac{a\rho}{\sigma}\right)^{i}(x)\right).
\end{equation*}
Again, using \eqref{E:lieb_notation} and Lemma \ref{L:theta_lem} we see this is on the order of $|p|^2$.  Thus, we see that
\begin{equation*}
- \frac{\chi_n\sigma^2}{2\alpha}(x)\nabla_x\left(\theta^2(x,z,p)+2\theta(x,z,p)\right)
\end{equation*}
is on the order of $|p|^2$ which implies that $\nabla_x\check{a}^n(x,z,p)$ is on the order of $|p|^2$ as well.  Thus, we deduce that
\begin{equation*}
\delta(p)[\check{a}^n](x,z,p) = \check{a}^n_z(x,z,p) + \frac{p'\nabla_x\check{a}^n(x,z,p)}{p'p}
\end{equation*}
is on the order of $|p|$ so that
\begin{equation*}
\limsup_{|p|\uparrow\infty}\sup_{x\in\ol{E}_k,z\in[-C(k),C(k)]} \frac{1}{\EN(x,p)}\delta(p)[\check{a}^n](x,z,p) = 0
\end{equation*}
finishing the proof.
\end{proof}

%--------------------------------------%

\begin{lemma}\label{L:lieb_D}
Let $D_k$ be as in \eqref{E:lieb_D}.  For any constant $C(k)>0$ we have that
\begin{equation}\label{E:lieb_D_inf}
\begin{split}
D^\infty_k &\dfn \limsup_{|p|\uparrow\infty}\sup_{x\in\ol{E}_k, z\in[-C(k),C(k)]} |D_k(x,z,p)| < \infty.\\
\end{split}
\end{equation}
\end{lemma}

\begin{proof}[Proof of Lemma \ref{L:lieb_D}]
Since $|\nabla_p\EN(x,p)| = |A(x)p|$ is clear by the ellipticity of $A$ that
\begin{equation*}
\limsup_{|p|\uparrow\infty}\sup_{x\in\ol{E}_k} \frac{|p|^2\Lambda_k + |p||\nabla_p\EN(x,p)|}{\EN(x,p)} < \infty.
\end{equation*}
As for $(1/\EN(x,p))|p||\nabla_p\check{a}^n|$, from \eqref{E:check_a_n_p} we see that
\begin{equation*}
\nabla_p\check{a}^n(x,z,p) = \left(b-\frac{\chi_n\mu}{\sigma}\rho'a\right)(x)-\alpha \left(A-\chi_n a\rho\rho'a'\right)(x)p + \chi_n\sigma \theta(x,z,p)a\rho,
\end{equation*}
In light of Lemma \ref{L:theta_lem} we see this term is on the order of $|p|$. Therefore
\begin{equation*}
\limsup_{|p|\uparrow\infty}\sup_{x\in\ol{E}_k} \frac{|p||\nabla_p\check{a}^n(x,z,p)|}{\EN(x,p)} < \infty,
\end{equation*}
which finishes the result.
\end{proof}

%------------------------------------------%

\begin{lemma}\label{L:ITO_result}

Let $G^n$ be from Proposition \ref{P:local_pde_exist} and recall the function $\hpi^n$ from \eqref{E:opt_pi_n} and wealth process $\hwe^n$ from \eqref{E:wealth_dynamics_n}. Then, for $\hz^n$ as in \eqref{E:hat_Z_n} it follows that
\begin{equation*}
\begin{split}
\hz^n_s &= 1 -\alpha\int_t^s \hz^n_{u-}1_{u\leq\tau^n\wedge\delta^n}\left(\hpi^n\chi_n\sigma \rho' + (\nabla G^n)'a\right)(u,X_u)dW_u\\
&-\alpha\int_t^s \hz^n_{u-} 1_{u\leq\tau^n\wedge\delta^n}\hpi^n\sqrt{\chi_n}\sigma\sqrt{1-\chi_n\rho'\rho}(u,X_u)dW^0_u\\
&+ \int_t^s \hz^n_{u-}1_{u\leq\tau^n}\left(e^{\alpha\left(\hpi^n + G^n\right)(u,X_u)}-1\right)dM^n_u.
\end{split}
\end{equation*}
Furthermore, on $[t,T]\times E_n$ it follows that
\begin{equation}\label{E:functional_mkt_px_of_risk_n}
\begin{split}
&\chi_n(\mu-\gamma) + \chi_n\sigma\rho'\left(\hpi^n\chi_n\sigma \rho' + (\nabla G^n)'a\right) + \sqrt{\chi_n}\sigma\sqrt{1-\chi_n\rho'\rho}\hpi^n\sqrt{\chi_n}\sigma\sqrt{1-\chi_n\rho'\rho}\\
&\qquad   - \chi_n\gamma\left(e^{\alpha\left(\hpi^n + G^n\right)}-1\right) = 0.
\end{split}
\end{equation}

\end{lemma}

\begin{proof}[Proof of Lemma \ref{L:ITO_result}]
Write $\hz^n_s = e^{-\alpha Y^n_s}$ where $Y^n_s = \hwe^n_s - G^n(t,x) + 1_{s\wedge\tau^n < \delta^n} G(s\wedge\tau^n, X_{s\wedge\tau^n})$.  \ito's formula implies
\begin{equation}\label{E:ITO_1}
\begin{split}
\hz^n_s &=  1 - \alpha\int_t^s \hz^n_{u-}d(Y^n)^c_u + \frac{\alpha^2}{2}\int_t^s \hz^n_{u-}d[Y^n,Y^n]^c_u + \sum_{t<u\leq s}\hz^n_{u-1}\left(\frac{\hz^n_u}{\hz^n_{u-}}-1\right),
\end{split}
\end{equation}
where $(Y^n)^c$ is the continuous part of $Y^n$. From \eqref{E:S_n_dynamics} and the integration by parts formula:
\begin{equation*}
\begin{split}
dY^n_s & = 1_{s\leq\tau^n}1_{s\leq\delta^n}\hpi^n(s,X_s)\left(\left(\chi_n\mu\right)(X_s)ds + \left(\chi_n\sigma\rho\right)(X_s)'dW_s + \left(\sqrt{\chi_n}\sigma\sqrt{1-\chi_n\rho'\rho}\right)(X_s)dW^0_s\right)\\
&\qquad  - 1_{s\leq\tau^n}\hpi^n(s,X_s)dH^n_s + 1_{s\wedge\tau^n \leq \delta^n}1_{s\leq\tau^n}\left(G^n_t + LG^n\right)(s,X_s)ds\\
&\qquad  + 1_{s\wedge\tau^n\leq\delta^n}1_{s\leq\tau^n}\left((\nabla G^n)'a\right)(X_s)dW_s - 1_{s\leq\tau^n}G^n(s,X_s)dH^n_s.
\end{split}
\end{equation*}
Collecting terms, this implies
\begin{equation}\label{E:ITO_2}
\begin{split}
d(Y^n)^{c}_s &=1_{s\leq\tau^n\wedge\delta^n}\left(\hpi^n\chi_n\mu + G^n_t + LG^n\right)(s,X_s)ds\\
&\quad + 1_{s\leq\tau^n\wedge\delta^n}\left(\hpi^n\chi_n\sigma\rho' + (\nabla G^n)'a\right)(s,X_s)dW_s\\
&\quad + 1_{s\leq\tau^n\wedge\delta^n}\left(\hpi^n\sqrt{\chi_n}\sigma\sqrt{1-\chi_n\rho'\rho}\right)(s,X_s)dW^0_s,
\end{split}
\end{equation}
and
\begin{equation}\label{E:ITO_3}
d\left[Y^n,Y^n\right]^c_s = 1_{s\leq\tau^n\wedge\sigma^n}\left((\hpi^n)^2\chi_n\sigma^2 + 2\hpi^n\chi_n\sigma(\nabla G^n)'a\rho + (\nabla G^n)'A\nabla G^n\right)(s,X_s)ds.
\end{equation}
Note that $\hz^n$ will only jump at $u\leq s$ if $u=\delta^n\leq\tau^n$.  In this case $\Delta Y^n_u = -\hpi^n(u,X_u) - G^n(u,X_u)$ so that $\hz^n_u/\hz^n_{u-} = e^{\alpha(\hpi^n(u,X_u) + G^n(u,X_u))}$.  It thus follows that
\begin{equation}\label{E:ITO_4}
\begin{split}
\sum_{t<u\leq s} Z^n_{u-}\left(\frac{Z^n_u}{Z^n_{u-}}-1\right) &= \int_t^s Z^n_{u-}1_{u\leq\tau^n}\left(e^{\alpha\left(\hpi^n + G^n\right)(u,X_u)}-1\right)dH^n_u\\
&= \int_t^s Z^n_{u-}1_{u\leq\tau^n}\left(e^{\alpha\left(\hpi^n + G^n\right)(u,X_u)}-1\right)dM^n_u\\
&\qquad\qquad  + \alpha\int_t^sZ^n_{u-}1_{u\leq\tau^n\wedge\delta^n}\frac{1}{\alpha}\left(e^{\alpha\left(\hpi^n + G^n\right)}-1\right)\chi_n\gamma(u,X_u)du
\end{split}
\end{equation}
Collecting the results in \eqref{E:ITO_2},\eqref{E:ITO_2}, \eqref{E:ITO_3} and using them in \eqref{E:ITO_4} gives, in differential notation,
\begin{equation}\label{E:ITO_5}
\begin{split}
\frac{d\hz^n_u}{\hz^n_{u-}} & = -\alpha 1_{u\leq \tau^n\wedge\delta^n}\mathbf{A}^n_u du\\
&\quad -\alpha 1_{u\leq\tau^n\wedge\delta^n}\left(\left(\hpi^n\chi_n\sigma \rho' + (\nabla G^n)'a\right)(u,X_u)dW_u + \hpi^n\sqrt{\chi_n}\sigma\sqrt{1-\chi_n\rho'\rho}(u,X_u)dW^0_u\right)\\
&\quad + 1_{u\leq \tau^n}\left(e^{\alpha\left(\hpi^n + G^n\right)(u,X_u)}-1\right)dM^n_u.
\end{split}
\end{equation}
where, at $(u,X_u)$
\begin{equation*}
\begin{split}
\mathbf{A}^n_u &= \hpi^n\chi_n\mu + G^n_t + LG^n - \frac{1}{2}\alpha\left((\hpi^n)^2\chi_n\sigma^2 + 2\hpi^n\chi_n\sigma(\nabla G^n)'a\rho + (\nabla G^n)'A\nabla G^n\right)\\
&\qquad -\frac{1}{\alpha}\left(e^{\alpha\left(\hpi^n+G^n\right)}-1\right)\chi_n\gamma.
\end{split}
\end{equation*}
Now, for any $(u,y)\in (t,T)\times E_n$ we have (suppressing function arguments), using that $G^n$ solves \eqref{E:G_PDE_local}:
\begin{equation*}
G^n_t + LG^n - \frac{\alpha}{2}(\nabla G^n)'A\nabla G^n = -\frac{\chi_n\sigma^2}{2\alpha}\left(\frac{2\gamma}{\sigma^2} + \left(\frac{\mu}{\sigma^2} - \frac{\alpha}{\sigma}(\nabla G^n)'a\rho\right)^2 - \theta_{G^n}^2 - 2\theta_{G^n}\right).
\end{equation*}
Plugging this into the above gives
\begin{equation*}
\begin{split}
\mathbf{A}^n_u =&\hpi^n\chi_n\mu - \frac{1}{2}\alpha (\hpi^n)^2\chi_n\sigma^2  - \alpha \hpi^n\chi_n\sigma(\nabla G^n)'a\rho - \frac{1}{\alpha}\left(e^{\alpha\left(\hpi^n + G^n\right)}-1\right)\chi_n\gamma\\
&\quad - \frac{\chi_n\sigma^2}{2\alpha}\left(\frac{2\gamma}{\sigma^2} + \left(\frac{\mu}{\sigma^2} - \frac{\alpha}{\sigma}(\nabla G^n)'a\rho\right)^2 - \theta_{G^n}^2 - 2\theta_{G^n}\right).
\end{split}
\end{equation*}
Note that $\chi_n$ factors out of the right hand side above.  Grouping by $\hpi^n$, the remaining terms are
\begin{equation*}
\begin{split}
&\hpi^n\left(\mu - \alpha\sigma(\nabla G^n)'a\rho\right) - \frac{1}{2}\alpha (\hpi^n)^2\sigma^2   - \frac{1}{\alpha}\left(e^{\alpha\left(\hpi^n + G^n\right)}-1\right)\gamma\\
&\quad - \frac{\sigma^2}{2\alpha}\left(\frac{2\gamma}{\sigma^2} + \left(\frac{\mu}{\sigma^2} - \frac{\alpha}{\sigma}(\nabla G^n)'a\rho\right)^2 - \theta_{G^n}^2 - 2\theta_{G^n}\right).
\end{split}
\end{equation*}
We next plug in for $\hpi^n$ from \eqref{E:opt_pi_n}. This gives
\begin{equation*}
\begin{split}
&\frac{1}{\alpha}\left(\frac{\mu}{\sigma^2} - \frac{\alpha}{\sigma}(\nabla G^n)'a\rho - \theta_{G_n}\right)\left(\mu - \alpha\sigma(\nabla G^n)'a\rho\right) - \frac{1}{2}\alpha\sigma^2 \frac{1}{\alpha^2}\left(\frac{\mu}{\sigma^2} - \frac{\alpha}{\sigma}(\nabla G^n)'a\rho - \theta_{G_n}\right)^2\\
&\quad   - \frac{1}{\alpha}\left(e^{\alpha\left(\frac{1}{\alpha}\left(\frac{\mu}{\sigma^2} -\frac{\alpha}{\sigma}(\nabla G^n)'a\rho - \theta_{G_n}\right) + G^n\right)}-1\right)\gamma\\
&\quad - \frac{\sigma^2}{2\alpha}\left(\frac{2\gamma}{\sigma^2} + \left(\frac{\mu}{\sigma^2} - \frac{\alpha}{\sigma}(\nabla G^n)'a\rho\right)^2 - \theta_{G^n}^2 - 2\theta_{G^n}\right).
\end{split}
\end{equation*}
There are a number of cancellations here.  The remaining terms are
\begin{equation*}
\begin{split}
\frac{\sigma^2}{\alpha}\theta_{G^n} - \frac{\gamma}{\alpha}e^{\frac{\mu}{\sigma^2} - \frac{\alpha}{\sigma}(\nabla G^n)'a\rho + \alpha G^n - \theta_{G_n}}& = \frac{e^{-\theta_{G^n}}\sigma^2}{\alpha}\left(\theta_{G^n}e^{\theta_{G^n}} - \frac{\gamma}{\sigma^2} e^{\frac{\mu}{\sigma^2} - \frac{\alpha}{\sigma}(\nabla G^n)'a\rho + \alpha G^n}\right) = 0,
\end{split}
\end{equation*}
where the last equality follows by the definition of $\theta$.  Therefore, from \eqref{E:ITO_5} the first result follows.  It remains to show \eqref{E:functional_mkt_px_of_risk_n}. To this end, we have
\begin{equation*}
\begin{split}
&\chi_n(\mu-\gamma) -\alpha \chi_n\sigma\rho'\left(\hpi^n\chi_n\sigma \rho + a\nabla G^n\right) -\alpha \sqrt{\chi_n}\sigma\sqrt{1-\chi_n\rho'\rho}\hpi^n\sqrt{\chi_n}\sigma\sqrt{1-\chi_n\rho'\rho}\\
&\qquad\qquad  - \chi_n\gamma \left(e^{\alpha\left(\hpi^n + G^n\right)}-1\right)\\
&\qquad = \chi_n\mu - \alpha\chi_n\sigma(\nabla G^n)'a\rho - \alpha\chi_n\sigma^2\hpi^n - \chi_n \gamma e^{\alpha\left(\hpi^n + G^n\right)}.
\end{split}
\end{equation*}
Note the $\chi_n$ factors out.  The remaining terms are, after plugging in for $\hpi^n$ from \eqref{E:opt_pi_n}
\begin{equation*}
\begin{split}
&\mu - \alpha\sigma(\nabla G^n)'a\rho - \alpha\sigma^2\frac{1}{\alpha}\left(\frac{\mu}{\sigma^2} - \frac{\alpha}{\sigma}(\nabla G^n)'a\rho - \theta_{G_n}\right) - \gamma e^{\alpha\left(\frac{1}{\alpha}\left(\frac{\mu}{\sigma^2} - \frac{\alpha}{\sigma}(\nabla G^n)'a\rho - \theta_{G^n}\right) + G^n\right)}\\
&\qquad = \sigma^2e^{-\theta_{G^n}}\left(\theta_{G^n}e^{\theta_{G^n}} - \frac{\gamma}{\sigma^2}e^{\frac{\mu}{\sigma^2} - \frac{\alpha}{\sigma}(\nabla G^n)'a\rho + \alpha G^n}\right) = 0,
\end{split}
\end{equation*}
where the last equality follows by the definition of $\theta$.  Thus, \eqref{E:functional_mkt_px_of_risk_n} holds.
\end{proof}

%---------------------------------------%

\begin{lemma}\label{L:Z_n_to_Z} Let $t\leq T, x\in E$ be fixed. Let $A,B,C$ be $\filt^{W,W^0}$ predictable process satisfying \eqref{E:mkt_px_of_risk} and  set $Z$ via the right hand side of \eqref{E:Z_Q_structure}. Recall \eqref{E:ell_n_def} and create the $\filtg^n$ predictable processes $A^n,B^n,C^n$ on $[t,T\wedge\tau^n]$ via $A^n_u = A_u1_{u\leq \delta^n \wedge T\wedge \tau^{n-1}}$, $B^n_u = B_u 1_{u\leq \delta^n \wedge T\wedge \tau^{n-1}} - \ell^n(X_u)1_{\delta^n\wedge T\wedge \tau^{n-1} < u \leq \delta^n\wedge T\wedge\tau^n}$ and $C^n_u = C_u1_{u\leq \delta^n \wedge T\wedge \tau^{n-1}}$.
Set
\begin{equation*}
Z^n_s =\EN\left(\int_t^\cdot (A^n_u)'dW_u + \int_t^\cdot B^n_u dW^0_u + \int_{t}^\cdot C^n_u dM^n_u\right)_{s\wedge T\wedge \tau^n}.
\end{equation*}
Then for $n$ large enough so that $x\in E_{n-1}$:
\begin{equation}\label{E:Z_n_to_Z_lem}
\delta^n\wedge T\wedge\tau^{n-1} = \delta\wedge T\wedge\tau^{n-1};\qquad Z^n_{\delta^n\wedge T\wedge\tau^{n-1}} = Z_{\delta\wedge T\wedge\tau^{n-1}}.
\end{equation}
\end{lemma}

\begin{proof}[Proof of Lemma \ref{L:Z_n_to_Z}]

We prove the first equality in \eqref{E:Z_n_to_Z_lem}.  Indeed, if $\delta^n\leq T\wedge\tau^{n-1}$ then
\begin{equation*}
\int_0^{\delta}\gamma(X_u)du = -\log(U) = \int_0^{\delta^n}\chi_n(X_u)\gamma(X_u)du = \int_0^{\delta^n} \gamma(X_u)du,
\end{equation*}
since $\chi_n(x) = 1$ on $E_{n-1}$. This shows that $\delta^n\leq \delta$, but we already know from Lemma \ref{L:delta_n_delta} that $\delta^n\geq \delta$ so in fact they are equal. If $\delta^n > T\wedge\tau^{n-1}$ and $\delta> T\wedge\tau^{n-1}$ then $\delta^n\wedge T\wedge\tau^{n-1} = T\wedge \tau^{n-1} = \delta\wedge T\wedge \tau^{n-1}$.  Lastly, if $\delta^n > T\wedge\tau^{n-1}$ and $\delta \leq  T\wedge\tau^{n-1}$ then
\begin{equation*}
-\log(U) = \int_0^\delta \gamma(X_u)du = \int_0^{\delta}\chi_n(X_u)\gamma(X_u)du,
\end{equation*}
which shows that $\delta^n\leq \delta \leq T\wedge\tau^{n-1}$, a contradiction.  This yields the first result. Next, since $A^n=A$, $B^n=B$ on $[t,\delta^n\wedge T\wedge\tau^{n-1}] = [t,\delta\wedge T\wedge\tau^{n-1}]$ we immediately see that
\begin{equation*}
\EN\left(\int_t^\cdot (A^n_u)'dW_u + \int_t^{\cdot}B^n_udW_u\right)_{\delta^n\wedge T\wedge\tau^{n-1}} = \EN\left(\int_t^\cdot A_u'dW_u + \int_t^\cdot B_udW^0_u\right)_{\delta\wedge T\wedge\tau^{n-1}}.
\end{equation*}
It remains to consider the stochastic exponential for the jump processes.  Here
\begin{equation*}
\EN\left(\int_{t+}^\cdot C^n_u dM^n_u\right)_{\delta^n\wedge T\wedge \tau^{n-1}} = e^{-\int_t^{\delta^n\wedge T\wedge \tau^{n-1}}\chi_n(X_u)\gamma(X_u)du}\left(1_{\delta^n > \delta^n\wedge T\wedge \tau^{n-1}} + (1+C^n_{\delta^n})1_{\delta^n\leq \delta^n\wedge T\wedge\tau^{n-1}}\right).
\end{equation*}
Now, the above argument showed that $\delta^n\leq T\wedge\tau^{n-1}$ implies $\delta = \delta^n\leq T\wedge\tau^{n-1}$ and it is not possible for $\delta^n > T\wedge\tau^{n-1}$ and $\delta\leq T\wedge\tau^{n-1}$ so that in fact $\delta^n>T\wedge\tau^{n-1}$ implies $\delta>T\wedge\tau^{n-1}$. Thus, if $\delta^n \leq T\wedge\tau^{n-1}$ then
\begin{equation*}
\begin{split}
\EN\left(\int_{t+}^\cdot C^n_u dM^n_u\right)_{\delta^n\wedge T\wedge \tau^{n-1}} & = e^{-\int_t^{\delta^n\wedge T\wedge\tau^{n-1}}\chi_n(X_u)\gamma(X_u)du}\left(1+C^n_{\delta^n\wedge T\wedge\tau^{n-1}}\right);\\
&=e^{-\int_t^{\delta\wedge T\wedge\tau^{n-1}}\gamma(X_u)du}\left(1+C_{\delta\wedge T\wedge\tau^{n-1}}\right);\\
&= \EN\left(\int_{t+}^\cdot C_u dM_u\right)_{\delta\wedge T\wedge \tau^{n-1}}.
\end{split}
\end{equation*}
Similarly, if $\delta^n > T\wedge\tau^{n-1}$ then
\begin{equation*}
\begin{split}
\EN\left(\int_{t+}^\cdot C^n_u dM^n_u\right)_{\delta^n\wedge T\wedge \tau^{n-1}} & = e^{-\int_t^{\delta^n\wedge T\wedge\tau^{n-1}}\chi_n(X_u)\gamma(X_u)du};\\
&=e^{-\int_t^{\delta\wedge T\wedge\tau^{n-1}}\gamma(X_u)du}\left(1_{\delta > \delta\wedge T\wedge\tau^{n-1}}\right);\\
&= \EN\left(\int_{t+}^\cdot C_u dM_u\right)_{\delta\wedge T\wedge \tau^{n-1}}.
\end{split}
\end{equation*}
Therefore, \eqref{E:Z_n_to_Z_lem} holds finishing the proof.

\end{proof}

\bibliographystyle{siam}
%\bibliography{rag}
\bibliography{master}
%\bibliography{/home/scottrob/Bibliography/master.bib}

\end{document}